\documentclass{lmcs} %
\pdfoutput=1
\usepackage[utf8]{inputenc}

\usepackage{lastpage}
\lmcsdoi{20}{2}{16}
\lmcsheading{}{\pageref{LastPage}}{}{}%
{Aug.~21,~2020}{Jun.~19,~2024}{}


\usepackage{hyperref}
\usepackage{cleveref}
\usepackage{xypic}
\theoremstyle{plain} %

\allowdisplaybreaks

\newcommand{\renewtheorem}[1]{%
  \expandafter\let\csname #1\endcsname\relax
  \expandafter\let\csname c@#1\endcsname\relax
  \expandafter\let\csname end#1\endcsname\relax
  \newtheorem{#1}%
}

\theoremstyle{plain}

\renewtheorem{thm}{Theorem}[section]
\renewtheorem{cor}[thm]{Corollary}
\renewtheorem{lem}[thm]{Lemma}
\renewtheorem{prop}[thm]{Proposition}
\renewtheorem{asm}[thm]{Assumption}

\theoremstyle{definition}

\renewtheorem{rem}[thm]{Remark}
\renewtheorem{rems}[thm]{Remarks}
\renewtheorem{exa}[thm]{Example}
\renewtheorem{exas}[thm]{Examples}
\renewtheorem{nota}[thm]{Notation}
\renewtheorem{defi}[thm]{Definition}
\renewtheorem{conv}[thm]{Convention}
\renewtheorem{conj}[thm]{Conjecture}
\renewtheorem{prob}[thm]{Problem}
\renewtheorem{oprob}[thm]{Open Problem}
\renewtheorem{algo}[thm]{Algorithm}
\renewtheorem{obs}[thm]{Observation}
\renewtheorem{qu}[thm]{Question}
\renewtheorem{fact}[thm]{Fact}
\renewtheorem{pty}[thm]{Property}

\newenvironment{proposition}[0]{\begin{prop}}{\end{prop}}
\newenvironment{theorem}[0]{\begin{thm}}{\end{thm}}
\newenvironment{definition}[0]{\begin{defi}}{\end{defi}}
\newenvironment{example}[0]{\begin{exa}}{\end{exa}}
\newenvironment{remark}[0]{\begin{rem}}{\end{rem}}
\newenvironment{notation}[0]{\begin{nota}}{\end{nota}}
\newenvironment{corollary}[0]{\begin{cor}}{\end{cor}}

\usepackage{lmcs-macros}

\begin{document}

\title{Transpension: The Right Adjoint to the Pi-type}

\author[A.~Nuyts]{Andreas Nuyts\lmcsorcid{0000-0002-1571-5063}}	%
\thanks{Andreas Nuyts holds a Postdoctoral Fellowship from the Research Foundation - Flanders (FWO; 1247922N), and carried out most of this research holding a PhD Fellowship from the Research Foundation - Flanders (FWO; 1110817N). This research was partially conducted at Vrije Universiteit Brussel and funded by the Research Foundation - Flanders (FWO; G0G0519N). This research is partially funded by the Research Fund KU Leuven.}	%

\address{DistriNet, KU Leuven, Belgium}	%
\email{andreas.nuyts@kuleuven.be, dominique.devriese@kuleuven.be}  %

\author[D. Devriese]{Dominique Devriese\lmcsorcid{0000-0002-3862-6856}}%

\begin{abstract}
\noindent 
Presheaf models of dependent type theory have been successfully applied to model HoTT, parametricity, and directed, guarded and nominal type theory.
There has been considerable interest in internalizing aspects of these presheaf models, either to make the resulting language more expressive, or in order to carry out further reasoning internally, allowing greater abstraction and sometimes automated verification.
While the constructions of presheaf models largely follow a common pattern, approaches towards internalization do not.
Throughout the literature, various internal presheaf operators ($\amaze$, $\Phi/\name{extent}$, $\Psi/\Gel$, $\Glue$, $\Weld$, $\mill$, the strictness axiom and locally fresh names) can be found and little is known about their relative expressiveness.
Moreover, some of these require that variables whose type is a shape (representable presheaf, e.g.\ an interval) be used affinely.

We propose a novel type former, the transpension type, which is right adjoint to universal quantification over a shape.
Its structure resembles a dependent version of the suspension type in HoTT.
We give general typing rules and a presheaf semantics in terms of base category functors dubbed multipliers.
Structural rules for shape variables and certain aspects of the transpension type depend on characteristics of the multiplier.
We demonstrate how the transpension type and the strictness axiom can be combined to implement all and improve some of the aforementioned internalization operators (without formal claim in the case of locally fresh names).
\end{abstract}

\maketitle

\section{Introduction and Related Work} \label{sec:intro}

\subsection{The power of presheaves}
Presheaf semantics \cite{Hofmann97-presheaf-chapter,psh-universes} are an excellent tool for modelling relational preservation properties of (dependent) type theory. They have been applied to
parametricity (which is about preservation of relations) \cite{dtt-parametricity,moulin-param3,reldtt,paramdtt},
univalent type theory (preservation of equivalences) \cite{model-cubical,cubical-unifying,cubical,huber,univalence,orton-phd,orton-pitts-axioms},
directed type theory (preservation of morphisms),
guarded type theory (preservation of the stage of advancement of computation) \cite{clock-cat}
and even combinations thereof \cite{guarded-cubical,cavallo-harper-paramhott-journal,riehl-shulman-dtt,weaver-licata-dua}.%
\footnote{We omit models that are not explicitly structured as presheaf models~\cite{cartesian-cubical-tt,2dtt,north-dirdtt}.}
The presheaf models just cited almost all follow a common pattern: First one chooses a suitable base category $\catW$.
The presheaf category over $\catW$ is automatically a model of dependent type theory with the important basic type formers \cite{Hofmann97-presheaf-chapter} as well as a tower of universes \cite{psh-universes}.
Next, one identifies a suitable notion of fibrancy and replaces or supplements the existing type judgement $\Gamma \sez T \type$ with one that classifies fibrant types:
\begin{description}
	\item[HoTT] For homotopy type theory (HoTT, \cite{hottbook}), one considers Kan fibrant types, i.e.\  presheaves in which edges can be composed and inverted as in an $\infty$-groupoid.
	The precise definition may differ in different treatments.
	\item[Parametricity] For parametric type theory, one considers discrete types \cite{dtt-parametricity,cavallo-harper-paramhott-journal,reldtt,paramdtt}: essentially those that satisfy Reynolds' identity extension property \cite{reynolds} which states that homogeneously related objects are equal. This can be expressed by requiring that any non-dependent function $\IX \to A$ from the relational interval, is constant.
	\item[Directed] In directed type theory, one may want to consider Segal, covariant, discrete and Rezk types \cite{riehl-shulman-dtt} and possibly also Conduch\'e types \cite{conduche,robust}\cite[ex.\ 8.1.27]{nuyts-phd}.
	\item[Guarded] In guarded type theory, one considers clock-irrelevant types \cite{clock-cat}: types $A$ such that any non-dependent function $\clocksym \to A$ from the clock type, is constant.
	
	\item[Nominal] Nominal type theory \cite{nomdtt,freshmltt} can be modelled in the Schanuel topos \cite[\S 6.3]{nominal-sets}.
    This is the subcategory of nullary affine cubical sets (see \cref{ex:multip:affine-cubes} later on) that send pushouts in the base category to pullbacks in $\Set$.
    This ensures that if a cell depending on names $\accol{i, j, k}$ in fact only depends on $\accol{i, j}$ and in fact also only depends on $\accol{i, k}$, then it only depends on $\accol{i}$.
\end{description}
To the extent possible, one subsequently proves that the relevant notions of fibrancy are closed under basic type formers, so that we can restrict to fibrant types and still carry out most of the familiar type-theoretic reasoning and programming.
Special care is required for the universe $\uni{}$: it is generally straightforward to adapt the standard Hofmann-Streicher universe to classify only fibrant types, but the universe of fibrant types is in general not automatically fibrant itself.
\begin{description}
	\item[HoTT] In HoTT, the Hofmann-Streicher universe of Kan types is usually automatically Kan.
	\item[Parametricity] In earlier work on parametricity with Vezzosi \cite{paramdtt,reldtt}, we made the universe of discrete types discrete by modifying its presheaf structure and introduced a parametric modality in order to use that universe.
	In contrast, Atkey et al.\ \cite{dtt-parametricity} and Cavallo and Harper \cite{cavallo-harper-paramhott-journal} simply accept that their universes of discrete types are not discrete.
	\item[Directed] In directed type theory, one could expect, perhaps via a directed univalence result \cite{weaver-licata-dua}, that the universe of covariant types is Segal.
	\item[Guarded] In guarded type theory, Bizjak et al.\ \cite{bgcmb16} let the universe depend on a collection of in-scope clock variables lest the clock-indexed later modality $\rhd : \forall(\kappa : \clocksym).\uni{\Delta} \to \uni{\Delta}$ (where $\kappa \in \Delta$) be non-dependent and therefore constant (not clock-indexed) by clock-irrelevance of $\uni{} \to \uni{}$ \cite{clock-cat}.
\end{description}

\subsection{Internalizing the power of presheaves} \label{sec:internalization}
Purely metatheoretic results about type theory certainly have their value.
Pa\-ra\-me\-tri\-city, for instance, has originated and proven its value as a metatheoretic technique for reasoning about programs.
However, with dependent type theory being not only a programming language but also a logic, it is preferable to formulate results about it within the type system, rather than outside it.
We highlight two particular motivations for doing so: to enlarge the end user's toolbox, and to be able to prove internally that a type is fibrant.

\subsubsection*{Enlarging the end user's toolbox}
One motivation for internalizing metatheorems is to enlarge the toolbox of the end user of the proof assistant.
If this is the only goal, then we can prove the desired results in the model on pen and paper and then internalize them ad hoc with an axiom with or without computation rules.
\begin{description}
	\item[HoTT] Book HoTT \cite{hottbook} simply postulates the univalence axiom without computational behaviour, as justified e.g. by the model of Kan-fibrant simplicial sets \cite{univalence}.
	
	CCHM cubical type theory \cite{cubical} provides the $\Glue$ type, which comes with introduction, elimination, $\beta$- and $\eta$-rules and which turns the univalence axiom into a theorem with computational behaviour. It also contains CCHM-Kan-fibrancy of all types as an axiom, in the form of the CCHM-Kan composition operator, with decreed computational behaviour that is defined by induction on the type.

	\item[Parametricity] Bernardy, Coquand and Moulin \cite{moulin-param3,moulin} (henceforth: BCM) internalize their (unary, but generalizable to $k$-ary) cubical set model of parametricity using two combinators $\Phi$ and $\Psi$ \cite{moulin}, a.k.a.\ $\name{extent}$ and $\Gel$ \cite{cavallo-harper-paramhott-journal}. $\Phi$ internalizes the presheaf structure of the function type, and $\Psi$ that of the universe.

    The combinator $\Phi$ and at first sight also $\Psi$ require that the cubical set model lacks diagonals.
    Indeed, to construct a value over the primitive interval, $\Phi$ and $\Psi$ each take one argument for every endpoint and one argument for the edge as a whole.
    Nested use of these combinators, e.g. to create a square, will take $(k+1)^2$ arguments for $k^2$ vertices, $2k$ sides and $1$ square as a whole but none for specifying the diagonal.
    For this reason, BCM's type system enforces a form of \emph{affine} use of interval variables.
    Similarly, connections as in CCHM \cite{cubical} are ruled out.
    In the current paper, we will see that these requirements are not absolute for $\Psi$: there is apparently a very natural `automatic' way to define the behaviour on diagonals and connections where the $\Psi$-type is not explicitly specified by its arguments.

    In earlier work with Vezzosi \cite{paramdtt}, we have internalized parametricity instead using the $\Glue$ type \cite{cubical} and its dual $\Weld$.
    Later on, we added a primitive $\mill$ \cite{psh-charting-design-space} for swapping $\Weld$ and $\Pi(i : \IX)$.
    These operations are sound in presheaves over any base category where we can multiply with $\IX$ -- including cube categories with diagonals or connections -- and are (therefore) strictly less expressive than $\Phi$ which is not.
    Discreteness of all types was internalized as a non-computing \emph{path degeneracy} axiom.\footnote{It is worth noting that it was not possible to use affine interval variables in the setting of \cite{paramdtt}: The type system features parametric $\Pi$-types which are modelled as ordinary $\Pi$-types with non-discrete domain. Discreteness of the $\Pi$-type can be proven solely from discreteness of the codomain, simply by swapping interval variable and function argument. This is however not possible in the affine setting, where only variables introduced prior to an interval variable are taken to be fresh for that interval variable and the exchange rule with an interval variable only works one way.}

    \item[Directed] Weaver and Licata \cite{weaver-licata-dua} use a bicubical set model to show that directed HoTT \cite{riehl-shulman-dtt} can be soundly extended with a directed univalence \emph{axiom}.

	\item[Guarded] In guarded type theory \cite{clock-cat},
	one axiomatizes L\"ob induction and clock-irrelevance.
	
	\item[Nominal] One version of nominal type theory \cite{freshmltt} provides the locally fresh name abstraction $\nu(i : \IX)$ which can be used anywhere (i.e.\ the goal type remains the same after we abstract over a fresh name). The operation introduces a name but requires a body that is fresh for the name (i.e.\ we do not get to use it). This would be rather useless, were it not that we are allowed to \emph{capture} the fresh name (see \cref{sec:recover}).
\end{description}

\subsubsection*{Internalizing fibrancy proofs}
Another motivation to internalize aspects of presheaf categories, is for building parts of the model inside the type theory, thus abstracting away certain categorical details such as the very definition of presheaves, and for some type systems enabling automatic verification of these constructions.
Given the common pattern in models described in the previous section, it is particularly attractive to try and define fibrancy and prove results about it internally.

In the context of HoTT, Orton and Pitts \cite{orton-phd,orton-pitts-axioms} study CCHM-Kan-fibrancy \cite{cubical} in a type theory extended with a set of axioms, of which all but one serve to characterize the interval and the notion of cofibration. One axiom, \emph{strictness}, provides a type former $\Strict$ for strictifying partial isomorphisms, which exists in every presheaf category.
In order to construct a universe of fibrant types, Licata et al.\ postulate an ``amazing right adjoint'' $\IX \amaze \loch$ to the non-dependent path functor $\IX \to \loch$  \cite{internal-universes,orton-phd}, which indeed exists in presheaves over cartesian base categories if $\IX$ is representable.
Since $\IX \amaze \loch$ and its related axioms are global operations (only applicable to closed terms, unless you want to open Pandora's box as we do in the current paper), they keep everything sound by introducing a judgemental comonadic \emph{global} modality $\flat$.

Orton et al.'s formalization \cite{internal-universes,orton-phd,orton-pitts-axioms} is only what we call \emph{meta-internal}: the argument is internalized to \emph{some} type theory which still only serves as a metatheory of the type system of interest.
Ideally, we would also be able to define and prove fibrancy of types \emph{within} the type theory of interest, which we call \emph{auto-internal}. This has several advantages:
\begin{itemize}
	\item A general approach to auto-internalization of notions of fibrancy saves us from a proliferation of type systems, each with axiomatic internal fibrancy operations with hard-coded computational behaviour that proceeds by case analysis on the construction of the type. Proving fibrancy auto-internally will in general be more typesafe than hard-coding it in a language implementation that is often written in a simply-typed language such as Haskell and OCaml.
	\item Given an auto-internal implementation, we can still pretend that we have a meta-internal situation by restricting ourselves to a subset of the language. But we automatically get a two-level type theory \cite{hts,annenkov-2ltt}, where we have access to non-fibrant types from within. (This does not prove conservativity of two-level type theory over the object system.)
	\item In directed type theory, there are various relevant notions of fibrancy, many of which are not well preserved by basic type formers, so access to non-fibrant types may be a necessity to get any work done at all.
\end{itemize}
Auto-internal treatments exist of discrete types in parametricity \cite{cavallo-harper-paramhott-journal}, and discrete, fibrewise-Segal and Rezk types in directed type theory \cite{riehl-shulman-dtt}, but not yet for covariant, Segal or Kan fibrant types due to the need to consider paths in the context $\IX \to \Gamma$.

\subsection{The transpension type} \label{sec:quantifiers-as-modalities}

What is striking about the previous section is that, while most authors have been able to solve their own problems, a common approach is completely absent. We have encountered $\Phi$ and $\Psi$ \cite{moulin}, the amazing right adjoint $\amaze$ \cite{internal-universes}, $\Glue$ \cite{cubical,paramdtt}, $\Weld$ \cite{paramdtt}, $\mill$ \cite{psh-charting-design-space}, the strictness axiom \cite{orton-pitts-axioms} and locally fresh names \cite{freshmltt}. We have also seen that $\Phi$ and $\Psi$ presently require an affine base category, and that $\amaze$ presently requires the global modality $\flat$.

The goal of the current paper is to develop a smaller collection of internal primitives that impose few restrictions on the choice of base category and allow the internal construction of the aforementioned operators when sound.
To this end, we introduce the \textbf{transpension} type former $\transpshrt{i} : \Ty(\Gamma) \to \Ty(\Gamma, i : \IX)$ which in cartesian settings is right adjoint to $\Pi(i : \IX) : \Ty(\Gamma, i : \IX) \to \Ty(\Gamma)$ and is therefore not a quantifier binding $i$, but a coquantifier that \emph{depends} on it.
This same operation was already considered in topoi by Yetter \cite{yetter}, who named it $\nabla$.
Using the transpension and $\Strict$, we can construct $\Phi$ (when sound), $\Psi$, $\amaze$ and $\Glue$, and heuristically translate a subsystem of the nominal dependent type system FreshMLTT \cite{freshmltt} featuring variable capture and locally fresh names.
Given a type former for certain pushouts, we can also construct $\Weld$ and $\mill$.
The transpension coquantifier $\transplong{u : \IU} : \Ty(\Gamma) \to \Ty(\Gamma, u : \IU)$ is part of a sequence of adjoints $\pairshrt{u} \dashv \wknshrt{u} \dashv \funcshrt{u} \dashv \transpshrt{u}$, preceded by the $\Sigma$-type, weakening and the $\Pi$-type.
Adjointness of the first three is provable from the structural rules of type theory. 
However, it is not immediately clear how to add typing rules for a further adjoint.
Birkedal et al.\ \cite{dra} explain how to add a single modality that has a left adjoint in the semantics.
If we want to have two or more adjoint modalities internally, then we can use a multimodal type system such as \MTT{} \cite{mtt-journal,mtt}.
Each modality in \MTT{} needs a semantic left adjoint, so we can only internalize $\wknshrt{u}$, $\funcshrt{u}$ and $\transpshrt{u}$.
A drawback which we accept (as a challenge for future work), is that $\wknshrt{u}$ and $\funcshrt{u}$ become modalities which are a bit more awkward to deal with than ordinary weakening and $\Pi$-types.

A further complication is that the aforementioned modalities bind or depend on a variable, a phenomenon which is not supported by \MTT{}.
We solve this by grouping shape variables such as $u : \IU$ in a \textbf{shape context} which is not considered part of the type-theoretic context but instead serves as the \emph{mode} of the judgement.
This way, we are also rid of the requirement that all internal operations commute with shape substitution; in fact, the transpension generally does not (\cref{sec:ff:rules:discussion}, \cite{transpension-techreport}).

\subsection{Contributions}
Our central contribution is to reduce the plethora of interal presheaf operators in the literature to only a few operations. 
\begin{itemize}
	\item To this end, we formulate a type system \Msys{} featuring a \textbf{transpension type} $\transplong{u : \IU}$, right adjoint to $\funclong{u : \IU}$, with typing rules built on \emph{extensional} \MTT{} \cite{mtt-journal,mtt}.
    We explain how it is reminiscent of the suspension type from HoTT \cite{hottbook}.
	
	\item More generally, the transpension type can be right adjoint to any quantifier-like operation $\lollilong{u : \IU}$ which need neither respect the exchange rule, nor weakening or contraction. In this setting, we also introduce the \textbf{fresh weakening} coquantifier $\freshlong{u : \IU}$, which is left adjoint to $\lollilong{u : \IU}$ and therefore coincides with weakening $\wknlong{u : \IU}$ in cartesian settings.
	
	\item We provide a categorical semantics for $\transplong{u : \IU}$ in almost any presheaf category $\Psh(\catW)$ over base category $\catW$, for almost any representable object $\IU = \yoneda U$, $U \in \catW$.
    To accommodate non-cartesian variables, our system is not parametrized by a representable object $\IU = \yoneda U$, but by an arbitrary endofunctor $\loch \multip U$ on $\catW$: the \textbf{multiplier}.\footnote{In the technical report \cite{transpension-techreport}, we generalize multipliers beyond \emph{endo}functors.}
		We introduce \textbf{criteria} for characterizing the multiplier
		(\cref{def:multip})
		which we use as requirements for internal type theoretic features.
		We identify a complication for base categories (most notably in guarded and nominal type theory) that are not \emph{objectwise pointable}, and define dimensionally split morphisms (a generalization of split epimorphisms) in order to include those base categories.
		We exhibit relevant multipliers in base categories found in the literature (\cref{sec:examples}).

	\item We show that \textbf{all general presheaf internalization operators} that we are aware of -- viz. $\Phi$/$\name{extent}$
  (when sound),
  $\Psi$/$\Gel$ \cite{moulin,moulin-param3}, the amazing right adjoint $\amaze$ \cite{internal-universes}, $\Glue$ \cite{cubical,paramdtt}, $\Weld$ \cite{paramdtt}, $\mill$ \cite{psh-charting-design-space} and (with no formal claim) locally fresh names -- can be \textbf{recovered} from just the transpension type, the strictness axiom and pushouts along $\snd : \vfi \times A \to A$ where $\vfi : \Prop$ (see \cref{fig:recover} for a dependency graph).
  In the process, some of these operators can be \textbf{improved}: We generalize $\Psi$ from affine-like ($\top$-slice full) to arbitrary multipliers, including cartesian ones and we justify $\IU \amaze \loch$ without a global modality and get $\beta$- and $\eta$-rules for it.
  Moreover, since our system provides an operation $\locknovar{\transpshrt{u}}$ for quantifying contexts, we take a step towards auto-internalizing Orton et al.'s work \cite{internal-universes,orton-phd,orton-pitts-axioms}.
  When $\Phi$ is not sound (e.g. in settings with diagonals or connections), we suggest the internal notion of \textbf{transpensive} types to retain some of its power.
  Finally, a form of higher dimensional pattern matching is enabled by exposing $\lollilong{u : \IU}$ internally as a left adjoint.

	\item In a technical report \cite{transpension-techreport}, we investigate how the modalities introduced in this paper commute with each other, and with prior modalities (i.e.\  those already present before adding the transpension type). We also consider composite multipliers, and natural transformations between modalities (called 2-cells in \MTT{}) arising from natural transformations between multipliers.
\end{itemize}
While \MTT{} \cite{mtt-journal,mtt} satisfies canonicity and even normalization insofar as its modality system (called the \emph{mode theory}) does \cite{mtt-normalization},
we will instantiate \MTT{} on a mode theory for which we presently do not have a computational theory,
and we will extend it with some additional typing rules.
For this reason, we build on extensional \MTT{} \cite{mtt}, and defer canonicity and decidability of type-checking to future work.

\subsection{Overview of the paper}
In \cref{sec:ff}, we study in a simplified setting (a system called \FFsys) how the transpension type resembles the suspension type from HoTT and demonstrate how to put it in action.
In \cref{sec:mtt}, we give a brief overview of the typing rules of \MTT{} and decorate the \MTT{} syntax with \emph{left adjoint reminders}.
In \cref{sec:quantifiers-as-modalities-full}, we define the mode theory of \Msys{}, an instantiation of \MTT{}. This mode theory is extremely general and essentially contains all semantic adjoint pairs of a given domain and codomain as modalities between the corresponding modes. In the next two sections, we highlight a number of interesting modalities:
in \cref{sec:subst} we consider modalities related to shape substitutions, and in \cref{sec:substructural} we define and study multipliers and consider modalities arising from them, including the transpension modality.
In \cref{sec:comparison}, we pseudo-embed \FFsys{} in \Msys{} and generalize some of the results obtained for \FFsys{}.
In \cref{sec:add}, we supplement \Msys{} with a few specialized typing rules.
In \cref{sec:structure}, we investigate the structure of the transpension type in \Msys{}.
In \cref{sec:recover}, we explain how to recover known internal presheaf operators.
We conclude in \cref{sec:discussion}.

\section{First Steps: A Fully Faithful Transpension System (\FFsys{})} \label{sec:ff}
In this section, for purposes of demonstration, we present simplified typing rules for the transpension type which apply in a specific setting (as proven in \cref{sec:comparison}).
Using these, we will already be able to exhibit the transpension type as similar to a dependent version of the suspension type in HoTT \cite{hottbook}, and to prove internally that it is right adjoint to universal quantification.
Moreover, in order to showcase how the transpension type allows us to internalize the presheaf structure of other types, we will demonstrate a technique which we call higher-dimensional pattern matching and which has already been demonstrated by Pitts \cite{nominal-transp} in nominal type theory using locally fresh names \cite{freshmltt}.

\begin{figure}
	\small
	\vspace*{-2em}
	\figtitle{Linear/affine shape variables:}
	\begin{equation*}
		\inferencel{ff:ctx-shp}{
			\Gamma \ctx
		}{
			\Gamma, u : \IU \ctx
		}{}
		\qquad
		\inferencel{ff:ctx-shp:fmap}{
			\sigma : \Gamma \to \Gamma'
		}{
			(\sigma, u/u') : (\Gamma, u : \IU) \to (\Gamma', u' : \IU)
		}{}
		\qquad
		\inferencedeadl{\inferencelabel{ff:ctx-shp:wkn} (optional)}{
			\sigma : \Gamma \to \Gamma'
		}{
			\sigma : (\Gamma, u : \IU) \to \Gamma'
		}{}
	\end{equation*}
	
	\figskip
	
	\figtitle{Linear/affine function type:}
	\begin{equation*}
		\inferencel{ff:forall}{
			\Gamma, u : \IU \sez A \type
		}{
			\Gamma \sez \lollishrt{u} .A \type
		}{}
		\qquad
		\inferencel{ff:forall:intro}{
			\Gamma, u : \IU \sez a : A
		}{
			\Gamma \sez \lambda u.a : \lollishrt{u} .A
		}{}
		\qquad
		\inferencel{ff:forall:elim}{
			\Gamma \sez f : \lollishrt{u} . A \\
			\text{No shape vars in $\Delta$}
		}{
			\Gamma, u : \IU, \delta : \Delta \sez f\,u : A
		}{}
	\end{equation*}
	
	\figskip
	
	\figtitle{Telescope quantification:}
	\vspace*{-2ex}
	\begin{equation*}
		\inferencel{ff:ctx-forall}{
			\Gamma, u : \IU, \delta : \Delta \ctx \\
			\text{No shape vars in $\Delta$}
		}{
			\Gamma, \lollishrt{u} .(\delta : \Delta) \ctx
		}{}
		\quad
		\inferencel{ff:ctx-forall:fmap}{
			(\sigma, u/u', \tau/\delta') : (\Gamma, u : \IU, \delta : \Delta) \to (\Gamma', u' : \IU, \delta' : \Delta')
		}{
			(\sigma, \lambdabar u.\tau/\lambdabar u'.\delta') : (\Gamma, \lollishrt{u} .(\delta : \Delta)) \to (\Gamma', \lollishrt{u'} .(\delta' : \Delta'))
		}{}
	\end{equation*}
	
	\begin{equation*}
		\begin{array}{l c l}
		\inferencelabel{ff:ctx-forall:nil} & \quad &
		\inferencelabel{ff:ctx-forall:fmap:nil} \\
		(\Gamma, \lollishrt{u} .()) = \Gamma &&
		(\sigma, \lambdabar u.()/\lambdabar u'.()) = \sigma
		\end{array}
	\end{equation*}
	
	\figskip
	
	\figtitle{Telescope application}
	\vspace*{-2ex}
	\begin{equation*}
		\inferencel{ff:ctx-app}{
			\Gamma, u : \IU, \delta : \Delta \ctx
		}{
			(v/u, (\lambdabar u.\delta)\,v / \delta)
			 : (\Gamma, \lollishrt{u} .(\delta : \Delta), v : \IU)
			 \to (\Gamma, u : \IU, \delta : \Delta) \\
		}{}
	\end{equation*}
	
	\begin{equation*}
		\begin{array}{l}
		\inferencelabel{ff:ctx-app:nat}: \text{The following diagram commutes:} \\
		\xymatrix{
			(\Gamma, \lollishrt{u} .(\delta : \Delta), v : \IU)
				\ar[rr]^{(v/u, (\lambdabar u.\delta)\,v / \delta)}
				\ar[d]^{(\sigma, \lambdabar u.\tau/\lambdabar u'.\delta', v/v')}
			&&
			(\Gamma, u : \IU, \delta : \Delta)
				\ar[d]^{(\sigma, u/u', \tau/\delta')}
			\\
			(\Gamma', \lollishrt{u'} .(\delta' : \Delta'), v' : \IU)
				\ar[rr]_{(v'/u', (\lambdabar u'.\delta')\,v' / \delta')}
			&&
			(\Gamma', u' : \IU, \delta' : \Delta')
		}
		\end{array}
		\quad
		\begin{array}{l}
			\inferencelabel{ff:ctx-app:nil} \\
			(v/u, (\lambdabar u.())\,v / ()) = (v/u) \\
			\phantom{a} \\
			\inferencelabel{ff:ctx-forall:fmap:ctx-app} \\
			(\lambdabar v.(\lambdabar u.\delta)\,v/\lambdabar u.\delta) = \id_{(\Gamma, \lollishrt u.(\delta : \Delta))}
		\end{array}
	\end{equation*}
	
	\figskip
	
	\figtitle{Transpension type:}
	\vspace*{-2ex}
	\begin{equation*}
		\inferencel{ff:transp}{
			\Gamma, u : \IU, \delta : \Delta \ctx \\
			\Gamma, \lollisym \,u.(\delta : \Delta) \sez A \type
		}{
			\Gamma, u : \IU, \delta : \Delta \sez \transpshrt{u} \codot A \type
		}{}
		\qquad
		\inferencel{ff:transp:intro}{
			\Gamma, \lollisym \,u.(\delta : \Delta) \sez a : A
		}{
			\Gamma, u : \IU, \delta : \Delta \sez \meridshrt{u}\,a : \transpshrt{u} \codot A
		}{}
		\qquad
		\inferencel{ff:transp:elim}{
			\Gamma, u : \IU \sez t : \transpshrt{u} \codot A
		}{
			\Gamma \sez \unmeridsym(u.t) : A
		}{}
	\end{equation*}
	
	\begin{equation*}
		\inferencel{ff:transp:beta}{
			\Gamma \sez a : A
		}{
			\Gamma \sez \unmeridsym(u.\meridshrt{u}\,a) = a : A
		}{}
		\qquad
		\inferencel{ff:transp:eta}{
			\Gamma, u : \IU, \delta : \Delta \sez t : \transplong{u : \IU} \codot A
		}{
			\Gamma, u : \IU, \delta : \Delta \sez t = \\
			\qquad \meridshrt{u}\,(\unmeridsym(v.t[v/u, (\lambdabar u.\delta)\,v/\delta])) : \transpshrt{u} \codot A
		}{}
	\end{equation*}
	
	\begin{equation*}
		\hspace{-2em}
		\inferencel{ff:transp:nat}{
			\Gamma', \lollisym \,u'.(\delta' : \Delta') \sez A \type \\
			(\sigma, u/u', \tau/\delta') : \\
			\qquad (\Gamma, u : \IU, \delta : \Delta) \to (\Gamma', u' : \IU, \delta' : \Delta')
		}{
			\Gamma, u : \IU, \delta : \Delta \sez (\transpshrt{u'} \codot A)[\sigma, u/u', \tau/\delta'] = \\
			\qquad \transpshrt{u} \codot (A[\sigma, \lambdabar u.\tau/\lambdabar u'.\delta']) \type
		}{}
		\qquad
		\inferencel{ff:transp:intro:nat}{
			\Gamma', \lollisym \,u'.(\delta' : \Delta') \sez a : A \\
			(\sigma, u/u', \tau/\delta') : \\
			\qquad (\Gamma, u : \IU, \delta : \Delta) \to (\Gamma', u' : \IU, \delta' : \Delta')
		}{
			\Gamma, u : \IU, \delta : \Delta \sez (\meridshrt{u'} \codot a)[\sigma, u/u', \tau/\delta'] = \\
			\qquad \meridshrt{u} \codot (a[\sigma, \lambdabar u.\tau/\lambdabar u'.\delta'])
			: (\transpshrt{u'} \codot A)[\sigma, u/u', \tau/\delta']
		}{}
		\hspace{-2em}
	\end{equation*}
	\caption{Selection of typing rules for a fully faithful transpension type.}
	\label{fig:ff}
\end{figure}

\subsection{Typing rules} \label{sec:ff:rules}
To do this, we first present, in \cref{fig:ff}, typing rules for the transpension type in a specific setting:
a dependent type system with linear or affine shape variables $u : \IU$,
where shape variable contraction is forbidden.

\subsubsection{Linear/affine shape variables} \label{sec:ff:rules:vars}
Variables to the left of $u$ are understood to be fresh for $u$; variables introduced after $u$ may be substituted with terms depending on $u$.
In particular, we have no contraction $(w/u, w/v) : (w : \IU) \to (u, v : \IU)$, while exchange $(x : A, u : \IU) \to (u : \IU, x : A)$ only works in one direction.
This is enforced by the special substitution rule for shape variables (\ruleref{ff:ctx-shp:fmap}).
Weakening of shape variables is optionally allowed (\ruleref{ff:ctx-shp:wkn}).
The examples in which the type system will be put to use, are agnostic as to whether exchange of shape variables $(u : \IU, v : \IU) \to (v : \IU, u : \IU)$ is possible and models of both situations exist.

\subsubsection{Linear/affine function type} \label{sec:ff:rules:forall}
The system features a linear/affine function type $\lollishrt{u}.A$ over $\IU$, with unsurprising formation and introduction rules (\ruleref{ff:forall}, \ruleref{ff:forall:intro}).
The rule for $f\,u$ (\ruleref{ff:forall:elim}) requires that the function $f$ be fresh for $u$, i.e.\ that $f$ depend only on variables to the left of $u$ \cite{moulin-param3,moulin}. For simplicity, we require that $u$ is the last shape variable in the context.

\subsubsection{Transpension type} \label{sec:ff:rules:transp}
Additionally, the system contains a transpension type $\transpshrt{u} \codot A$ over $\IU$, with more unusual rules.
Similar to the \emph{introduction rule} of multimode type theory (\MTT) (\ruleref{wdra:intro} in \cref{fig:mtt:wdra}), the \emph{meridian} constructor of the transpension type (\ruleref{ff:transp:intro}) works by dependent transposition \cite{dra}\cite[\S 2.1.3]{reldtt-techreport}\cite[\S 5.1.3-5.2]{nuyts-phd}: a term of type $\transpshrt{u} \codot A$ in a given context, is equivalent to a term of type $A$ in the context obtained by applying the left adjoint to $\transpshrt{u}$ -- which is universal quantification over $u$ -- to the context.
However, the situation is a bit more subtle than in \MTT{}, in the sense that the left adjoint $\lollishrt u$ is itself a binder and therefore acts on objects already living in a context. For this reason, we consider the entire situation in a further context $\Gamma$. So we start from a context $\Gamma$, a telescope $\delta : \Delta$ (where $\delta$ denotes the vector of variables in the telescope) in context $\Gamma, u : \IU$ and a term $a : A$ that does not live in telescope $\Delta$ but in the universally quantified telescope $\lollishrt u .(\delta : \Delta)$ \cite{param-app,abstract-atomicity}, which extends $\Gamma$.
The resulting meridian $\meridshrt u \codot a : \transpshrt u \codot A$ then lives in telescope $\Delta$, which extends $\Gamma, u : \IU$.
Thus, we remark that both $\transpshrt{u} \codot A$ and $\meridshrt{u} \codot a$ depend on $u$, whereas $A$ and $a$ do not, so in a way the transpension lifts data to a higher dimension, turning points into $\IU$-cells.

The type formation rule \ruleref{ff:transp} is parallel to the modal type formation rule of \MTT{} (\ruleref{wdra} in \cref{fig:mtt:wdra}), which internalizes a (weak) dependent right adjoint; its premises are such that the introduction rule is well-typed.

The elimination rule \ruleref{ff:transp:elim} is equivalent to the existence of the function
\begin{equation*}
	\lambda f.\unmeridsym(u.f\,u) : (\lollishrt u . \transpshrt u \codot A) \to A,
\end{equation*}
which is essentially the co-unit of the adjunction; this differs from the elimination rule of \MTT{} (\ruleref{wdra:elim} in \cref{fig:mtt:wdra}) which works by pattern-matching, but is parallel to the projection function in \cref{thm:projmod}, as are the $\beta$- and $\eta$-rules which internalize the adjunction laws and to which we get back in \cref{sec:ff:rules:ctx-app}.
The elimination rule takes data again to a lower dimension: it turns a dependent $\IU$-cell in the transpension into a point in $A$.

\subsubsection{Admissibility of telescope rules} \label{sec:ff:rules:admissible}
The typing rules for the transpension type rely on the unusual notions of \textbf{telescope quantification and application}, which remain to be discussed.
Before doing so, we remark that one can take either of two viewpoints w.r.t.\ these rules.
One can take a syntactic viewpoint, viewing each of the typing rules concerned as a formal typing rule, i.e.\ as a constructor of our generalized algebraic syntax \cite{gat,gat-phd,tt-in-tt}.
Alternatively, it is possible to prove metatheoretically that each of these rules is admissible, by defining $\lollishrt u .(\delta : \Delta)$ as a telescope of the same length as $\Delta$, but where each variable's type is universally quantified over $u : \IU$.
This latter view is the one that inspires most of our notations, but we make a point of not violating the former possibility, because that one allows a pseudo-embedding\footnote{Pseudo, because \ruleref{ff:ctx-forall:nil} is only an isomorphism in the general system, but that would undidactically complicate notations in the current section.} of the current specialized system in the main system of this paper (\cref{sec:comparison}).

\subsubsection{Telescope quantification} \label{sec:ff:rules:ctx-forall}

Given a context $\Gamma, u : \IU, \delta : \Delta$ with no shape variables in $\Delta$, the rule \ruleref{ff:ctx-forall} creates a new context $\Gamma, \lollishrt u . (\delta : \Delta)$, which is just $\Gamma$ again if $\Delta$ has zero variables (\ruleref{ff:ctx-forall:nil}).
From the syntactic viewpoint, it would perhaps be cleaner to write something like $[\lollishrt u](\Gamma, u : \IU, \delta : \Delta)$, or more generally $[\lollishrt u]\Theta$ for any context $\Theta$ featuring $u : \IU$ as its last shape variable.
However, the notation we have chosen is \emph{possible} since every such context is of the form $\Theta = \Gamma, u : \IU, \delta : \Delta$ for some context $\Gamma$ and telescope $\Delta$, and moreover it is justified by the admissibility proof as well as in the following sense:
\begin{enumerate}
	\item The variables in $\Gamma$ can be accessed in context $[\lollishrt u](\Gamma, u : \IU, \delta : \Delta)$,
	\item For every variable $y : B$ in $\Delta$, we get a term of type $\lollishrt v.B[\sigma]$ in context $[\lollishrt u](\Gamma, u : \IU, \delta : \Delta)$ (for a suitable substitution $\sigma$).
\end{enumerate}
To see (1), we make use of the functoriality rule \ruleref{ff:ctx-forall:fmap},
which from the syntactic viewpoint we could more cleanly write as $[\lollilong{u/u'}] \rho : [\lollishrt u] \Theta \to [\lollishrt{u'}] \Theta'$ for $\rho : \Theta \to \Theta'$.
The alternative notation in the typing rule is again \emph{possible} since any such $\rho$ is of the form $\rho = (\sigma, u/u', \tau/\delta')$ for some $\sigma : \Gamma \to \Gamma'$ and well-typed vector of terms $\tau$, and justified for similar reasons as above.
By applying functoriality to the weakening substitution $(\Gamma, u : \IU, \delta : \Delta) \to (\Gamma, u : \IU)$, we get a substitution $(\Gamma, \lollishrt u .(\delta : \Delta)) \to (\Gamma, \lollishrt u .()) = \Gamma$, which we can use to ignore $\lollishrt u .(\delta : \Delta)$ altogether and thus get access to the variables in context $\Gamma$.

To see (2), we need telescope application.

The transpension type and the meridian constructor respect substitution (\ruleref{ff:transp:nat}, \ruleref{ff:transp:intro:nat}), and this can only be stated thanks to functoriality (\ruleref{ff:ctx-forall:fmap}).

\subsubsection{Telescope application} \label{sec:ff:rules:ctx-app}
In \cref{sec:ff:rules:transp} above, we noted that the formation and introduction rules of the transpension type are in line with those of the modal type in \MTT{} (\cref{fig:mtt:wdra}) and act by dependent transposition.
In fact, the same is true for the formation and introduction rules of the linear/affine function type, which is a dependent right adjoint to shape variable extension of contexts (\ruleref{ff:ctx-shp}).
In order for the types to be adjoints internally -- which requires that we can define unit and co-unit functions and that the adjunction laws are statable and satisfied -- we need their left adjoint operations to be adjoints, i.e.\ for any context $\Psi$ and any context $\Theta = (\Gamma, u : \IU, \Delta)$, we need substitutions $\Psi \to [\lollishrt{u}] \Theta = (\Gamma, \lollishrt u . (\delta : \Delta))$ to be equivalent to substitutions $(\Psi, u : \IU) \to \Theta = (\Gamma, u : \IU, \delta : \Delta)$ respecting $u$.

One way to ensure this is by providing natural unit and co-unit substitutions.
For the unit, we need substitutions $\Psi \to [\lollishrt u](\Psi, u : \IU) = (\Psi, \lollishrt u.()) = \Psi$, so we can take the identity.
In other words, the unit is given by \ruleref{ff:ctx-forall:nil}, with naturality given by \ruleref{ff:ctx-forall:fmap:nil}.

For the co-unit, we need substitutions $(\Gamma, \lollishrt u.(\delta : \Delta), v : \IU) \to (\Gamma, u : \IU, \Delta)$, which are given by \ruleref{ff:ctx-app} and made natural by \ruleref{ff:ctx-app:nat}.
Again, from the syntactic viewpoint it would be cleaner to write $\appsym_\Theta : ([\lollishrt u]\Theta, v : \IU) \to \Theta$.
However, intuitively, semantically and in the admissibility proof, what it does is applying the `function' $\lambdabar u.\delta : \lollishrt u.\Delta$ to $v : \IU$, which inspires the notation in the typing rule.

Since the unit is the identity, the adjunction laws simply require that whiskering the co-unit with either adjoint also yields the identity.
The fact that $\appsym_{(\Gamma, u : \IU)} = \id_{(\Gamma, u : \IU)}$ is exactly what is asserted by \ruleref{ff:ctx-app:nil}.
The fact that $[\lollilong{v/u}]\appsym_{(\Gamma, u : \IU, \delta : \Delta)} = \id_{(\Gamma, \lollishrt u .(\delta : \Delta))}$ is exactly what is asserted by \ruleref{ff:ctx-forall:fmap:ctx-app}.

With the unit and co-unit for the adjunction $(-, u : \IU) \dashv [\lollishrt u]$ on contexts in place, we can now state the $\beta$- and $\eta$-rules of the transpension type (\ruleref{ff:transp:beta}, \ruleref{ff:transp:eta}) parallel to those for the modal type in \MTT{} (\cref{thm:projmod}).

We now show (2) from above, i.e.\ assuming $\Delta$ lists a variable $y : B$, we seek to derive a term $\Gamma, \lollishrt u .(\delta : \Delta) \sez t : \lollishrt v .B[\sigma]$. This can be done using \ruleref{ff:ctx-app} as follows:
\begin{equation*}
	\inference{
	\inference{
		\Gamma, u : \IU, \delta : \Delta \sez y : B
	}{
		\Gamma, \lollishrt u .(\delta : \Delta), v : \IU \sez y[v/u, (\lambdabar u.\delta)\,v/\delta] : B[v/u, (\lambdabar u.\delta)\,v/\delta]
	}{}
	}{
		\Gamma, \lollishrt u .(\delta : \Delta) \sez \lambda v.(y[v/u, (\lambdabar u.\delta)\,v/\delta]) : \lollishrt v .(B[v/u, (\lambdabar u.\delta)\,v/\delta])
	}{}.
\end{equation*}
\begin{definition} \label{def:lambdabar}
	For any variable $y : B$ in telescope $\Delta$, we define $\Gamma, \lollishrt{u}.(\delta : \Delta) \sez \lambdabar u.y := \lambda v.(y[v/u, (\lambdabar u.\delta)\,v/\delta]) : \lollishrt v .(B[v/u, (\lambdabar u.\delta)\,v/\delta])$.
\end{definition}
\begin{proposition} \label{thm:subst-lambdabar}
	For any variable $y : B$ in telescope $\Delta$
	and any substitution $(\id_\Gamma, u/u, \tau/\delta) : (\Gamma, u : \IU, \delta' : \Delta') \to (\Gamma, u : \IU, \delta : \Delta)$, we have
	\[
		\Gamma \sez (\lambdabar u.y)[\lambdabar u.\tau/\lambdabar u.\delta] = \lambda u.(\tau_y[u/u, (\lambdabar u.\delta')\,u/\delta']) : \lollishrt{u}.B[\tau/\delta],
	\]
	where $\tau_y = y[\tau/\delta]$ is the component of the vector $\tau$ for variable $y$.
\end{proposition}
\begin{proof}
	We have
	\begin{align*}
		(\lambdabar u.y)[u/u, \tau/\delta]
		&= (\lambda u.y[u/u, (\lambdabar u.\delta)\,u/\delta])[\lambdabar u.\tau/\lambdabar u.\delta]
		& \text{(\cref{def:lambdabar})} \qquad \\
		&= \lambda u.(y[u/u, (\lambdabar u.\delta)\,u/\delta][\lambdabar u.\tau/\lambdabar u.\delta, u/u]) \\
		&= \lambda u.(y[u/u, \tau/\delta][u/u, (\lambdabar u.\delta')\,u/\delta'])
		& \text{(\ruleref{ff:ctx-app:nat})} \qquad \\
		&= \lambda u.(\tau_y[u/u, (\lambdabar u.\delta')\,u/\delta']). \tag*{\qedhere}
	\end{align*}
\end{proof}
\begin{corollary} \label{thm:subst-lambdabar-nil}
	For any variable $y : B$ in telescope $\Delta$
	and any substitution $(\id_\Gamma, u/u, \tau/\delta) : (\Gamma, u : \IU) \to (\Gamma, u : \IU, \delta : \Delta)$, we have $\Gamma \sez (\lambdabar u.y)[\lambdabar u.\tau/\lambdabar u.\delta] = \lambda u.\tau_y : \lollishrt{u}.B[\tau/\delta]$, where $\tau_y = y[\tau/\delta]$ is the component of the vector $\tau$ for variable $y$.
\end{corollary}
\begin{proof}
	This follows from \ruleref{ff:ctx-app:nil}.
\end{proof}

\subsubsection{Discussion} \label{sec:ff:rules:discussion}
The type system presented above is less general than the paper's main system \Msys{}.
In \cref{sec:ff:rules:ctx-app}, we saw that the unit of the adjunction on contexts is invertible.
This is equivalent to the left adjoint $(-, u : \IU) : \Ctx \to \Ctx/(u : \IU)$ being fully faithful \cite{nlab:coreflective}, and the requirement on presheaf models to support the typing rules in the current section (with \ruleref{ff:ctx-forall:nil} an isomorphism) is exactly that: the \emph{multiplier} functor interpreting $(-, u : \IU)$ has to be fully faithful w.r.t.\ the slice category over $(u : \IU)$.

By uniqueness of the adjoint, we can also conclude that the co-unit of the adjunction $\lollishrt u \dashv \transpshrt u$ is invertible,\footnote{In fact, this is exactly what the $\beta$- and $\eta$-rules of the transpension type say (when $\Delta$ is empty).} which is equivalent to the right adjoint $\transpshrt u$ being fully faithful \cite{nlab:reflective}, whence the section title.

However, the current typing rules become unusable in a more general setting, as well as in more specific settings where we may start adding operations that we need in important applications.
First, we have no story for substitutions which exist in cubical type systems such as endpoints $(0/i) : \Gamma \to (\Gamma, i : \IX)$ \cite{moulin-param3,model-cubical,cubical} or connections $(j \wedge k/i) : (\Gamma, j, k : \IX) \to (\Gamma, i : \IX)$ \cite{cubical}, as there is no formation rule for $\transpshrt{0} \codot A$ or $\transplong{j \wedge k} \codot A$.
Secondly, in non-fully-faithful generalizations featuring the contraction rule for shape variables, the transpension is not stable under substitution of the shape variables preceding $u$, so in those settings the way we internalized the transpension type here was too na\"ive.%
\footnote{Indeed, write $\wknshrt{u}$ for the operation of cartesian weakening over a shape variable $u : \IU$, which is an example of a substitution involving shape variables. If in general $\wknshrt{u} \circ \transpshrt{v} \cong \transpshrt{v} \circ \wknshrt{u}$, then by uniqueness of the left adjoint we would find that $\funcshrt{v} \circ \pairshrt{u} \cong \pairshrt{u} \circ \funcshrt{v}$. This is clearly false for cartesian shapes such as the interval $\IX$ in HoTT. For more information on how the transpension type commutes with other operations, see the technical report \cite{transpension-techreport}.}
In order to obtain a type system that does not fail in the presence of endpoints, connections or shape variable contraction, in the rest of the paper we will rely on \MTT{}, which we briefly summarize in \cref{sec:mtt}.

\subsection{Poles} \label{sec:ff:poles}
We can still try to get a grasp on $\transpshrt{0} \codot A$ in cubical type systems, however.
In general we have $T[0/i] \cong (\lollishrt{i} . (\idtp \IX i 0) \to T)$.
Assuming $T = \transpshrt{i} \codot A$ and $\name{oneIsNotZero} : (\idtp \IX 1 0) \to \Empty$, the latter type is inhabited by
\begin{equation*}
	\pole_0 := \lambda i . \lambda e.\meridshrt{i}\,\paren{
		\case{ \paren{
			\name{oneIsNotZero}~\paren{(\lambdabar i.e)~1}
		} }{}
	} : \lollishrt{i} . (\idtp \IX i 0) \to \transpshrt{i} \codot A
\end{equation*}
where $\lambdabar i.e$ has type $\lollishrt{i}.(\idtp \IX i 0)$.
Moreover, using the $\eta$-rules for functions and the transpension type and a (provable propositional) $\eta$-rule for $\Empty$, we can show that this is the only element.
Thus we see that the transpension type essentially consists of one meridian $(i : \IX) \to \transpshrt{i} \codot T$ for every $t : T$, and that these meridians are all equal to $\pole_0$ at $i = 0$ and analogously to $\pole_1$ at $i = 1$.
This makes the transpension type quite reminiscent of a dependent version of the suspension type from HoTT \cite{hottbook}, although the quantification of the context in the formation and construction rules is obviously a distinction.

\subsection{Internal transposition} \label{sec:ff:transposition}
We can internally show that the following types are isomorphic:%
\footnote{This statement of internal transposition is not parallel to the general \MTT{} one (\cref{thm:transpose}). The current variation is provable from the general result because the right adjoint $\transpshrt u$ is fully faithful.}
\begin{equation*}
	(\lollilong{u : \IU}.A) \to B
	\quad \cong \quad
	\lollilong{u : \IU}.(A \to \transpshrt{u} \codot B).
\end{equation*}
Indeed, given $f : (\lollilong{u : \IU}.A) \to B$, we can define $g : \lollilong{u : \IU}.(A \to \transpshrt{u} \codot B)$ by
\[
	g\,u\,a = \meridshrt{u}\,(f\,(\lambdabar u.a)).
\]
Conversely, given $g : \lollilong{u : \IU}.(A \to \transpshrt{u} \codot B)$, we can define $f : (\lollilong{u : \IU}.A) \to B$ by
\[
	f\,\hat a = \unmeridsym\paren{
		u.g\,u\,(\hat a\,u)
	}.
\]
These constructions are mutually inverse. Indeed, plugging the definition of $f$ into that of $g$, we find in context $(\Gamma, u : \IU, a : A)$:
\begin{align*}
	\meridshrt u \codot \paren{  \unmeridsym(u.g\,u\,((\lambdabar u.a)\,u))  }
	&= \meridshrt u \codot \paren{  \unmeridsym(u.(g\,u\,a)[u/u, (\lambdabar u.(a))\,u/(a)])  } = g\,u\,a
\end{align*}
using the $\eta$-rule of the transpension type. Conversely, plugging $g$ into $f$, we find in context $(\Gamma, \hat a : \lollishrt u . A)$:
\begin{align*}
	&\mathrel{\phantom{=}} \unmeridsym\paren{  u.\enspace(\meridshrt{u} \codot (f\,(\lambdabar u.a)))[u/u, \hat a\,u/a]\enspace  } \\
	&= \unmeridsym\paren{  u.(\meridshrt{u} \codot (\enspace(f\,(\lambdabar u.a))[\lambdabar u.(\hat a\,u)/\lambdabar u.(a)]\enspace))  }
	& \text{(\ruleref{ff:transp:intro:nat})} \\
	&= \unmeridsym\paren{  u.(\meridshrt{u} \codot (\enspace(f\,(\lambda u.\hat a\,u))\enspace))  }
	& \text{(\cref{thm:subst-lambdabar-nil})} \\
	&= f\,(\lambda u.(\hat a\,u)) = f\,\hat a.
	& \text{(\ruleref{ff:transp:beta})}
\end{align*}

\subsection{Higher-dimensional pattern matching} \label{sec:ff:hdpm}
Now that we know internally that $\lollishrt{u}$ is a left adjoint (with internal right adjoint $\transpshrt{u}$), we can proceed to conclude that it preserves colimits, e.g.\ we can show $i : (\lollishrt{u}.A \uplus B) \cong (\lollishrt{u}.A) \uplus (\lollishrt{u}.B)$.
The map to the left is trivially defined by case analysis.
The map to the right is equivalent by transposition to a function $\lollishrt{u}.(A \uplus B \to \transpshrt{u} \codot ((\lollishrt{v}.A[v/u]) \uplus (\lollishrt{v}.B[v/u])))$.
This is in turn constructed by case analysis from the transpositions of the coproduct's constructors $\inl$ and $\inr$.

By straightforward application of \cref{sec:ff:transposition}, the transpositions of the constructors are:
\begin{align*}
	\lambda u.\lambda a.\meridshrt{u} \codot (\inl\,(\lambdabar u.a))
	&: \lollishrt{u}.\paren{ A \to \transpshrt{u} \codot ((\lollishrt{v}.A[v/u]) \uplus (\lollishrt{v}.B[v/u])) } \\
	\lambda u.\lambda b.\meridshrt{u} \codot (\inr\,(\lambdabar u.b))
	&: \lollishrt{u}.\paren{ B \to \transpshrt{u} \codot ((\lollishrt{v}.A[v/u]) \uplus (\lollishrt{v}.B[v/u])) }
\end{align*}
Pasting these together, we get
\begin{align*}
	&\lambda u.\lambda c.\case c {
		\inl\,a &\mapsto& \meridshrt{u} \codot (\inl\,(\lambdabar u.a)) \\
		\inr\,b &\mapsto& \meridshrt{u} \codot (\inr\,(\lambdabar u.b))
	} \\
	&: \lollishrt{u}.\paren{ A \uplus B \to \transpshrt{u} \codot ((\lollishrt{v}.A[v/u]) \uplus (\lollishrt{v}.B[v/u])) }.
\end{align*}
Transposing again as in \cref{sec:ff:transposition}, we find
\begin{align*}
	&i : (\lollishrt{u}.A \uplus B) \to (\lollishrt{u}.A) \uplus (\lollishrt{u}.B) \\
	&i\,\hat c = \unmeridsym\,\paren{  u.\case{\hat c\,u}{
		\inl\,a &\mapsto& \meridshrt{u} \codot (\inl\,(\lambdabar u.a)) \\
		\inr\,b &\mapsto& \meridshrt{u} \codot (\inr\,(\lambdabar u.b))
	}  }.
\end{align*}
Let us consider our categorically motivated creation from a more type-theoretical perspective.
We obtain an argument $\hat c : \lollishrt{u}.A \uplus B$ which we would like to pattern match on, in order to create an element of type $(\lollishrt{u}.A) \uplus (\lollishrt{u}.B)$.
Of course we cannot pattern match on a function, so we call the $\unmeridsym$ constructor which brings $u : \IU$ in scope and changes the goal to $\transpshrt{u} \codot ((\lollishrt{v}.A[v/u]) \uplus (\lollishrt{v}.B[v/u]))$.
We can then reduce $\hat c$ by one dimension by applying it to $u$, allowing a case analysis.
The first case brings in scope $a : A$ (and the second case will be analogous), so we are in context $(\Gamma, \hat c : \lollishrt{u}.A \uplus B, u : \IU, a : A)$.
We then use the meridian constructor, which again removes $u$ from scope, turns $a : A$ into a function $\lambdabar u.a : \lollishrt{u} . A$ and again reduces the goal to $(\lollishrt{u}.A) \uplus (\lollishrt{u}.B)$, so that $\inl$ completes the proof.
We have essentially pattern matched on a higher-dimensional object!

Let us now check that $i$ is indeed inverse to the trivial implementation of $i\inv$. %
We have:
\begin{align*}
	&\mathrel{\phantom{=}} (i \circ i\inv)(\inl\,\hat a)
	= i(\lambda u.\inl\,(\hat a\,u)) \\
	&= \unmeridsym \paren{u.
		\enspace
		\paren{\meridshrt{u}\,(\inl\,(\lambdabar u.a))}[u/u, \hat a\,u/a]
		\enspace
	} \\
	&= \unmeridsym \paren{u.\paren{
		\meridshrt{u}\,(\inl\enspace(\lambdabar u.a)[\lambdabar u.\hat a\,u/\lambdabar u.a]\enspace)
	}} 
	& \text{(\ruleref{ff:transp:intro:nat})} \\
	&= \unmeridsym \paren{u.\paren{
		\meridshrt{u}\,(\inl\,(\lambda u.\hat a\,u))
	}}
	& \text{(\cref{thm:subst-lambdabar-nil})} \\
	&= \inl\,\hat a.
	& \text{(\ruleref{ff:transp:beta})}
\end{align*}
and similar for $(i \circ i\inv)(\inr\,\hat b)$.
Using the technique of higher dimensional pattern matching just developed, we can prove the other equation also by pattern matching! By similar steps as before, we have:
\begin{align*}
	(i\inv \circ i)(\lambda u.\inl\,(\hat a\,u))
	&= i\inv(\unmeridsym(u.\enspace(\meridshrt{u} \codot (\inl\,(\lambdabar u.a)))[u/u, \hat a\,u/a]\enspace)) \\
	&= i\inv(\unmeridsym(u.(\meridshrt{u} \codot (\inl\,\hat a)))) = i\inv(\inl\,\hat a) = \lambda u.\inl\,(\hat a\,u),
\end{align*}
and a similar result for $(i\inv \circ i)(\lambda u.\inr\,(\hat b\,u))$.

\section{Multimode Type Theory} \label{sec:mtt}
As announced, we will rely on the extensional version of Gratzer et al.'s multimode and multimodal dependent type system \MTT{} \cite{mtt-journal,mtt} in order to frame the transpension and its left adjoints as modal operators.
We refer to the original work for details, but give a brief overview in the current section.
In \cref{sec:mtt:left}, we decorate the usual \MTT{} notation with reminders of the modalities' semantic left adjoints, which are syntactically obscured by the lock notation.

\subsection{The mode theory} \label{sec:mtt:mode-theory}
\MTT{} is parametrized by a \emph{mode theory}, which is a strict 2-category whose objects, morphisms and 2-cells we will refer to as \textbf{modes}, \textbf{modalities} and, well, \textbf{2-cells} respectively. Semantically, every mode $p$ will correspond to an entire model of dependent type theory $\interp p$. A modality $\ismod \mu p q$ will consist of a functor $\interp{\locknovar \mu} : \interp q \to \interp p$ acting on contexts and substitutions, and an operation $\interp{\modshade \mu}$ that is almost a dependent right adjoint (DRA \cite{dra}) to $\interp{\locknovar \mu}$; for all our purposes it will be an actual DRA and even one arising from a weak CwF morphism \cite[lemma 17]{dra}\cite{reldtt-techreport}. A 2-cell $\iskey \alpha \mu \nu$ is interpreted as a natural transformation $\interp{\key \alpha} : \interp{\locknovar \nu} \to \interp{\locknovar \mu}$ and hence also gives rise to an appropriate transformation $\interp{\keyshade \alpha} : \interp{\modshade \mu} \to \interp{\modshade \nu}$.

\subsection{Judgement forms} \label{sec:mtt:jud}
The judgement forms of \MTT{} are listed in \cref{fig:mtt:jud}.
All forms are annotated with a mode $p$ which specifies in what category they are to be interpreted.
Every judgement form also has a corresponding equality judgement, which is respected by everything as the typing rules are to be read as a specification of a generalized algebraic theory (GAT \cite{gat,tt-in-tt}).
The statements $p \mode$ and $\ismod \mu p q$ and $\iskey \alpha \mu \nu$ are simply requirements about the mode theory. This means we give no syntax or equality rules for modalities and 2-cells: these are fixed by the choice of mode theory.
\begin{figure}
	\small
	\figtitle{Judgement forms:}
	\begin{align*}
		&p \sep \Gamma \ctx
		&\text{$\Gamma$ is a context at mode $p$,} \\
		&p \sep \sigma : \Delta \to \Gamma
		&\text{$\sigma$ is a simultaneous substitution from $\Delta$ to $\Gamma$ at mode $p$,} \\
		&p \sep \Gamma \sez T \type
		&\text{$T$ is a type in context $\Gamma$ at mode $p$,} \\
		&p \sep \Gamma \sez t : T
		&\text{$t$ has type $T$ in context $\Gamma$ at mode $p$.}
	\end{align*}
	\caption{Judgement forms of \MTT{} \cite{mtt-journal}.}
	\label{fig:mtt:jud}
\end{figure}

\subsection{Typing rules} \label{sec:mtt:typing}
The typing rules are listed in \cref{fig:mtt:dtt,fig:mtt:struct,fig:mtt:wdra,fig:mtt:pi}
and discussed below.

\subsubsection{The type theory at each mode}
\begin{figure}
	\figtitle{The type theory at each mode:}\\
	Basic rules of dependent type theory (including all desired types) at each mode $q$, e.g.:
	\begin{equation*}
		\inferencel{ctx-empty}{
			q \mode
		}{
			q \sep \cdot \ctx
		}{}
		\qquad
		\inferencel{ctx-ext}{
			q \sep \Gamma \sez T \type
		}{
			q \sep \Gamma, x : T \ctx
		}{}
		\qquad
		\inferencel{ctx-ext:intro}{
			q \sep \sigma : \Delta \to \Gamma
			\quad
			q \sep \Delta \sez t : T[\sigma]
		}{
			q \sep (\sigma, t/x) : \Delta \to (\Gamma, x : T) \\
			\infwhere \tau = (\wknvar x \circ \tau, x[\tau]/x) \\
			\infnowhere (\sigma, t/x) \circ \rho = (\sigma \circ \rho, t[\rho]/x)
		}{}
	\end{equation*}
	\begin{equation*}
		\inferencel{ctx-ext:wkn}{
			q \sep \Gamma \ctx \qquad
			q \sep \Gamma \sez T \type
		}{
			q \sep \wknvar x : (\Gamma, x : T) \to \Gamma \\
			\infwhere \wknvar x \circ (\sigma, t/x) = \sigma
		}{}
		\quad
		\inferencel{ctx-ext:var}{
			q \sep \Gamma \ctx \qquad
			q \sep \Gamma \sez T \type
		}{
			q \sep \Gamma, x : T \sez x : T[\wknvar x] \\
			\infwhere x[\sigma, t/x] = t
		}{}
		\quad
		\inferencel{sigma}{
			q \sep \Gamma \sez A \type \\
			q \sep \Gamma, x : A \sez B \type
		}{
			q \sep \Gamma \sez (x : A) \times B \type
		}{}
	\end{equation*}
	\begin{equation*}
		\inferencel{uni}{
			q \sep \Gamma \ctx \quad
			\ell \in \IN
		}{
			q \sep \Gamma \sez \unimode{\ell}{q} \type_{\ell+1}
		}{}
		\qquad
		\inferencel{uni:elim}{
			q \sep \Gamma \sez t : \unimode \ell q
		}{
			q \sep \Gamma \sez \El(t) \type_\ell \\
			\infwhere \El(\tycode T) = T
		}{}
		\qquad
		\inferencel{uni:intro}{
			q \sep \Gamma \sez T \type_\ell
		}{
			q \sep \Gamma \sez \tycode T : \unimode \ell q \\
			\infwhere \tycode{\El(t)} = t
		}{}
	\end{equation*}
	\caption{\MTT{} includes all rules of ordinary DTT at each mode.}
	\label{fig:mtt:dtt}
\end{figure}
Since every mode corresponds to a model of all of dependent type theory (DTT), we start by importing \textbf{all the usual typing rules of DTT}, to be applied in \MTT{} at any given fixed mode.
Some examples of such rules are given in \cref{fig:mtt:dtt}, where we have consciously included rules for non-modal context extension, even though these will be generalized to modal rules later on.
One reason to do so is that other rules of DTT, such as \ref{rule:sigma}, depend on these and therefore cannot be imported without.
Another is that this way, we have a warm-up towards the modal rules and in particular we can make a point about de Bruijn indices.
Although variables in \cref{fig:mtt:dtt} are named, the rules \ref{rule:ctx-ext:wkn} and \ref{rule:ctx-ext:var} effectively enforce a \textbf{de Bruijn} discipline, where we can only name the last variable in the context and have to weaken explicitly if it is deeper down, e.g.
\[
	x : A, y : B, z : C \sez x[\wknvar y][\wknvar z] : A[\wknvar x][\wknvar y][\wknvar z].
\]
We take the viewpoint\footnote{as is done in the MTT technical report \cite{mtt}; the paper \cite{mtt-journal} speaks from a more implementation-oriented perspective.} that the official system is unnamed and uses this substitution-based de Bruijn discipline to refer to variables.
In order to improve human communication, we will name variables anyway and use the resulting redundancy to leave weakening substitutions implicit unambiguously.
This allows for the following unofficial admissible `rule'
\[
	\inferencel{ctx-ext:var:lookup}{
		q \sep \Gamma, x : T, \Delta \ctx
	}{
		q \sep \Gamma, x : T, \Delta \sez x : T
	}{}.
\]
Furthermore, we use other common notational conventions such as writing $(t/x)$ instead of $(\idsub_\Gamma, t/x) : \Gamma \to (\Gamma, x : T)$.

We assume that DTT has a \textbf{universe} \`a la Coquand with mutually inverse encoding and decoding operations (which we will henceforth suppress), and we ignore cumulativity-related hassle, referring to Gratzer et al.\ \cite{mtt-journal} for details.

\subsubsection{Modal types, part 1}
\begin{figure}
	\begin{equation*}
		\begin{array}{l}
			\inferencel{wdra}{
				\ismod \mu p q \\
				p \sep \Gamma, \lock{\tf m}{\mu} \sez A \type_\ell
			}{
				q \sep \Gamma \sez \Modify{\tf m} \mu A \type_\ell
			}{}
			\infnewline
			\inferencel{wdra:intro}{
				\ismod \mu p q \\
				p \sep \Gamma, \lock{\tf m}{\mu} \sez a : A
			}{
				q \sep \Gamma \sez \modify{\tf m}{\mu} a : \Modify{\tf m} \mu A
			}{}
		\end{array}
		\quad
		\begin{array}{l}
			\inferencel{wdra:elim}{
				\ismod \mu p q \quad
				\ismod \nu q r \\
				q \sep \Gamma, \lock{\tf n}{\nu} \sez \hat a : \Modify{\tf m} \mu A \\
				r \sep \Gamma, \ctxmod{\tf n}{\nu}{\hat x}{\Modify{\tf m} \mu A} \sez C \type \\
				r \sep \Gamma, \ctxmod{\tf n \twhisk \tf m}{\nu \circ \mu}{x}{A} \sez c : C[\varsub{\tf n}{\modify{\tf m} \mu x \vartrans{}{\tf n \twhisk \tf m}}{\hat x}]
			}{
				r \sep \Gamma \sez \letmodify{\tf n}{\nu}{\tf m}{\mu}{x}{\hat a} c : C[\varsub{\tf n}{\hat a}{\hat x}] \\
				\infwhere \letmodify{\tf n}{\nu}{\tf m}{\mu}{x}{\modify{\tf m} \mu a} c = c[\varsub{\tf n \twhisk \tf m} a x]
			}{}
		\end{array}
	\end{equation*}
	\caption{Typing rules for \MTT's modal types (weak DRAs) \cite{mtt-journal}\cite[fig.\ 5.6]{nuyts-phd}.}
	\label{fig:mtt:wdra}
\end{figure}
Before proceeding to the \MTT{}-specific structural rules, let us first have a look at the formation and introduction rules \ref{rule:wdra} and \ref{rule:wdra:intro} of \textbf{modal types} $\Modifynovar{\mu}{A}$ in \cref{fig:mtt:wdra}.
These are not unlike the formation and introduction rules of the transpension type in \cref{fig:ff} and work by transposition: we apply the left adjoint of the modality $\modshade \mu$ (in the form of a lock) to the premise's context. As such, they behave like DRAs, but their elimination rule \ref{rule:wdra:elim} (which we consider later) is weaker, so we call them weak DRAs.

\subsubsection{Structural rules}
\begin{figure}
	\figtitle{Context locking:}
	
	\begin{flushleft}
		Note: We write $(\sigma, \locknovar \mu)$ as shorthand for $(\sigma, \keytyped{\id}{\mu}{\mu})$, an instance of \ruleref{lock:fmap}.
	\end{flushleft}
	\small
	\begin{equation*}
		\inferencel{lock}{
			q \sep \Gamma \ctx \quad
			\ismod \mu p q
		}{
			p \sep \Gamma, \lock{\tf m}{\mu} \ctx \\
        			\infwhere \Gamma = (\Gamma, \lock{\ttriv}{\id}) \\
				\infnowhere (\Gamma, \lock{\tf n}{\nu}, \lock{\tf m}{\mu}) = (\Gamma, \lock{\tf{n} \twhisk \tf{m}}{\nu \circ \mu})
		}{}
		\qquad
		\inferencel{lock:fmap}{
			q \sep \sigma : \Gamma \to \Delta \qquad
			\ismod{\mu, \nu}{p}{q} \qquad
			\iskey{\alpha}{\mu}{\nu}
		}{
			p \sep (\sigma, \ttrans{\alpha}{\tf m}{\tf n}) : (\Gamma, \lock{\tf n} \nu) \to (\Delta, \lock{\tf m} \mu) \\
			\infwhere \sigma = (\sigma, \keytyped{\id}{\id}{\id})
			\quad (\sigma, \ttrans{\alpha'}{\tf m'}{\tf n'}, \ttrans{\alpha}{\tf m}{\tf n}) = (\sigma, \ttrans{(\alpha' \whisk \alpha)}{\tf m' \twhisk \tf m}{\tf n' \twhisk \tf n}) \\
			\infnowhere \id = (\id, \keytyped{\id}{\mu}{\mu})
			\quad (\sigma, \ttrans{\alpha}{\tf m}{\tf n}) \circ (\tau, \ttrans{\beta}{\tf n}{\tf o}) = (\sigma \circ \tau, \ttrans{(\beta \circ \alpha)}{\tf m}{\tf o})
		}{}
	\end{equation*}
	\normalsize
	
	\figskip
	
	\figtitle{Modal context extension:}
	
	\begin{flushleft}
	We consider the non-modal rule \ref{rule:ctx-ext} and its introduction, elimination and computation rules as a special case of \ref{rule:ctx-modext} for $p = q$ and $\modshade \mu = \modshade \id$.
	\end{flushleft}
	\begin{equation*}
		\inferencel{ctx-modext}{
			q \sep \Gamma \ctx \quad
			\ismod \mu p q \\
			p \sep \Gamma, \lock{\tf m}{\mu} \sez T \type
		}{
			q \sep \Gamma, \ctxmod{\tf m}{\mu}{x}{T} \ctx
		}{}
		\qquad
		\inferencel{ctx-modext:intro}{
			q \sep \sigma : \Delta \to \Gamma \qquad
			\ismod \mu p q \\
			p \sep \Delta, \lock{\tf m} \mu \sez t : T[\sigma, \locknovar \mu]
		}{
			q \sep (\sigma, \varsub{\tf m}{t}{x}) : \Delta \to (\Gamma, \ctxmod{\tf m} \mu x T) \\
			\infwhere \tau = (\wknvar x \circ \tau, \varsub{\tf m}{x[\tau, \locknovar \mu]}{x}) \\
			\infnowhere (\sigma, \varsub{\tf m}{t}{x}) \circ \rho = (\sigma \circ \rho, \varsub{\tf m}{t[\rho, \locknovar \mu]}{x})
		}{}
	\end{equation*}
	\begin{equation*}
		\quad
		\inferencel{ctx-modext:wkn}{
			q \sep \Gamma \ctx \qquad
			\ismod \mu p q \\
			p \sep \Gamma, \lock{\tf m} \mu \sez T \type
		}{
			q \sep \wknvar x : (\Gamma, \ctxmod{\tf m} \mu x T) \to \Gamma \\
			\infwhere \wknvar x \circ (\sigma, \varsub{\tf m}{t}{x}) = \sigma
		}{}
		\qquad
		\inferencel{ctx-modext:var}{
			q \sep \Gamma \ctx &
			\ismod \mu p q \\
			p \sep \Gamma, \lock{\tf m} \mu \sez T \type
		}{
			q \sep \Gamma, \ctxmod{\tf m}{\mu}{x}{T}, \lock{\tf m'} \mu \sez x\vartrans{}{\tf m'} : T[\wknvar x, \locknovar \mu] \\
			\infwhere \Delta, \locknovar \mu \sez x\vartrans{}{\tf m'} [\sigma, \varsub{\tf m'''}{t}{x}, \locknovar \mu] = t : T[\sigma, \locknovar \mu]
		}{}
	\end{equation*}
	\caption{Structural rules of \MTT{} \cite{mtt-journal}\cite[fig.\ 5.5]{nuyts-phd}.}
	\label{fig:mtt:struct}
\end{figure}
The structural rules of \MTT{} are listed in \cref{fig:mtt:struct}. \textbf{Context formation} starts with the empty context which exists at any mode, and proceeds by adding locks and variables.

Adding \textbf{locks} (\ref{rule:lock}) is strictly functorial: it preserves identity and composition of modalities. In fact, it is strictly 2-functorial: it also has an action on 2-cells (\ref{rule:lock:fmap}, producing substitutions between locked contexts) that preserves identity and composition of 2-cells. It is also strictly bifunctorial: we can combine a substitution and a 2-cell to a substitution between locked contexts. If the 2-cell is the identity, then we write $\locknovar \mu$ for $\keytyped{\id}{\mu}{\mu}$.

A \textbf{modal variable} $\ctxmod{\tf m}{\mu}{x}{T}$ introduced by \ref{rule:ctx-modext} is essentially the same as a non-modal variable $\hat x : \Modify{\tf m}{\mu}{T}$ (which in turn is shorthand for $\ctxmod{\ttriv}{\id}{\hat x}{\Modify{\tf m}{\mu}{T}}$), but the judgemental modal annotation allows direct access to a variable of type $A$ through the variable rule.
Hence, the type $T$ is checked the same way as it would be in $\Modify{\tf m}{\mu}{T}$.
Terms $t$ substituted for a modal variable $x$ are also checked in the locked context, as if we would be substituting $\modify{\tf m}{\mu}{t}$ instead.
The \textbf{variable} rule does not produce $x : \Modify{\tf m}{\mu}{T}$ but instead uses transposition to move the modality $\modshade \mu$ to the left in the form of a lock. As such, it can be seen as implicitly involving a co-unit.

By analogy with \ref{rule:ctx-ext:var:lookup}, we would like a more general unofficial variable `rule' that allows accessing a variable $x$ that is buried under a general telescope $\Delta$ rather than a single lock.
By \ref{rule:lock:fmap}, we can weaken under locks (which, like ordinary weakening, we will leave notationally implicit), so we can easily remove all variables from $\Delta$ and then apply strict functoriality of \ref{rule:lock} to fuse the remaining locks, obtaining a single lock $\lock{\ticks(\Delta)}{\locks(\Delta)}$, where the modality $\modshade{\locks(\Delta)}$ is defined as follows:
\begin{equation*}
		\locks(\cdot) = \id, \qquad
		\locks(\Delta, \lock{\tf m} \mu) = \locks(\Delta) \circ \mu, \qquad
		\locks(\Delta, \ctxmod{\tf m} \mu x T) = \locks(\Delta).
\end{equation*}
Then the only remaining thing we need is a 2-cell $\iskey \alpha \mu {\locks(\Delta)}$.
This leads to the following admissible variable `rule'
\[
	\inferencel{ctx-modext:var:lookup}{
		q \sep \Gamma \ctx &
		\ismod \mu p q \\
		p \sep \Gamma, \lock{\tf m} \mu \sez T \type &
		\iskey \alpha \mu {\locks(\Delta)}
	}{
		q \sep \Gamma, \ctxmod{\tf m}{\mu}{x}{T}, \Delta \sez x\vartrans{\alpha}{\ticks(\Delta)} : T[\idsub_\Gamma, \ttrans{\alpha}{\tf m}{\ticks(\Delta)}]
	}{}
\]
where $x\vartrans{\alpha}{\ticks(\Delta)}$ is defined as $x\vartrans{}{\tf m'}[\idsub_{(\Gamma, \ctxmod{\tf m}{\mu}{x}{T})}, \keytyped{\alpha}{\mu}{\locks(\Delta)}]$, leaving the necessary weakenings under locks implicit.
Substitution is then given by
\begin{align*}
	x \vartrans{\alpha}{\ticks(\Delta)} [\idsub_\Gamma, \varsub{\tf m}{a}{x}, \idsub_\Delta]
	&= x\vartrans{}{\tf m'}[\idsub_{(\Gamma, \ctxmod{\tf m}{\mu}{x}{T})}, \ttrans{\alpha}{\tf m'}{\ticks(\Delta)}] [\idsub_\Gamma, \varsub{\tf m}{a}{x}, \idsub_\Delta] \\
	&= x\vartrans{}{\tf m'} [\idsub_\Gamma, \varsub{\tf m}{a}{x}, \locknovar \mu] [\idsub_{\Gamma}, \ttrans{\alpha}{\tf m'}{\ticks(\Delta)}] \\
	&= a[\idsub_{\Gamma}, \ttrans{\alpha}{\tf m'}{\ticks(\Delta)}],
\end{align*}
or more briefly $x \varkey \alpha[a/x] = a[\key \alpha]$.

\subsubsection{Modal types, part 2}
\textbf{Modal elimination} (\ref{rule:wdra:elim}, \cref{fig:mtt:wdra}) uses a \textsf{let}-syntax to turn the modal type into a judgemental annotation on a variable.

\subsubsection{Modal function types}
Modal \textbf{function type} formation and introduction (\cref{fig:mtt:pi}) are by simple abstraction. Modal function application checks the argument in a locked context, just like modal variable substitution does.
\begin{figure}
	\begin{equation*}
			\inferencel{modpi}{
				p \sep \Gamma, \lock{\tf m} \mu \sez A \type_\ell \quad
				\ismod \mu p q \\
				q \sep \Gamma, \ctxmod{\tf m}{\mu}{x}{A} \sez B \type_\ell
			}{
				q \sep \Gamma \sez (\ctxmod{\tf m} \mu x A) \to B \type_\ell
			}{}
	\end{equation*}
	\begin{equation*}
			\inferencel{modpi:intro}{
				\ismod \mu p q \\
				q \sep \Gamma, \ctxmod{\tf m}{\mu}{x}{A} \sez b : B
			}{
				q \sep \Gamma \sez \lambda(\varmod\mu x).b : (\ctxmod{\tf m} \mu x A) \to B \\
				\infwhere \lambda(\varmod \mu x).(f \modappnovar{\mu} {x\vartrans{}{\tf m}}) = f
			}{}
			\qquad
			\inferencel{modpi:elim}{
				q \sep \Gamma \sez f : (\ctxmod{\tf m} \mu x A) \to B \\
				p \sep \Gamma, \lock{\tf m} \mu \sez a : A \quad
				\ismod \mu p q
			}{
				q \sep \Gamma \sez f \modappnovar{\mu} a : B[a/x] \\
				\infwhere (\lambda(\varmod \mu x).b) \modappnovar{\mu} a = b[\varsub{\tf m} a x]
			}{}
	\end{equation*}
	\caption{Typing rules for \MTT's modal $\Pi$-types \cite{mtt-journal}\cite[fig.\ 5.7]{nuyts-phd} (of which non-modal $\Pi$-types are a special case for $p = q$ and $\mu = \id$).}
	\label{fig:mtt:pi}
\end{figure}

\subsection{Left adjoint reminders} \label{sec:mtt:left}
In \MTT{}, we can on one hand apply a modality $\ismod \mu p q$ to a type $T$ to obtain a modal type $\Modify{\tf m}{\mu}{T}$, and on the other hand we can apply its left adjoint $\locknovar{\mu}$ to a context $\Gamma$ to obtain $\Gamma, \lock{\tf m}{\mu}$.
Type-theoretically, it is only sensible that the lock $\locknovar{\mu}$ mentions $\modshade \mu$ for $\modshade \mu$ is a premise of the \ref{rule:lock} rule.
From a semantic/categorical viewpoint however, the indirection of mentioning $\modshade \mu$ when you are actually talking about some left adjoint functor $K = \interp{\locknovar\mu}$ to $M = \interp{\modshade \mu}$ can really get in the way of understanding what is going on.

For example, in \cref{sec:substructural}, we will come up with a modality $\modshade{\transpshrt{u}}$ for the transpension, and then the introduction rule $\ref{rule:wdra:intro}$ will take the form
\begin{equation*}
	\inference{
		p \sep \Delta, \lock{\tf t}{\transpshrt{u}} \sez t : T
	}{
		q \sep \Delta \sez \modify{\tf t}{\transpshrt{u}}{t} : \Modify{\tf t}{\transpshrt{u}}{T}
	}{}
\end{equation*}
for certain modes $p$ and $q$, which is quite reminiscent of the introduction rule of the transpension type in \FFsys{} (\cref{sec:ff}) \emph{if you bear in mind} that $\interp{\locknovar{\transpshrt{u}}}$ is essentially $\lollishrt{u}$.

Instead of littering this paper with remarks of the form `recall that $\interp{\locknovar \mu} = \ldots$', we will decorate locks, keys and variables with superscript reminders of the left adjoints to the modalities (and 2-cells mediating them) that are already in subscript.
Concretely, we will assign:
\begin{itemize}
	\item to every modality $\ismod \mu p q$ a \emph{left name} $\islmod \kappa q p$ (writing this succinctly as $\ismodadj \kappa \mu p q$),
	\item to every 2-cell $\iskey{\alpha}{\mu}{\mu'}$ a left name $\islkey{\omega}{\kappa'}{\kappa}$ (where $\modadj \kappa \mu$ and $\modadj{\kappa'}{\mu'}$; writing succinctly $\iskeyadj{\omega}{\alpha}{\kappa}{\mu}{\kappa'}{\mu'}$).
\end{itemize}
Of course if $\modadj{\kappa'}{\mu'}$ and $\modadj{\kappa}{\mu}$, then the composite will be $\modadj{\kappa \circ \kappa'}{\mu' \circ \mu}$.
If a modality $\modshade \mu$ has a left adjoint \emph{modality} $\modshade \nu$, then we will always use $\lmodshade \nu$ as the left name of $\modshade \mu$, and similar for 2-cells.
Then, we can write $\bilock \kappa \mu$ for $\locknovar \mu$, and $\bikeytyped{\omega}{\alpha}{\kappa}{\mu}{\kappa'}{\mu'}$ or just $\bikey{\omega}{\alpha}$ for $\keytyped{\alpha}{\mu}{\mu'}$, and $x \varbikeytyped{\omega}{\alpha}{\kappa}{\mu}{\kappa'}{\mu'}$ or just $x \varbikey \omega \alpha$ for $x \varkeytyped{\alpha}{\mu}{\mu'}$.
Note that we have
\begin{align}
(\Gamma, \bilock{\kappa_1}{\mu_1}, \bilock{\kappa_2}{\mu_2})
&= (\Gamma, \bilock{\kappa_2 \circ \kappa_1}{\mu_1 \circ \mu_2}),
&(\sigma, \bikey{\omega_1}{\alpha_1}, \bikey{\omega_2}{\alpha_2})
&= (\sigma, \bikey{\omega_2 \whisk \omega_1}{\alpha_1 \whisk \alpha_2}),
&a[\bikey{\omega}{\alpha}][\bikey{\psi}{\beta}]
&= a[\bikey{\omega \circ \psi}{\beta \circ \alpha}]. \label{eq:variance-compose}
\end{align}

\subsection{Results} \label{sec:mtt:results}
We highlight some results about \MTT{} that are relevant in the current paper.
\begin{proposition} \label{thm:modify-functorial}
	We have $\Modify{\ttriv}{\id}{A} \cong A$ and $\Modify{\tf n \twhisk \tf m}{\nu \circ \mu}{A} \cong \Modify{\tf n}{\nu}{\Modify{\tf m}{\mu}{A}}$.
\end{proposition}
\begin{proposition} \label{thm:modify-2cell}
	For any 2-cell $\iskey \alpha \mu \nu$, we have $\Modify{\tf m}{\mu}{A} \to \Modify{\tf n}{\nu}{A[\ttrans{\alpha}{\tf m}{\tf n}]}$.
\end{proposition}
\begin{proposition}[Projection] \label{thm:projmod}
	If $\modshade \kappa \dashv \modshade \mu$ internal to the mode theory, with unit $\iskey{\eta}{\id}{\mu \circ \kappa}$ and co-unit $\iskey{\eps}{\kappa \circ \mu}{\id}$, then there is a function $\eps : (\tymod{\tf k}{\kappa}{\Modify{\tf m}{\mu}{A}}) \to A[\ttrans{\eps}{\tf k \twhisk \tf m}{\ttriv}]$, satisfying a $\beta$- and (thanks to extensionality) an $\eta$-law:
	\[
		\eps \modappnovar{\kappa} (\modify{\tf m}{\mu}{a}) = a[\ttrans{\eps}{\tf k \twhisk \tf m}{\ttriv}], \qquad
		\hat a = \modify{\tf m}{\mu}{(\eps \modappnovar{\kappa} (\hat a[\ttrans{\eta}{\ttriv}{\tf m \twhisk \tf k}]))}.
	\]
	Combined with these rules, $\eps$ is equally expressive as the $\mathsf{let}$-eliminator for $\Modifynovar{\mu}{\loch}$.
\end{proposition}
\begin{proposition}[Internal transposition] \label{thm:transpose}
	Let $\modshade \kappa \dashv \modshade \mu$ internal to the mode theory, with unit $\iskey{\eta}{\id}{\mu \circ \kappa}$ and co-unit $\iskey{\eps}{\kappa \circ \mu}{\id}$.
	Adding left names, we get $\modadj \zeta \kappa \dashv \modshade \mu$ with $\iskeyadj{\eps'}{\eta}{\id}{\id}{\zeta \circ \kappa}{\mu \circ \kappa}$ and $\iskeyadj{\eta'}{\eps}{\kappa \circ \zeta}{\kappa \circ \mu}{\id}{\id}$.
	
	Then there is an isomorphism of contexts expressing that $\kappa$ respects context extension:
	\begin{align*}
		\sigma = (\varsub{\tf k}{x \varbikey{\eps'}{\eta}}{y})
		&:
		(\Gamma, x : A, \bilock{\kappa}{\mu})
		\cong
		(\Gamma, \bilock{\kappa}{\mu}, \ctxmod{\tf k}{\kappa}{y}{A[\bikey{\eps'}{\eta}]}).
	\end{align*}
	The inverse is given by:
	\begin{equation*}
		\xymatrix{
			(\Gamma, \bilock{\kappa}{\mu}, \ctxmod{\tf k}{\kappa}{y}{A[\bikey{\eps'}{\eta}]})
				\ar[d]^{\displaystyle (\idsub_{(\Gamma, \bilock{\kappa}{\mu}, y)}, \bikey{\eta'}{\eps})}
			\\
			(\Gamma, \bilock{\kappa}{\mu}, \ctxmod{\tf k}{\kappa}{y}{A[\bikey{\eps'}{\eta}]}, \bilock{\zeta}{\kappa}, \bilock{\kappa}{\mu})
				\ar[d]^{\displaystyle (\id_\Gamma, \bikey{\eps'}{\eta}, y/x, \bilock \kappa \mu)}
			\\
			(\Gamma, x : A, \bilock{\kappa}{\mu})
		}
	\end{equation*}
	
	Correspondingly, given $B$ in the codomain context of $\sigma$, there is an isomorphism of types
	\begin{align*}
		(x : A) \to \Modify{\tf m}{\mu}{B[\sigma]}
		\qquad \cong \qquad
		\Modify{\tf m}{\mu}{(\ctxmod{\tf k}{\kappa}{y}{A[\bikey{\eps'}{\eta}]}) \to B}
	\end{align*}
	expressing internal transposition.
\end{proposition}

\section{The Modal Transpension System (\Msys{}):\texorpdfstring{\\}{} General Mode Theory and Semantics} \label{sec:quantifiers-as-modalities-full}
As mentioned in \cref{sec:quantifiers-as-modalities}, in our modal transpension system (\Msys{}), the transpension modality $\modshade{\transpshrt{u}}$ will be part of an adjoint triple of internal modalities $\modshade{\wknshrt{u}} \dashv \modshade{\funcshrt{u}} \dashv \modshade{\transpshrt{u}}$ together with weakening $\modshade{\wknshrt{u}}$ and universal quantification $\modshade{\funcshrt{u}}$ or, more generally in potentially non-cartesian systems, an adjoint triple $\modshade{\freshshrt{u}} \dashv \modshade{\lollishrt{u}} \dashv \modshade{\transpshrt{u}}$ together with fresh weakening $\modshade{\freshshrt{u}}$ and substructural (e.g.\ linear/affine) universal quantification $\modshade{\lollishrt{u}}$.
The further left adjoints $\lmodshade{\pairshrt{u}}$ (cartesian) or $\lmodshade{\sumshrt{u}}$ (potentially non-cartesian) cannot be internalized because every internal \MTT{} modality needs to have a further semantic left adjoint; thus, they will only appear as left names.

Notably, the aforementioned modalities all bind or depend on a variable, a phenomenon which is not supported by \MTT{}.
We shall address this issue in the current section by grouping shape variables such as $u : \IU$ in a \textbf{shape context} which is not considered part of the type-theoretic context but instead serves as the \emph{mode} of the judgement.

We assume that there are no prior modalities, i.e.\  that the type system to which we wish to add a transpension type is non-modal in the sense that it has a single mode and only the identity modality.
We assume that this single prior mode is modelled by the presheaf category $\Psh(\catW)$.
Prior modalities and in particular their commutation with the modalities mentioned above, are considered in the technical report \cite{transpension-techreport}.

\subsection{Shape contexts} \label{sec:shp-ctx}
Assume we have in the prior system a context $\XX$ modelled by a presheaf $\Xi$ over $\catW$.
Then the presheaves $\Psh(\catW/\Xi)$ over the category of elements $\catW/\Xi$ of the presheaf $\Xi$ are also a model of dependent type theory.
Denoting the judgements of the latter system with a prefix $\XX \sep$, it happens to be the case that judgements $\XX \sep \Gamma \sez J$ (i.e.\ $\Gamma \sez J$ in $\Psh(\catW/\Xi)$) have precisely the same meaning as judgements $\XX.\Gamma \sez J$ in $\Psh(\catW)$ (for a suitable but straightforward translation of $J$).
Thus, we will group together all shape variables (variables for which we want a transpension type) in a \textbf{shape context} $\XX$ in front of the typing context.
Our judgements will then take the form $\XX \sep \Gamma \sez J$.
Modal techniques will be used to signal what part of the context $\Gamma$ is fresh for a shape variable $u : \IU$, as this can then no longer be signalled by the position of $u : \IU$ in the context.
All of this allows us to frame the shape context $\XX$ as the \emph{mode} of the judgement, as it determines the category $\Psh(\catW/\Xi)$ in which the judgement is modelled.

Concretely, we fix a set of \textbf{shapes} and generate shape contexts by the following rules:
\begin{equation*}
	\inferencel{shp-ctx-empty}{}{
		{\cdot} \shpctx
	}{}
	\qquad
	\inferencel{shp-ctx-ext}{
		\XX \shpctx \qquad
		\IU \shape
	}{
		\XX, u : \IU \shpctx
	}{}
\end{equation*}
In \cref{sec:add:boundary}, we will additionally add boundary variables.
More generally, users of the current system could add shape context constructors at will, as long as they can be interpreted as presheaves over $\catW$.

\subsection{Mode theory} \label{sec:mode-theory}
For simplicity, we take a highly general mode theory and will then only be able to say interesting things about specific modalities and 2-cells.
In practice, and especially in implementations, one will want to select a more syntactic subtheory right away.

As \textbf{modes}, we take shape contexts. An interpretation function $\interp \loch$ from shape contexts to presheaves over $\catW$ will be defined in \cref{sec:substructural:def}.
The mode $\XX$ is modelled in $\Psh(\catW/\interp \XX)$.%
\footnote{As we will see later on, the available shapes must in some sense already be present in the base category, so that a context consisting purely of shapes will in general be representable.
As such, we could alternatively interpret modes as \emph{representable} presheaves over $\catW$, which via the Yoneda-embedding are just the objects of $\catW$.
This is perfectly possible and would (again by inserting the Yoneda-embedding) require virtually no changes to our approach, although a number of intermediate results in the technical report \cite{transpension-techreport} would become unnecessary.
However, the current approach is strictly more general, allows us to speak about boundaries (definitions \ref{def:boundary} and \ref{def:dir-dim-split}) in the shape context, and did not require any compromise in the strength of our results.}

As \textbf{modalities} $\ismod{\mu}{\XX_1}{\XX_2}$, we take all functors $\interp{\locknovar \mu} : \Psh(\catW/\interp{\XX_2}) \to \Psh(\catW/\interp{\XX_1})$ which have a right adjoint $\interp{\modshade \mu}$ that is then automatically a weak CwF morphism \cite{transpension-techreport}\allowbreak\cite[thm. 6.4.1]{nuyts-phd} giving rise to a DRA \cite[lemma 17]{dra}\cite[\S 2.1.3]{reldtt-techreport}.%
\footnote{A designated right adjoint can be retrieved from the left adjoint without the axiom of choice \cite[\S 2.3.6]{transpension-techreport}.}

As \textbf{2-cells} $\iskey{\alpha}{\mu}{\nu}$, we take all natural transformations $\interp{\key \alpha} : \interp{\locknovar \nu} \to \interp{\locknovar \mu}$, which automatically give rise to natural transformations $\interp{\keyshade \alpha} : \interp{\modshade \mu} \to \interp{\modshade \nu}$.

\section{\Msys{} Modalities for Substitution} \label{sec:subst}
In the previous section, we have defined modalities as left adjoint functors and 2-cells as natural transformations. As such, we have neglected to provide an actual syntax; any syntax we use should be shallowly defined on semantic objects.

We take a similar approach to shape substitutions.
A shape substitution from $\XX_1$ to $\XX_2$ is defined as a presheaf morphism $\sigma : \interp{\XX_1} \to \interp{\XX_2}$.
We will consistently write the interpretation brackets so as to avoid confusion with modalities $\XX_1 \to \XX_2$.
A presheaf morphism \emph{is} not a modality but it gives rise to a pair of modalities:%
\footnote{Note in particular that $\modshade{\wknshrt{\loch}}$ turns the arrow around: a presheaf morphism (shape substitution) $\sigma : \interp{\XX_1} \to \interp{\XX_2}$ gives rise to a substitution modality $\ismod{\wknshrt{\sigma}}{\XX_2}{\XX_1}$ sending types $T$ in shape context $\XX_2$ to types $\Modify{\tf o}{\wknshrt{\sigma}}{T}$ in shape context $\XX_1$.}
\begin{theorem} \label{thm:subst}
	Any presheaf morphism $\sigma : {\Xi_1} \to {\Xi_2}$ gives rise to a triple of adjoint functors
	\[
		\pairpsh{} \sigma \dashv \wknpsh{} \sigma \dashv \funcpsh{} \sigma,
	\]
	\[
		\pairpsh{} \sigma, \funcpsh{} \sigma : \Psh(\catW/{\Xi_1}) \to \Psh(\catW/{\Xi_2}) \qquad
		\wknpsh{} \sigma : \Psh(\catW/{\Xi_2}) \to \Psh(\catW/{\Xi_1})
	\]
	the latter two of which can be internalized as modalities  (with an additional left name) $\modadj{\pairshrt{\sigma}}{\wknshrt{\sigma}} \dashv \modshade{\funcshrt{\sigma}}$ with
	\[
		\interp{\bilock{\pairshrt{\sigma}}{\wknshrt{\sigma}}} = \pairpsh{} \sigma, \qquad
		\interp{\modshade{\wknshrt{\sigma}}} = \wknpsh{} \sigma, \qquad
		\interp{\bilock{\wknshrt{\sigma}}{\funcshrt{\sigma}}} = \wknpsh{} \sigma, \qquad
		\interp{\modshade{\funcshrt{\sigma}}} = \funcpsh{} \sigma.
	\]
	We denote the units and co-units as
	\begin{align*}
		\copypsh{}{\sigma} &: \id \to \wknpsh{}{\sigma} \circ \pairpsh{}{\sigma}
		& %
		\droppsh{}{\sigma} &: \pairpsh{}{\sigma} \circ \wknpsh{}{\sigma} \to \id \\
		\constpsh{}{\sigma} &: \id \to \funcpsh{}{\sigma} \circ \wknpsh{}{\sigma}
		& %
		\apppsh{}{\sigma} &: \wknpsh{}{\sigma} \circ \funcpsh{}{\sigma} \to \id \\
		\lkeyshade{\dropsym_\sigma} \Dashv \keyshade{\constsym_\sigma} &: \modshade{\id} \Rightarrow \modshade{\funcshrt{\sigma} \circ \wknshrt{\sigma}}
		& \lkeyshade{\copysym_\sigma} \Dashv \keyshade{\appsym_\sigma} &: \modshade{\wknshrt{\sigma} \circ \funcshrt{\sigma}} \Rightarrow \modshade{\id}
	\end{align*}
	Under the correspondence of semantic contexts $\XX \sep \Gamma \ctx$ (i.e.\ presheaves over $\catW/\interp \XX$) with semantic types $\interp \XX \sez \Gamma \type$, the functor $\wknpsh{}{\sigma}$ is exactly the semantics of ordinary type substitution in the standard presheaf model \cite{Hofmann97-presheaf-chapter} and hence, if $\sigma$ is a weakening substitution, then $\pairpsh{}{\sigma}$ and $\funcpsh{}{\sigma}$ are naturally isomorphic to the semantics of ordinary $\Sigma$- and $\Pi$-types.
	
	The functor $\wknpsh{}{\loch}$ and the modality $\modshade{\funcshrt{\loch}}$ are strictly functorial (they respect identity and composition of presheaf morphisms on the nose) whereas the functors $\pairpsh{}{\loch}$, $\funcpsh{}{\loch}$ and the modality $\modshade{\wknshrt{\loch}}$ are pseudofunctorial\footnote{However, Gratzer et al.\ \cite{mtt} have a strictification theorem for models of \MTT{} which could be used to strictify $\modshade{\wknshrt{\loch}}$.}.
\end{theorem}
\begin{proof}
	The morphism $\sigma$ gives rise to a functor $\pairslice{} \sigma : \catW/\Xi_1 \to \catW/\Xi_2 : (W, \psi) \to (W, \sigma\psi)$ and hence via left Kan extension, precomposition and right Kan extension \cite{stacks-adjoints} to a triple of adjoint functors $\pairpsh{}{\sigma} \dashv \wknpsh{}{\sigma} \dashv \funcpsh{}{\sigma}$ between the presheaf categories.
	The claim about type substitution follows from unfolding the definitions and the claims about $\Sigma$- and $\Pi$-types then follow from uniqueness of adjoints.
	
	Strict functoriality of $\wknpsh{}{\loch}$ follows immediately from the construction.
	Strict functoriality of $\modshade{\funcshrt{\loch}}$ then follows from the fact that a modality $\modshade \mu$ is fully defined by the semantic left adjoint $\interp{\locknovar \mu}$.
	Pseudofunctoriality of the others follows by uniqueness of the adjoint.
\end{proof}
\begin{remark}[Substitution as a DRA] \label{rem:subst-dra}
	In presheaf models, contexts are essentially the same thing as closed types (a property called \emph{democracy}).
	The shape substitution operation for contexts $\bilock{\wknshrt{\sigma}}{\funcshrt{\sigma}}$ is modelled by $\wknpsh{}{\sigma}$, i.e.\ by ordinary substitution.
	However, the shape substitution operation applicable to types is the modal type former $\Modifynovar{\wknshrt{\sigma}}{\loch}$, which is \emph{not} equivalent.
	Indeed, if $\Modifynovar{\wknshrt{\sigma}}{T}$ is closed, i.e.\ $\XX_1 \sep {\cdot} \sez \Modifynovar{\wknshrt{\sigma}}{T} \type$, then $T$ lives in context $\XX_2 \sep {\cdot}, \bilock{\pairshrt{\sigma}}{\wknshrt{\sigma}} \sez T \type$, so it is not closed.
	The operation $\Modifynovar{\wknshrt{\sigma}}{\loch}$ is in general still modelled as a substitution, but now it is one between the semantic contexts $\interp{\XX_1}.\interp{\Gamma}$ and $\interp{\XX_2}.\pairpsh{}{\sigma}\interp{\Gamma}$ which are \emph{isomorphic}.
	This is especially clear if $\sigma : \interp{\XX}.\interp{\Delta} \to \interp{\XX}$ is a weakening substitution for a context $\XX \sep \Delta \ctx$, in which case we are dealing with $ \interp{\XX}.\interp{\Delta}.\interp{\Gamma}$ and $\interp{\XX}.(\Sigma \interp{\Delta} \interp{\Gamma})$.
	We can still let $\Modifynovar{\wknshrt{\sigma}}{\loch}$ act on a closed type $\Xi_2 \sep {\cdot} \sez S \type$, however, but we first have to weaken $S$ to bring it to context ${\cdot}, \bilock{\pairshrt{\sigma}}{\wknshrt{\sigma}}$.
	The composite of these two operations -- weakening over $\bilock{\pairshrt{\sigma}}{\wknshrt{\sigma}}$ and then applying $\Modifynovar{\wknshrt{\sigma}}{\loch}$ -- is in fact equivalent with the operation $\bilock{\pairshrt{\sigma}}{\wknshrt{\sigma}}$ on contexts.
	
	This remark is relevant to the $\modshade{\wknshrt{\sigma}}$ modality specifically because its lock $\bilock{\pairshrt{\sigma}}{\wknshrt{\sigma}}$ does not preserve the empty context, whereas most other modalities' locks do.
\end{remark}
\begin{notation} \label{not:subst}
	We will use a slightly unconventional notation for \textbf{substitutions} in order to have them make maximal sense both as a substitution (as in $\modshade{\wknshrt{\sigma}}$) and as a function domain (as in $\modshade{\funcshrt{\sigma}}$):
	\begin{itemize}
		\item Every weakened variable (if weakening is available for the given shape) will be declared, e.g.\ we get a presheaf morphism $(u : \IU) : \interp{\XX, u : \IU} \to \interp\XX$ and hence $\ismod{\wknlong{u : \IU}}{\XX}{(\XX, u : \IU)}$ and $\ismod{\funclong{u : \IU}}{(\XX, u : \IU)}{\XX}$.
		Furthermore, we may omit the shape, writing just $\modshade{\wknshrt{u}}$ and $\modshade{\funcshrt{u}}$.
		Thus, this is what weakening and shape abstraction look like:
		\[
			\inference{
				\XX \sep \Gamma, \bilock{\pairshrt{u}}{\wknshrt{u}} \sez t : T
			}{
				\XX, u : \IU \sep \Gamma \sez \modify{\tf o}{\wknshrt{u}}{t} : \Modify{\tf o}{\wknshrt{u}}{T}
			}{},
			\qquad
			\inference{
				\XX, u : \IU \sep \Gamma, \bilock{\wknshrt{u}}{\funcshrt{u}} \sez t : T
			}{
				\XX \sep \Gamma \sez \modify{\tf p}{\funcshrt{u}}{t} : \Modify{\tf p}{\funcshrt{u}}{T}
			}{}.
		\]
		The projection function (\cref{thm:projmod}) for $\modshade{\funcshrt{u}}$ is function application:
		\begin{align*}
			\inference{
				\XX, u : \IU \sep \Gamma, \bilock{\wknshrt{u} \circ \pairshrt{u}}{\wknshrt{u} \circ \funcshrt{u}} \sez A \type
			}{
				\XX, u : \IU \sep \Gamma \sez \appsym_u : (\tymod{\tf o}{\wknshrt{u}}{\Modify{\tf p}{\funcshrt{u}}{A}}) \to A\brac{ \bikeytyped{\copysym_u}{\appsym_u}{\wknshrt{u} \circ \pairshrt{u}}{\wknshrt{u} \circ \funcshrt{u}}{\id}{\id} }
			}{\cref{thm:projmod}}
		\end{align*}
		The 2-cell $\iskeyadj{\copysym_u}{\appsym_u}{\wknshrt{u} \circ \pairshrt{u}}{\wknshrt{u} \circ \funcshrt{u}}{\id}{\id}$ signals a contraction of shape variables, namely of the one bound by the $\Pi$-modality and the one to which the function is applied.
		
		\item When a variable is substituted, we denote this as $u := t$ instead of $t/u$, e.g.\ in a cubical type theory we get a presheaf morphism $(i := 0) : \interp \XX \to \interp{\XX, i : \IX}$ and hence $\ismod{\wknlong{i := 0}}{(\XX, i : \IX)}{\XX}$ which binds $i$ and $\ismod{\funclong{i := 0}}{\XX}{(\XX, i : \IX)}$ which depends on $i$, so we may substitute $0$ for $i$ but we may also abstract over the assumption that $i$ is $0$:
		\[
			\inference{
				\XX, i : \IX \sep \Gamma, \bilock{\pairlong{i := 0}}{\wknlong{i := 0}} \sez t : T
			}{
				\XX \sep \Gamma \sez \modify{\tf o}{\wknlong{i := 0}}{t} : \Modify{\tf o}{\wknlong{i := 0}}{T}
			}{},
			\qquad
			\inference{
				\XX \sep \Gamma, \bilock{\wknlong{i := 0}}{\funclong{i := 0}} \sez t : T
			}{
				\XX, i : \IX \sep \Gamma \sez \modify{\tf p}{\funclong{i := 0}}{t} : \Modify{\tf p}{\funclong{i := 0}}{T}
			}{}.
		\]
		And apply:
		\begin{align*}
			\inference{
				\XX \sep \Gamma, \bilock{\wknlong{i := 0} \circ \pairlong{i := 0}}{\wknlong{i := 0} \circ \funclong{i := 0}} \sez A \type
			}{
				\XX \sep \Gamma \sez \applong{i := 0} : (\tymodnovar{\wknlong{i := 0}}{\Modifynovar{\funclong{i := 0}}{A}}) \to A \brac{ \bikeytyped{\copylong{i := 0}}{\applong{i := 0}}{\wknlong{i := 0} \circ \pairlong{i := 0}}{\wknlong{i := 0} \circ \funclong{i := 0}}{\id}{\id} }
			}{prop.\ \ref{thm:projmod}}
		\end{align*}
		\item Finally, if $\sigma$ involves weakening, then the codomain of the co-unit may be a \textbf{variable renaming} that is sugar for the identity, e.g.
		\[
			\iskey{\appsym_{(u/v : \IU)}}{\wknlong{u : \IU} \circ \funclong{v : \IU}}{\renamelong{u : \IU, v := u}}
		\]
		is exactly the same thing as
		\[
			\iskey{\appsym_{(u : \IU)}}{\wknlong{u : \IU} \circ \funclong{u : \IU}}{\id}.
		\]
		This way, we may adjust $\appsym_{(u : \IU)}$ in order to be able to apply to a \emph{different} variable:
		\begin{align*}
			\inference{
				\XX, v : \IU \sep \Gamma, \bilock{\wknshrt{v} \circ \pairshrt{u}}{\wknshrt{u} \circ \funcshrt{v}} \sez A \type
			}{
				\XX, u : \IU \sep \Gamma \sez \applong{u/v} : (\tymod{\tf o}{\wknshrt{u}}{\Modify{\tf p}{\funcshrt{v}}{A}}) \to \Modifynovar{\renamelong{u, v := u}}{   A\brac{ \bikeytyped{\copylong{v/u}}{\applong{u/v}}{\wknshrt{u} \circ \pairshrt{u}}{\wknshrt{u} \circ \funcshrt{u}}{\renamelong{v, u := v}}{\renamelong{u, v := u}} }   }
			}{}
		\end{align*}
	\end{itemize}
	Again, bear in mind that shape substitutions are in fact defined as presheaf morphisms and that therefore, notions such as \emph{weakening} and \emph{contraction} reflected by the syntax introduced here, need to be shallowly interpreted in presheaf morphisms.
\end{notation}
\begin{example} \label{ex:diag-pi}
	If $\IU$ is cartesian, i.e.\ $\interp{\XX, u : \IU} = \interp \XX \times \interp \IU$, then there is a diagonal substitution $(w : \IU, u := w, v := w) : \interp{\XX, w : \IU} \to \interp{\XX, u : \IU, v : \IU}$.
	Writing
	\begin{align*}
		&\keyshade{\alpha} = \keyshade{\id_{\modshade{\funcshrt{u}}} \whisk \id_{\modshade{\funcshrt{v}}} \whisk \constsym_{(w, u := w, v := w)}}
		: \modshade{\funcshrt{u} \circ \funcshrt{v}} \Rightarrow \\
		&\qquad \modshade{\funcshrt{u} \circ \funcshrt{v} \circ \funclong{w, u := w, v := w} \circ \wknlong{w, u := w, v := w}} \\
		&\qquad = \modshade{\funclong{(u : \IU} \circ (v : \IU) \circ (w, u := w, v := w)) \circ \wknlong{w, u := w, v := w}} \\
		&\qquad = \modshade{\funclong{(u : \IU} \circ (w, u := w)) \circ \wknlong{w, u := w, v := w}} \\
		&\qquad = \modshade{\funcshrt{w} \circ \wknlong{w, u := w, v := w}},
	\end{align*}
	where the equations use strict functoriality of $\modshade{\funcshrt{\loch}}$ and ordinary calculation of composition of substitutions,
	this allows us to type the na\"ively typed function $\lambda f.\lambda w.f\,w\,w : (\funcshrt{u}.\funcshrt{v}. A) \to \funcshrt{w}. A[w/u, w/v]$ as
	\begin{equation*}
		\Modify{\tf p_u}{\funclong{u : \IU}}{\Modify{\tf p_v}{\funclong{v : \IU}}{A}} \to
		\Modify{\tf p_w}{\funclong{w : \IU}}{\Modify{\tf o}{\wknlong{w : \IU, u := w, v := w}}{A[\ttrans{\alpha}{\tf p_u \twhisk \tf p_v}{\tf p_w \twhisk \tf o}]}}.
	\end{equation*}
\end{example}
\begin{remark} \label{rem:compute-wknsym}
	The reframing of shape substitutions as a modality, has the annoying consequence that substitution no longer reduces.
	However, both $\Modifynovar{\wknshrt{\sigma}}{\loch}$ and $\modifysym{\wknshrt{\sigma}}$ are semantically an ordinary substitution (along an isomorphism, see \cref{rem:subst-dra}).
	Thus, we could add computation rules such as:
	\begin{equation*}
		\begin{array}{r c l r c l}
			\Modify{\tf o}{\wknshrt{\sigma}}{A \times B} &=& \Modify{\tf o}{\wknshrt{\sigma}}{A} \times \Modify{\tf o}{\wknshrt{\sigma}}{B}, & {}\quad
			\Modify{\tf o}{\wknshrt{\sigma}}{\uni{}} &=& \uni{}, \\
			\modify{\tf o}{\wknshrt{\sigma}}{(a, b)} &=& (\modify{\tf o}{\wknshrt{\sigma}}{a}, \modify{\tf o}{\wknshrt{\sigma}}{b}), &
			\modify{\tf o}{\wknshrt{\sigma}}{\tycode A} &=& \tycode{\Modify{\tf o}{\wknshrt{\sigma}}{A}}.
		\end{array}
	\end{equation*}
	This is fine in an extensional type system, but would not play well with the $\beta$-rule for modal types in an intensional system.
	Indeed, $\beta$-reduction for $\Modify{\tf o}{\wknshrt{\sigma}}{A}$ requires a solution to the following problem: when $\hat a = \modify{\tf o}{\wknshrt{\sigma}}{a}$ definitionally, then we need to be able to infer $a$ up to definitional equality from $\hat a$.
	Alternatively, the eliminator for $\Modify{\tf o}{\wknshrt{\sigma}}{A}$ should somehow proceed by induction on $A$, e.g.\ an element of $\Modify{\tf o}{\wknshrt{\sigma}}{A \times B}$ could be eliminated as an element of $\Modify{\tf o}{\wknshrt{\sigma}}{A} \times \Modify{\tf o}{\wknshrt{\sigma}}{B}$.
	A third possibility would be to abolish elimination of $\Modifynovar{\wknshrt{\sigma}}{\loch}$ altogether, except when applied to type formers for which there is no definitional substitution-commutation law.
\end{remark}
\begin{remark} \label{rem:no-shape-subst}
	In type theory, we generally expect admissibility of substitution: given a derivable judgement $\Gamma \sez J$ and a substitution $\sigma : \Delta \to \Gamma$, we expect derivability of $\Delta \sez J[\sigma]$, where the operation $\loch[\sigma]$ can be applied to any term, type or other object in context and traverses its structure, leaving everything untouched except variables.
	A good way to guarantee admissibility of substitution is by making sure that every inference rule has a conclusion in a general context $\Gamma$ and that the context of any premise is obtained by applying a functorial operation to $\Gamma$.
	
	\emph{There is no such result for shape substitutions.}
	The conclusion of modal inference rules often has a non-general shape context, and the transpension type is in general not even respected by shape substitution \cite{transpension-techreport}.
	However, until we extend \Msys{} with additional rules in \cref{sec:add,sec:structure,sec:recover},\footnote{And possibly even after, see \cref{rem:close-subst}.} we do have a form of the usual result: given a derivable judgement $\XX \sep \Gamma \sez J$ and a substitution $\XX \sep \sigma : \Delta \to \Gamma$, we can derive $\XX \sep \Delta \sez J[\sigma]$.
\end{remark}

\section{Multipliers} \label{sec:substructural}
In this section, we introduce multipliers as a semantics for shapes, as well as the associated modalities $\modshade{\freshshrt{u}} \dashv \modshade{\lollishrt{u}} \dashv \modshade{\transpshrt{u}}$.
In \cref{sec:substructural:def}, we define multipliers and a number of criteria by which we can classify them.
In \cref{sec:substructural:dim-split}, we deal with a technical complication that we dubbed \emph{unpointability}\seelog{}, which shows up especially in models of guarded and nominal type theory.
In \cref{sec:examples}, we discuss an extensive number of examples.
In \cref{sec:wkn}, we consider how \emph{copointed}\seelog{} multipliers give rise to \emph{shape weakening} modalities that are a special instance of the modalities in \cref{sec:subst}.
In \cref{sec:act-elements}, we discuss how a multiplier and its associated operations lift from acting on base and slice categories to acting on categories of elements of semantic shape contexts.
Then in \cref{sec:multip}, we are finally ready to define the transpension modality and its adjoints and to state the quantification \cref{thm:quantification} that helps to understand them.
In \cref{sec:car-multip}, we say a bit more on cartesian multipliers.
In \cref{sec:techreport}, we briefly list the matters that are not discussed in the current paper but can be found in the technical report \cite{transpension-techreport}.

\subsection{Shapes and multipliers} \label{sec:substructural:def}
In \cref{sec:quantifiers-as-modalities-full}, we defined shape contexts as lists of variables and announced that these would be modelled as presheaves over $\catW$.
Several times, we have hinted at the fact that these shape variables need not satisfy all the usual structural rules (weakening, exchange and contraction).
In this section, we make these matters precise.

We associate to each shape $\IU$ a functor $\loch \multip U : \catW \to \catW$ which extends by left Kan extension to a functor $\loch \multip \yoneda U : \Psh(\catW) \to \Psh(\catW)$.%
\footnote{Both $\loch \multip U$ and $\loch \multip \yoneda U$ are to be regarded as single-character symbols, i.e.\  $\multip$ in itself is meaningless.
In most concrete applications, however, the multiplier is defined as some monoidal product $\loch \otimes U$ with a given object $U$, in which case the left Kan extension is naturally isomorphic to Day convolution with $\yoneda U$.
For this reason, we also refrain from defining $U := \top \multip U$ because we may not have $\top \otimes U = U$ on the nose for the object of interest $U$.}
We define the semantics of shape contexts $\interp \loch : \Shp \Ctx \to \Obj{\Psh(\catW)}$ as follows:
\begin{equation*}
	\interp \cdot = \top, \qquad
	\interp{\XX, u : \IU} = \interp \XX \multip \yoneda U.
\end{equation*}

Of course, if we model shape context extension with $u : \IU$ by an \emph{arbitrary} functor, then we will not be able to prove many results.
Depending on the properties of the functor, the variable $u$ will obey different structural rules and the $\Phi$-combinator \cite{moulin,moulin-param3} (\cref{sec:recover:phi}) will or will not be sound for $\IU$.
For this reason, we introduce some criteria that help us classify shapes.
Some of these criteria concern in fact the \emph{fresh weakening functor} for the given multiplier, which is essentially an instance of the following construction:
\begin{definition} \label{def:act-slice}
	Given a functor $F : \catV \to \catW$ and $V_0 \in \Obj\catV$, we define the action of $F$ on slice objects over $V_0$ as the functor
	\begin{equation*}
		F^{/V_0} : \catV/V_0 \to \catW/FV_0 : (V, \vfi) \mapsto (FV, F\vfi).
	\end{equation*}
\end{definition}
\begin{definition} \label{def:multip}
	Assume $\catW$ has a terminal object $\top$. A \textbf{multiplier} for an object $U$ is an endo\-functor\footnote{In the technical report \cite{transpension-techreport}, we generalize beyond endofunctors.} $\loch \multip U : \catW \to \catW$ such that $\top \multip U \cong U$.
	This gives us a natural second projection $\pi_2 : (\loch \multip U) \to U$. 
	
	We define the \textbf{fresh weakening functor} to the slice category as $\freshbase{U} : \catW \to \catW/U : W \mapsto (W \multip U, \pi_2)$, which is essentially the action of the multiplier on slice objects over $\top$.
	
	We say that a multiplier (as well as its shape) is:
	\begin{itemize}
		\item \textbf{Copointed}\seelog{} if it is equipped with a natural first projection $\pi_1 : (\loch \multip U) \to \mathrm{Id}$. \\
		\emph{This property carries over to $\loch \multip \yoneda U$ and thus allows for shape weakening.}
		\item A \textbf{comonad}\seelog{} if it is additionally equipped with a natural diagonal $\delta : (\loch \multip U) \to ((\loch \multip U) \multip U)$ such that $\pi_1 \circ \delta = (\pi_1 \multip U) \circ \delta = \id_{(\loch \multip U)}$ and $\delta \circ \delta = (\delta \multip U) \circ \delta : (\loch \multip U) \to (\loch \multip U)^3$. \\
		\emph{This property carries over to $\loch \multip \yoneda U$ and thus allows for shape variable contraction.}
		\item \textbf{Cartesian} if it is naturally isomorphic to the cartesian product with $U$. \\
		\emph{This property carries over to $\loch \multip \yoneda U$ and is thus a sufficient condition for allowing exchange. Additionally, it will erase the distinction between \emph{weakening} $\modshade{\wknshrt{u}}$ (\cref{sec:wkn}) and \emph{fresh shape weakening} $\modshade{\freshshrt{u}}$ (\cref{sec:multip}), see \cref{thm:quantification}.}
		\item \textbf{$\top$-slice faithful}\seelog{} if $\freshbase{U}$ is faithful. \\
		\emph{This is a basic well-behavedness property that is satisfied by all examples of interest. \Cref{ex:multip:initial} is a counterexample.}
		\item \textbf{$\top$-slice full}\seelog{} if $\freshbase{U}$ is full. \\
		\emph{For multipliers for objects other than $\top$, this property precludes the exchange rule (\cref{thm:tsfull-cartesian}).
		$\top$-slice fully faithful multipliers will give rise to fully faithful modalities for fresh weakening $\modshade{\freshshrt{u}}$ and transpension $\modshade{\transpshrt{u}}$ (\cref{thm:quantification}).}
		\item \textbf{$\top$-slice shard-free}\seelog{} if $\freshbase{U}$ is essentially surjective on slice objects $(V, \psi)$ such that $\psi : V \to U$ is dimensionally split (\cref{def:dim-split}). A \textbf{shard}\seelog{} is a slice object $(V, \psi)$ that is not up to isomorphism in the image of $\freshbase{U}$ even though $\psi$ is dimensionally split. \\
		\emph{Intuitively, a shard is a shape over $\IU$ that covers all of $\IU$ (as expressed by the fact that $\psi$ is dimensionally split, which just means split epi in most applications), but that is not prism-shaped in the direction of $\IU$ (as it would then be in the image of $\freshbase{U}$).
		Shard-freedom will be a requirement for the $\Phi$-rule \cite{moulin} to hold (\cref{thm:phi}) and for elimination of the transpension type by pattern matching to be sound (\cref{thm:transp-elim}).}
		\item \textbf{$\top$-slice right adjoint} if $\freshbase{U}$ has a left adjoint $\sumbase U : \catW/U \to \catW$.%
		\footnote{A functor $\loch \multip U$ with this property is usually called a \emph{parametric} or \emph{local right adjoint}, but the word `local' is overloaded \cite{nlab:locally} and so is `parametric', and we wanted uniform terminology.}
	\end{itemize}
\end{definition}
\begin{proposition} \label{thm:tsfull-cartesian}
	If a $\top$-slice full multiplier for $U$ is:
	\begin{itemize}
		\item a comonad, then $U$ is terminal,
		\item cartesian, then it is naturally isomorphic to the identity functor.
	\end{itemize}
\end{proposition}
\begin{proof}
	The second statement clearly follows from the first, so we only prove the first.
	Consider the following diagram:
	\begin{equation}
		\xymatrix{
			\top \multip U \ar[rr]^{\delta}
				\ar[rd]_{\pi_2}
			&& (\top \multip U) \multip U
				\ar[ld]^{\pi_2}
			\\
			& U
		}
	\end{equation}
	It commutes, because $\pi_2 = \pi_2 \circ (\pi_1 \multip U) : (\top \multip U) \multip U \to U$ and $(\pi_1 \multip U) \circ \delta = \id$.
	So it is a morphism of slice objects $\delta : \freshbase{U} \top \to \freshbase{U} (\top \times U)$ and thus, since $\freshbase{U}$ is full, of the form $\freshbase{U} \upsilon$ for some $\upsilon : \top \to \top \times U$. This means in particular that
	\begin{equation}
		\idsub_{\top \times U} = \pi_1 \circ \delta = \pi_1 \circ (\upsilon \times U) = \upsilon \circ \pi_1 : \top \times U \to \top \times U,
	\end{equation}
	so the identity on $\top \multip U$ factors over $\top$. Then $\top \multip U \cong U$ is terminal.
\end{proof}

\subsection{Pointability, dimensional splitness and boundaries} \label{sec:substructural:dim-split}
Before we move on to a list of examples, we owe the reader a definition for dimensional splitness (although the impatient reader may first read the example \cref{sec:examples}, ignoring shard-freedom, pointability and boundaries).
In most popular base categories, namely all the \emph{objectwise pointable} ones, we could have gotten away with saying `split epi' instead of `dimensionally split' (\cref{thm:dim-split-epi}).
\begin{definition}\label{def:pointable}
	Let $\catW$ be a category with terminal object $\top$. An object $W$ is \textbf{pointable}\seelog{} if $() : W \to \top$ is split epi, i.e.\ if there exists at least one morphism $\top \to W$.
	A category is \textbf{objectwise pointable}\seelog{} if each object is pointable.
\end{definition}
We have carefully chosen the above terminology to emphasize (1) that pointability is a property, not structure (the corresponding structure is called \emph{pointed}), and (2) that objectwise pointability does \emph{not} require that the pointings can be chosen naturally.
\begin{proposition}
	Let $\loch \multip U$ be a multiplier on an objectwise pointable category $\catW$. Then for any object $W$, the slice object $\freshbase U W$ is split epi.
\end{proposition}
\begin{proof}
	Any functor preserves split epimorphisms. We have $\freshbase U W = (W \multip U, \pi_2)$ and $\pi_2 : W \multip U \to U \cong \top \multip U$ is essentially the image of $W \to \top$.
\end{proof}

When dealing with a category that is not objectwise pointable, the above theorem does not hold and the definition of shard-freedom w.r.t. split epi slice objects would not make sense, so we need a somewhat more general notion:
\begin{definition} \label{def:dim-split} \label{def:boundary}
	Given a multiplier $\loch \multip U : \catW \to \catW$, we say that a morphism $\vfi : V \to U$ is \textbf{dimensionally split} if there is some $W \in \catW$ such that $\pi_2 : W \multip U \to U$ factors over $\vfi$. The other factor $\chi : W \multip U \to V$ such that $\pi_2 = \vfi \circ \chi$ will be called a \textbf{(dimensional) section} of $\vfi$. We write $\dimslice{\catW}{U}$ for the full subcategory of $\catW / U$ of dimensionally split slice objects.
	
	We define the \textbf{boundary} $\partial U$ as the subpresheaf of the Yoneda-embedding $\yoneda U$ consisting of those morphisms that are \emph{not} dimensionally split.
\end{definition}
Thus, a multiplier is $\top$-slice shard-free if and only if every dimensionally split slice object has an invertible dimensional section.
\begin{proposition}
	Let $\loch \multip U$ be a multiplier on $\catW$. Then for any object $W$, the slice object $\freshbase U W$ is dimensionally split with section $\idsub_{W \multip U}$. \qed
\end{proposition}
\begin{proposition} \label{thm:dim-split-epi}
	If $\catW$ is objectwise pointable, then a morphism $\vfi : V \to U$ is split epi if and only if it is dimensionally split. \cite{transpension-techreport}
\end{proposition}
The notion of dimensionally split morphisms lets us consider the boundary and shard-freedom (a requirement for modelling $\Phi$) also in base categories that are not objectwise pointable, where the output of $\freshbase U$ may not be split epi.
\begin{remark} \label{rem:cosieve-base}
	$\top$-slice shard-freedom can also be formulated using (co)sieves \cite{nlab:sieve}.
	A \textbf{sieve in $\catW$} is a full subcategory $\cat S$ such that if $W \in \cat S$ and $\vfi : V \to W$, then $V \in \cat S$.
	The dual (where $\vfi$ points the other way) is called a \textbf{cosieve in $\catW$}.
	Being full subcategories, (co)sieves can be regarded as subsets of $\Obj{\catW}$.
	A \textbf{sieve on $U \in \catW$} is a sieve in $\catW/U$ or, equivalently, a subpresheaf of $\yoneda U$.
	
	A multiplier is $\top$-slice shard-free if either of the following equivalent criteria is satisfied:
	\begin{itemize}
		\item The objects in the essential image of $\freshbase U$ constitute a cosieve in $\catW/U$ \cite{mathoverflow:image-cosieve}.
		\item The objects \emph{outside} the essential image of $\freshbase U$ constitute a sieve in $\catW/U$, i.e.\ a sieve on $U$.
	\end{itemize}
	The slice objects of the cosieve generated by objects of the essential image of $\freshbase U$, are called dimensionally split.
	The boundary $\partial U$ is the largest sieve on $U$ that is disjoint with the objects of the essential image of $\freshbase U$.
	
	If $\loch \multip U$ is $\top$-slice fully faithful, then the above conditions are furthermore equivalent to $\freshbase U$ being a Street opfibration.
\end{remark}

\subsection{Examples} \label{sec:examples}
Let us look at some examples of multipliers.
Their properties are listed in \cref{fig:multips}.
Most properties are easy to verify, so we omit the proofs.
\begin{example}[Identity] \label{ex:multip:id}
	The identity functor on an arbitrary category $\catW$ is a multiplier for $\top$.
	All slice objects $(W, ())$ are dimensionally split with invertible section $\idsub_W : \Id\,W \cong W$; hence the boundary is empty and there are no shards.
	It is $\top$-slice right adjoint, with $\sumbase \top : \catW/\top \to \catW : (W, ()) \mapsto W$.
\end{example}
\begin{example}[Cartesian product] \label{ex:multip:cartesian}
	Let $\catW$ be a category with finite products and $U \in \catW$.
	Then $\loch \times U$ is a multiplier for $U$, which is $\top$-slice full if and only if it is the identity (i.e.\ $U$ is terminal, cf.\ \cref{thm:tsfull-cartesian,ex:multip:id}).
	It is $\top$-slice right adjoint with $\sumbase U : \catW/U \to \catW : (W, \psi) \mapsto W$.
	Hence, we have $\sumbase U \freshbase{U} = \loch \times U$.
\end{example}
\begin{definition}
	Let $\RGcat$ be the category generated by the following diagram and equations:
	\begin{equation*}
		\xymatrix{
			\RGnode
				\ar@/^{1em}/[r]^{\RGsrc}
				\ar@/_{1em}/[r]_{\RGtgt}
			&
			\RGedge \ar[l]|{\RGrfl}
		}
		\qquad
		\begin{array}{r c l}
			\RGrfl \circ \RGsrc &=& \idsub_{\RGnode}, \\
			\RGrfl \circ \RGtgt &=& \idsub_{\RGnode}.
		\end{array}
	\end{equation*}
	A presheaf over $\RGcat$ is a reflexive graph.
	More generally, let $\ary \RGcat a$ (with $a \geq 0$) be the category generated by the following:
	\begin{equation*}
		\xymatrix{
			\RGnode
				\ar@/^{2.5em}/[rr]^{\RGend 0, \ldots, \RGend{a-1}}
				\ar@/^{2em}/[rr]
				\ar@/^{1.5em}/[rr]
				\ar@/^{1em}/[rr]%
			&&
			\RGedge \ar[ll]|{\RGrfl}
		}
		\qquad
		\begin{array}{r c l}
			\RGrfl \circ \RGend i &=& \idsub_{\RGnode}.
		\end{array}
	\end{equation*}
\end{definition}
\begin{example}[Cartesian cubes] \label{ex:multip:cartesian-cubes}
	Let $\ary \cubecat a$ be the category of cartesian $a$-ary cubes.
	It is the free cartesian monoidal category with same terminal object over $\ary \RGcat a$. Concretely:
	\begin{itemize}
		\item Its objects take the form $(i_1 : \IX, \ldots, i_n : \IX)$ (the names are desugared to de Bruijn indices, i.e. the objects are really just natural numbers),
		\item Its morphisms $(i_1 : \IX, \ldots, i_n : \IX) \to (j_1 : \IX, \ldots, j_m : \IX)$ are arbitrary functions
		\[
			\vfi : \accol{j_1, \ldots, j_m} \to \accol{i_1, \ldots, i_n} \cup \accol{0, \ldots, a-1} : j \mapsto j \psub \vfi.
		\]
		We also write $\vfi = (j_1 \psub \vfi/j_1, \ldots, j_m \psub \vfi /j_m)$. If a variable $i$ is not used, we may write $i/\novar$ to emphasize this.
	\end{itemize}
	This category is objectwise pointable if and only if $a \neq 0$.
	On this category, we consider the multiplier $\loch \times (i : \IX)$, which is an instance of \cref{ex:multip:cartesian} and therefore inherits all properties of cartesian multipliers.
	A slice $(V, \vfi)$ where $\vfi : V \to (i : \IX)$ is dimensionally split if $i \psub \vfi$ is not an endpoint, i.e.\ $i \psub \vfi \not\in \accol{0, \ldots, a-1}$, and in that case it is isomorphic to $\freshbase{(i : \IX)} V'$ where $V'$ is $V$ with $i \psub \vfi$ removed, so there are no shards.
	Clearly then, any morphism on the boundary factors through one of $a$ morphisms $(\varepsilon/i) : \top \to (i : \IX)$ where $\varepsilon \in \accol{0, \ldots, a-1}$, so $\partial \IX \cong \biguplus_{k = 0}^{a - 1} \top$.
\end{example}
\begin{example}[Affine cubes] \label{ex:multip:affine-cubes}
	Let $\ary \cubecat a_\affine$ be the category of affine $a$-ary cubes as used in \cite{model-cubical} (binary) or \cite{moulin-param3} (unary).
	It is the free semicartesian monoidal category with same terminal unit over $\ary \RGcat a$. Concretely:
	\begin{itemize}
		\item Objects are as in $\ary \cubecat a$,
		\item Morphisms are as in $\ary \cubecat a$ such that if $j \psub \vfi = k \psub \vfi \not\in \accol{0, \ldots, a-1}$, then $j = k$. This rules out diagonal maps.
	\end{itemize}
	This category is objectwise pointable if and only if $a \neq 0$.
	On this category, we consider the functor $\loch * (i : \IX) : W \mapsto (W, i : \IX)$, which is a multiplier for $(i : \IX)$.
	Dimensional splitness and the boundary are as in $\ary \cubecat a$.
	This functor is $\top$-slice right adjoint with $\sumbase{(i : \IX)}((W, j : \IX), (j/i)) = W$ and $\sumbase{(i : \IX)}(W, (\eps/i)) = W$ for each of the $a$ endpoints $\eps$.
	
	In the nullary case, $\ary \cubecat 0_\affine$ is the base category of the Schanuel topos, a sheaf topos equivalent to the category of nominal sets \cite{nominal-sets}.
	In that case, $\sumbase{(i : \IX)}$ is not just left adjoint to $\freshbase{(i : \IX)}$, but in fact an inverse and hence also right adjoint.
	This is in line with the fact that in nominal type theory \cite{freshmltt}, there is a single name quantifier which can be read as either existential or universal quantification.
\end{example}
\begin{figure}
	\footnotesize
	\begin{center}
	\begin{tabular}{l | c c || c | c c c c c c c c}
		\rotatebox{90}{Example}
		& \rotatebox{90}{Base category}
		& \rotatebox{90}{Multiplier}
		& \rotatebox{90}{\hspace{-1ex}\begin{tabular}{l}
			Objectwise poin- \\
			table category
		\end{tabular}}
		& \rotatebox{90}{\hspace{-1ex}\begin{tabular}{l}
			Copointed/ \\
			Weakening
		\end{tabular}}
		& \rotatebox{90}{Exchange}
		& \rotatebox{90}{\hspace{-1ex}\begin{tabular}{l}
			Comonad/ \\
			Contraction
		\end{tabular}}
		& \rotatebox{90}{Cartesian}
		& \rotatebox{90}{$\top$-s.\ faithful}
		& \rotatebox{90}{$\top$-s.\ full}
		& \rotatebox{90}{$\top$-s.\ shard-free}
		& \rotatebox{90}{$\top$-s.\ right adjoint}
		\\ \hline \hline
		\ref{ex:multip:id}
		& $\catW$ & $\Id$ %
		& ? %
		& \yes %
		& \yes %
		& \yes %
		& \yes %
		& \yes %
		& \yes %
		& \yes %
		& \yes %
		\\
		\ref{ex:multip:cartesian}
		& $\catW$ & $(\loch \times U) \not\cong \Id$ %
		& ? %
		& \yes %
		& \yes %
		& \yes %
		& \framebox{\yes} %
		& ? %
		& \nope %
		& ? %
		& \yes %
		\\
		\ref{ex:multip:cartesian-cubes}
		& $\ary \cubecat a$ & $\loch \times (i : \IX)$ %
		& $a \neq 0$ %
		& \gray{\yes} %
		& \gray{\yes} %
		& \gray{\yes} %
		& \yes %
		& \yes %
		& \gray{\nope} %
		& \yes %
		& \gray{\yes} %
		\\
		\ref{ex:multip:affine-cubes}
		& $\ary \cubecat a_\affine$ & $\loch * (i : \IX)$ %
		& $a \neq 0$ %
		& \yes %
		& \yes %
		& \nope %
		& \nope %
		& \yes %
		& \yes %
		& \yes %
		& \yes %
		\\
		\ref{ex:multip:cchm}
		& $\cchmcat$ & $\loch \times (i : \IX)$ %
		& \yes %
		& \gray{\yes} %
		& \gray{\yes} %
		& \gray{\yes} %
		& \yes %
		& \yes %
		& \gray{\nope} %
		& \nope %
		& \gray{\yes} %
		\\
		\ref{ex:multip:dcubes}
		& $\Dcubecat d$ & $\loch \times (i : \Dedge k)$ %
		& \yes %
		& \gray{\yes} %
		& \gray{\yes} %
		& \gray{\yes} %
		& \yes %
		& \yes %
		& \gray{\nope} %
		& \yes %
		& \gray{\yes} %
		\\
		\ref{ex:multip:clocks}
		& $\clockcat$ & $\loch \times (i : \clocksym_k)$ %
		& \nope %
		& \gray{\yes} %
		& \gray{\yes} %
		& \gray{\yes} %
		& \yes %
		& \yes %
		& \gray{\nope} %
		& \yes %
		& \gray{\yes} %
		\\
		\ref{ex:multip:twisted-cubes}
		& $\twistedcubecat$ & $\loch \multip \IX$ %
		& \yes %
		& \nope %
		& \nope %
		& \nope %
		& \nope %
		& \yes %
		& \yes %
		& \yes %
		& \yes %
		\\
		\ref{ex:multip:trees}
		& $n$ & $\min(\loch, i)$ %
		& \nope %
		& \gray{\yes} %
		& \gray{\yes} %
		& \gray{\yes} %
		& \yes %
		& \yes %
		& \gray{\nope} %
		& \yes %
		& \gray{\yes} %
		\\
		\ref{ex:multip:initial}
		& $\ary \cubecat 2_\bot$ & $\loch \times \bot$ %
		& \yes %
		& \gray{\yes} %
		& \gray{\yes} %
		& \gray{\yes} %
		& \yes %
		& \nope %
		& \gray{\nope} %
		& \yes %
		& \gray{\yes} %
	\end{tabular}
	\end{center}
	\caption[Transpension: Examples of multipliers]{Some interesting multipliers and their properties. Properties that follow from being cartesian are greyed out.}
	\label{fig:multips}
\end{figure}
\begin{example}[CCHM cubes] \label{ex:multip:cchm}
	Let $\cchmcat$ be the category of CCHM cubes \cite{cubical}, which is objectwise pointable. Its objects are as in $\ary \cubecat 2$ and its morphisms $(i_1 : \IX, \ldots, i_n : \IX) \to (j_1 : \IX, \ldots, j_m : \IX)$ are functions from $\accol{j_1, \ldots, j_m}$ to the free de Morgan algebra over $\accol{i_1, \ldots, i_n}$. We again consider $\loch \times (i : \IX)$, another instance of \cref{ex:multip:cartesian}.
	A slice object $(V, \vfi)$ is now (dimensionally) split if $i \psub \vfi$ is not an endpoint, so the boundary is again $\top \uplus \top$.
	The so-called \emph{connections} $(j \vee k/i), (j \wedge k/i) : (j : \IX, k : \IX) \to (i : \IX)$ are shards, because they have sections $(i/j, 0/k) : (i : \IX) \to (j : \IX, k : \IX)$ and $(i/j, 1/k) : (i : \IX) \to (j : \IX, k : \IX)$ respectively but are not in the image of $\freshbase{(i : \IX)}$.
\end{example}
\begin{example}[Depth $d$ cubes] \label{ex:multip:dcubes}
	Let $\Dcubecat d$ with $d \geq -1$ be the category of depth $d$ cubes, used as a base category in degrees of relatedness \cite{reldtt,reldtt-techreport}.
	This is a generalization of the category of binary cartesian cubes $\cubecat$, where instead of typing every dimension with \emph{the} interval $\IX$, we type them with the $k$-interval $\Dedge k$, where $k \in \accol{0, \ldots, d}$ is called the degree of relatedness of the edge.
	Its objects take the form $(i_1 : \Dedge{k_1}, \ldots, i_n : \Dedge{k_n})$.
	Conceptually, we have a map $\Dedge k \to \Dedge \ell$ if $k \geq \ell$.
	Thus, morphisms $\vfi : (i_1 : \Dedge{k_1}, \ldots, i_n : \Dedge{k_n}) \to (j_1 : \Dedge{\ell_1}, \ldots, j_m : \Dedge{\ell_m})$ send every variable $j : \Dedge \ell$ of the codomain to a value $j \psub \vfi$, which is either 0, 1 or a variable $i : \Dedge k$ of the domain such that $k \geq \ell$.
	\begin{itemize}
		\item If $d = -1$, then there is only one object $()$ and only the identity morphism, i.e. we have the point category.
		\item If $d = 0$, we just get $\cubecat$.
		\item If $d = 1$, we obtain the category of bridge/path cubes $\BPcubecat := \Dcubecat 1$. We write $\IP$ for $\Dedge 0$ (the path interval) and $\IB$ for $\Dedge 1$ (the bridge interval). Bridge/path cubical sets are used as a model for parametric quantifiers \cite{paramdtt,reldtt-techreport}.
	\end{itemize}
	On this category, we consider $\loch \times (i : \Dedge k)$, which is another instance of \cref{ex:multip:cartesian}.
	A slice object $(V, \vfi)$ is dimensionally split w.r.t.\ this multiplier if $i \psub \vfi$ is a variable of type $\Dedge k$ (and not $k' > k$), in which case a preimage is obtained as in $\cubecat$ (so there are no shards).
	Hence, a slice object is on the boundary if it factors over $(0/i), (1/i) : () \to (i : \Dedge k)$ or over $(i/i) : (i : \Dedge{k'}) \to (i : \Dedge k)$ for $k' > k$.
	These morphisms correspond to the cells of $\yoneda (i : \Dedge{k+1})$ for $k < d$, so $\partial(i : \Dedge k) \cong \yoneda (i : \Dedge{k+1})$ for $k < d$ and $\partial(i : \Dedge d) \cong \top \uplus \top$.
\end{example}
\begin{example}[Clocks] \label{ex:multip:clocks}
	Let $\clockcat$ be the category of clocks, used as a base category in guarded type theory \cite{clock-cat}.
	It is the free cartesian category over $\omega$.
	Concretely:
	\begin{itemize}
		\item Its objects take the form $(i_1 : \clocksym_{k_1}, \ldots, i_n : \clocksym_{k_n})$ where all $k_j \geq 0$.
		We can think of a variable of type $\clocksym_k$ as representing a clock (i.e. a time dimension) paired up with a certificate that we do not care what happens after the time on this clock exceeds $k$.
		\item Correspondingly, we should have a map $\clocksym_k \to \clocksym_\ell$ if $k \leq \ell$, because if the time exceeds $\ell$, then it certainly exceeds $k$ so our certificate can be legitimately adjusted.
		Then morphisms $\vfi : (i_1 : \clocksym_{k_1}, \ldots, i_n : \clocksym_{k_n}) \to (j_1 : \clocksym_{\ell_1}, \ldots, j_m : \clocksym_{\ell_m})$ are functions that send (de Bruijn) variables $j : \clocksym_\ell$ of the codomain to a variable $j \psub \vfi : \clocksym_k$ of the domain such that $k \leq \ell$.
	\end{itemize}
	This category is not objectwise pointable; indeed, the only pointable object is $()$.
	On this category we consider $\loch \times (i : \clocksym_k)$, which is another instance of \cref{ex:multip:cartesian}.
	A slice object $(V, \vfi)$ is dimensionally split w.r.t.\ this multiplier if $i \psub \vfi$ is a variable of type $\clocksym_k$ (and not $k' < k$), in which case a preimage is obtained as in $\cubecat$ (so there are no shards). Thus, $\partial(i : \clocksym_k) \cong \yoneda (i : \clocksym_{k-1})$ for $k > 0$ and $\partial(i : \clocksym_0) \cong \bot$.
\end{example}
\begin{example}[Twisted cubes] \label{ex:multip:twisted-cubes}
	Pinyo and Kraus's category of twisted cubes $\twistedcubecat$ \cite{pinyo-twisted}
	can be described as a subcategory of the category of non-empty finite linear orders (or, if you want, of its skeletalization: the category of simplices $\simplexcat$).
	On $\simplexcat$, we can define a functor $\loch \multip \IX$ such that $W \multip \IX = W\op \uplus_< W$, where we consider elements from the left smaller than those from the right.
	Now $\twistedcubecat$ is the subcategory of $\simplexcat$ whose objects are generated by $\top$ and $\loch \multip \IX$ (note that every object then also has an opposite since $\top\op = \top$ and $(V \multip \IX)\op \cong V \multip \IX$), and whose morphisms are given by
	\begin{itemize}
		\item $(\vfi, 0) : \Hom_\twistedcubecat(V, W \multip \IX)$ for all $\vfi : \Hom_\twistedcubecat(V, W\op)$,
		\item $(\vfi, 1) : \Hom_\twistedcubecat(V, W \multip \IX)$ for all $\vfi : \Hom_\twistedcubecat(V, W)$,
		\item $\vfi \multip \IX : \Hom_\twistedcubecat(V \multip \IX, W \multip \IX)$ for all $\vfi : \Hom_\twistedcubecat(V, W)$,
		\item $() : \Hom_\twistedcubecat(V, \top)$.
	\end{itemize}
	Note that this collection automatically contains all identities, composites, and opposites.
	Isomorphism to Pinyo and Kraus's category of twisted cubes can be seen from their ternary representation \cite[def. 34]{pinyo-twisted}.
	We now consider the multiplier $\loch \multip \IX : \twistedcubecat \to \twistedcubecat$, which
	Pinyo and Kraus call the twisted prism functor.
	A slice object $(V, \vfi)$ is dimensionally split if and only if it is of the form $\vfi = \psi \multip I$ (so there are no shards).
	Hence, all slice objects on the boundary factor over $((), 0) : () \to () \multip \IX$ or $((), 1) : () \to () \multip \IX$, so that $\partial \IX \cong \top \uplus \top$.
	The multiplier is $\top$-slice right adjoint with
	\begin{equation}
		\sumbase \IX : \left \{
			\begin{array}{l c l}
				(W, ((), 0)) &\mapsto& W\op \\
				(W, ((), 1)) &\mapsto& W \\
				(W \multip \IX, () \multip \IX) &\mapsto& W,
			\end{array}
		\right.
	\end{equation}
	with the obvious action on morphisms.
\end{example}
\begin{example}[Finite ordinals] \label{ex:multip:trees}
	In the base category $\omega$ of the topos of trees, used in guarded type theory \cite{birkedal:2012}, where $\Hom(i, j) = \set{{*}}{i \leq j}$, a cartesian product is given by $i \times j = \min(i, j)$.
	However, this category lacks a terminal object.
	Instead, on the subcategory $n$ with terminal object $n-1$, which is endowed with the same cartesian product,
	we consider the multiplier $\loch \times i$, which is again an instance of \cref{ex:multip:cartesian}.
	Any slice object $(j, {*})$ (where necessarily $j \leq i$) is dimensionally split with section ${*} : \min(i, j) = j \to j$; hence there are no shards and $\partial i = \bot$.
\end{example}
\begin{example}[Counterexample for $\top$-slice faithful] \label{ex:multip:initial}
	Let $\ary \cubecat 2_\bot$ be the category of binary cartesian cubes extended with an initial object. We consider the cartesian product $\loch \times \bot$ which sends everything to $\bot$. This is not $\top$-slice faithful, as $\freshbase \bot$ sends both $(0/i)$ and $(1/i) : () \to (i : \IX)$ to $[] : (\bot, []) \to (\bot, [])$.
	It is not $\top$-slice full, as there is no $\psi : () \to \bot$ such that $\psi \times \bot = [] : \freshbase \bot () \to \freshbase \bot \bot$.
\end{example}

\subsection{\Msys{} Modalities for weakening} \label{sec:wkn}
Recall that we write $\loch \multip \yoneda U$ for the left Kan extension of a multiplier $\loch \multip U$.
For any \textbf{copointed} multiplier $\loch \multip U : \catW \to \catW$ and any presheaf $\Xi \in \Psh(\catW)$, we get a presheaf morphism $\pi_1 : \Xi \multip \yoneda U \to \Xi$.
In this situation, the notations in \cref{thm:subst} are not very illuminating as they would only mention $\pi_1$ and not $\Xi$ or $U$.
Instead, we use the following notations:
\begin{notation}
	A functor acting on elements:
	\begin{itemize}
		\item $\pairslice U \Xi := \pairslice{}{\pi_1} : \catW/\Xi \multip \yoneda U \to \catW/\Xi$
	\end{itemize}
	Functors acting on presheaves:
	\begin{itemize}
		\item $\pairpsh{\yoneda U}{\Xi} := \pairpsh{}{\pi_1} : \Psh(\catW/\Xi \multip \yoneda U) \to \Psh(\catW/\Xi)$
		\item $\wknpsh{\yoneda U}{\Xi} := \wknpsh{}{\pi_1} : \Psh(\catW/\Xi) \to \Psh(\catW/\Xi \multip \yoneda U)$
		\item $\funcpsh{\yoneda U}{\Xi} := \funcpsh{}{\pi_1} : \Psh(\catW/\Xi \multip \yoneda U) \to \Psh(\catW/\Xi)$
	\end{itemize}
	Natural transformations:
	\begin{align*}
		\copypsh{\yoneda U}{\Xi} := \copypsh{}{\pi_1} &: 1 \to \wknpsh{\yoneda U}{\Xi} \circ \pairpsh{\yoneda U}{\Xi}
		& \droppsh{\yoneda U}{\Xi} := \droppsh{}{\pi_1} &: \pairpsh{\yoneda U}{\Xi} \circ \wknpsh{\yoneda U}{\Xi} \to 1 \\
		\constpsh{\yoneda U}{\Xi} := \constpsh{}{\pi_1} &: 1 \to \funcpsh{\yoneda U}{\Xi} \circ \wknpsh{\yoneda U}{\Xi}
		& \apppsh{\yoneda U}{\Xi} := \apppsh{}{\pi_1} &: \wknpsh{\yoneda U}{\Xi} \circ \funcpsh{\yoneda U}{\Xi} \to 1
	\end{align*}
	For modalities, we use the weakening notations already introduced in \cref{not:subst}: for $\Xi = \interp \XX$, we internalize the above functors as $\ismodadj{\pairlong{u : \IU}}{\wknlong{u : \IU}}{\XX}{(\XX, u : \IU)}$ and $\ismod{\funclong{u : \IU}}{(\XX, u : \IU)}{\XX}$, sometimes abbreviating to $\modadj{\pairshrt{u}}{\wknshrt{u}}$ and $\modshade{\funcshrt{u}}$,
	and the above natural transformations as $\keyadj{\dropsym_u}{\constsym_u}$ and $\keyadj{\copysym_u}{\appsym_u}$.
\end{notation}

\subsection{Acting on elements} \label{sec:act-elements}
A $\top$-slice right adjoint multiplier $\loch \multip U : \catW \to \catW$ as defined in \cref{def:multip} gives rise to a pair of adjoint functors $\sumbase U \dashv \freshbase U$ between $\catW$ and the slice category $\catW/U$,
and hence a pair of adjoint functors $\sumslice U \top \dashv \freshslice U \top$ between the categories of elements $\catW/\top$ and $\catW/(\top \multip \yoneda U)$ of the empty shape context $\interp{\cdot} := \top$ and the single variable shape context $\interp{u : \IU} := \top \multip \yoneda U \cong \yoneda U$ respectively.
As any functor between base categories gives rise to a triple of adjoint functors between presheaf categories, the adjoint pair $\sumbase U \dashv \freshbase U$ gives rise to an adjoint quadruple $\sumpsh{\yoneda U}{\top} \dashv \freshpsh{\yoneda U}{\top} \dashv \lollipsh{\yoneda U}{\top} \dashv \transppsh{\yoneda U}{\top}$ between the categories $\Psh(\catW/\top)$ and $\Psh(\catW/\yoneda U)$ that model the modes $()$ and $(u : \IU)$ respectively.
Thus, we are presently well-equipped to study the transpension type in a setting with \emph{at most one shape variable}.
Deeming this unsatisfactory, in the current section we intend to generalize the above functors, so that everywhere we mentioned $\top$ above we can instead have an arbitrary presheaf $\Xi$.
\begin{definition} \label{def:act-elements}
	Given a multiplier $\loch \multip U : \catW \to \catW$ and a presheaf $\Xi \in \Obj{\Psh(\catW)}$, we define:
	\[
		\freshslice U \Xi : \catW/\Xi \to \catW/(\Xi \multip \yoneda U) : (W, \xi) \mapsto (W \multip U, \xi \multip \yoneda U),
	\]
	where $\xi \multip \yoneda U$ denotes the $(W \multip U)$-shaped cell of $\Xi \multip \yoneda U$ obtained from $\xi$.
	
	We say that $\loch \multip U$ is:
	\begin{itemize}
		\item \textbf{Presheafwise faithful}\seelog{} if for all $\Xi$, the functor $\freshslice{U}{\Xi}$ is faithful,
		\item \textbf{Presheafwise full}\seelog{} if for all $\Xi$, the functor $\freshslice{U}{\Xi}$ is full,
		\item \textbf{Presheafwise shard-free}\seelog{}
		if for all $\Xi$, the functor $\freshslice{U}{\Xi}$ is essentially surjective on elements $(V, \vfi) \in {\catW}/{(\Xi \multip \yoneda U)}$ such that $\vfi$ is directly dimensionally split (\cref{def:dir-dim-split}).
		A \textbf{direct shard}\seelog{} is an element $(V, \xi) \in \catW / \Xi \multip \yoneda U$ that is not up to isomorphism in the image of $\freshslice{U}{\Xi}$ even though $\xi$ is directly dimensionally split.
		\item \textbf{Presheafwise right adjoint}\seelog{} if for all $\Xi$, the functor $\freshslice{U}{\Xi}$ has a left adjoint $\sumslice{U}{\Xi} : \catW / (\Xi \multip \yoneda U) \to \catW / \Xi$.
		We denote the unit as $\copyslice U{\Xi} : \Id \to \freshslice{U}{\Xi} \sumslice{U}{\Xi}$ and the co-unit as $\dropslice U{\Xi} : \sumslice{U}{\Xi} \freshslice{U}{\Xi} \to \Id$.
	\end{itemize}
\end{definition}
\begin{definition} \label{def:dir-dim-split} \label{def:dir-boundary}
	Given a multiplier $\loch \multip U : \catW \to \catW$, we say that a $V$-shaped presheaf cell $\vfi$ of ${\Xi \multip \yoneda U}$ is \textbf{directly dimensionally split}\seelog{} with direct dimensional section $\chi : W \multip U \to V$ if $\vfi \chi$ is of the form $\xi \multip \yoneda U$. The section can alternatively be presented as a morphism of elements $\chi : \freshslice{U}{\Xi} (W, \xi) \to (V, \vfi)$.
	We write $\dimslice{\catW}{(\Xi \multip \yoneda U)}$ for the full subcategory of $\catW / (\Xi \multip \yoneda U)$ of directly dimensionally split cells.
	
	We define the \textbf{(direct) boundary}\seelog{} $\Xi \multip \partial U$ as the subpresheaf of $\Xi \multip \yoneda U$ consisting of those cells that are \emph{not} directly dimensionally split.
\end{definition}
\begin{remark} \label{rem:cosieve-elements}
	Just like $\top$-slice shard-freedom (\cref{rem:cosieve-base}),
	presheafwise shard-freedom can be formulated using (co)sieves.
	A multiplier is presheafwise shard-free if either of the following equivalent criteria is satisfied:
	\begin{itemize}
		\item The objects in the essential image of $\freshslice U \Xi$ constitute a cosieve in $\catW/(\Xi \multip \yoneda U)$.
		\item The objects \emph{outside} the essential image of $\freshslice U \Xi$ constitute a sieve in $\catW/(\Xi \multip \yoneda U)$.
	\end{itemize}
	The objects of the cosieve generated by objects of the essential image of $\freshslice U \Xi$, are called directly dimensionally split.
	The boundary $\Xi \multip \partial U$ is the largest sieve in $\catW/(\Xi \multip \yoneda U)$ (largest subpresheaf of $\Xi \multip \yoneda U$) that is disjoint with the objects of the essential image of $\freshslice U \Xi$.
	
	If $\loch \multip U$ is presheafwise fully faithful, then the above conditions are furthermore equivalent to $\freshslice U \Xi$ being a Street opfibration.
\end{remark}
Since we can instantiate $\Xi$ with the terminal presheaf $\top \cong \yoneda \top$, we see that each of the presheafwise criteria implies the $\top$-slice criterion from \cref{def:multip}.
Below we give \emph{sufficient} conditions for a multiplier to satisfy the presheafwise criteria:
\begin{proposition}
	The multiplier $\loch \multip U : \catW \to \catW$ is:
	\begin{itemize}
		\item presheafwise faithful if it is $\top$-slice faithful,
		\item presheafwise full if it is $\top$-slice fully faithful,
		\item presheafwise shard-free if it is $\top$-slice full and shard-free,
		\item presheafwise right adjoint if it is $\top$-slice right adjoint.
	\end{itemize}
\end{proposition}
\begin{proof}
	See \cite{transpension-techreport}.
\end{proof}
\begin{example} \label{ex:act-elements:affine-cubes}
	Continuing \cref{ex:multip:affine-cubes} about $a$-ary affine cubes, let $\Xi = \yoneda W$.
	Then $\Xi \multip \yoneda(i : \IX) \cong \yoneda(W, i : \IX)$.
	Pick $(V, \vfi)$ in the category of elements, which is essentially the slice category over $(W, i : \IX)$, i.e.\ we view $\vfi$ as a morphism $V \to (W, i : \IX)$.
	Then $\vfi$ is directly dimensionally split if $i \psub \vfi$ is not an endpoint,
	and in that case $(V, \vfi)$ is isomorphic to $\freshslice{(i : \IX)}{\yoneda W}(V', \vfi')$ where $\vfi' : V' \to W$ is obtained by removing $i\psub \vfi$ and $i$ from the domain and codomain respectively.
	Thus, there are no direct shards, and the boundary cells are the ones where $i \psub \vfi$ is an endpoint, i.e.\ $\yoneda W \multip \partial \IX \cong \biguplus_{i = 0}^{a-1} \yoneda W$.
\end{example}
\begin{example} \label{ex:act-elements:cartesian-cubes}
	Continuing \cref{ex:multip:cartesian-cubes} about $a$-ary \emph{cartesian} cubes, let $\Xi = \yoneda W$.
	Then $\Xi \multip \yoneda(i : \IX) \cong \yoneda(W, i : \IX)$.
	Pick $(V, \vfi)$ in the category of elements, again we view $\vfi$ as a morphism $V \to (W, i : \IX)$.
	Then $\vfi$ is directly dimensionally split if $i \psub \vfi$ is not an endpoint, \emph{nor equal to $j \psub \vfi$ for some variable $j$ in $W$}, and in that case $(V, \vfi)$ is isomorphic to $\freshslice{(i : \IX)}{\yoneda W}(V', \vfi')$ where $\vfi' : V' \to W$ is obtained by removing $i\psub \vfi$ and $i$ from the domain and codomain respectively.
	Thus, there are no direct shards, and the boundary cells are the ones where $i \psub \vfi$ is an endpoint or equal to $j \psub \vfi$ for some variable in $W$.
\end{example}

The following (fairly obvious) theorem is paramount to the semantics of transpension elimination (\cref{sec:structure:transp-elim}) and the $\Phi$-rule (\cref{sec:recover:phi}):
\begin{theorem}[Quotient\seelog{} theorem] \label{thm:quotient}
	If a multiplier $\loch \multip U : \catW \to \catW$ is $\top$-slice fully faithful and shard-free (hence presheafwise fully faithful and shard-free), then $\freshslice{U}{\Xi} : \catW/\Xi \to \dimslice{\catW}{(\Xi \multip \yoneda U)}$ is an equivalence of categories. \qed
\end{theorem}

\subsection{\Msys{} Modalities for multipliers} \label{sec:multip}
We are now well-equipped to study the transpension type in a setting with multiple shape variables.
\begin{theorem} \label{thm:multip}
	Any $\top$-slice right adjoint%
	\footnote{Without $\top$-slice right adjointness, we lose the leftmost adjoint functor $\sumpsh{\yoneda U}{\Xi}$ and the leftmost adjoint modality $\modshade{\freshshrt{u}}$.}
	multiplier $\loch \multip U : \catW \to \catW$ and any presheaf $\Xi \in \Psh(\catW)$ give rise to a quadruple of adjoint functors
	\[
		\sumpsh{\yoneda U}{\Xi} \dashv \freshpsh{\yoneda U}{\Xi} \dashv \lollipsh{\yoneda U}{\Xi} \dashv \transppsh{\yoneda U}{\Xi},
	\]
	\[
		\sumpsh{\yoneda U}{\Xi}, \lollipsh{\yoneda U}{\Xi} : \Psh(\catW/\Xi \multip \yoneda U) \to \Psh(\catW/\Xi) \qquad
		\freshpsh{\yoneda U}{\Xi}, \transppsh{\yoneda U}{\Xi} : \Psh(\catW/\Xi) \to \Psh(\catW/\Xi \multip \yoneda U).
	\]
	If $\Xi = \interp \XX$,
	the latter three can be internalized as modalities (with an additional left name) $\modadj{\sumlong{u : \IU}}{\freshlong{u : \IU}} \dashv \modshade{\lollilong{u : \IU}} \dashv \modshade{\transplong{u : \IU}}$ with
	\[
		\interp{\bilock{\sumshrt{u}}{\freshshrt{u}}} = \sumpsh{\yoneda U}{\Xi}, \qquad
		\interp{\modshade{\freshshrt{u}}} =
		\interp{\bilock{\freshshrt{u}}{\lollishrt{u}}} = \freshpsh{\yoneda U}{\Xi}, \qquad
		\interp{\modshade{\lollishrt{u}}} =
		\interp{\bilock{\lollishrt{u}}{\transpshrt{u}}} = \lollipsh{\yoneda U}{\Xi}, \qquad
		\interp{\modshade{\transpshrt{u}}} = \transppsh{\yoneda U}{\Xi}.
	\]
	Overloading some notations from \cref{thm:subst} we denote the units and co-units as
	\begin{align*}
		\copypsh{\yoneda U}{\Xi} &: \id \to \freshpsh{\yoneda U}{\Xi} \circ \sumpsh{\yoneda U}{\Xi}
		& \droppsh{\yoneda U}{\Xi} &: \sumpsh{\yoneda U}{\Xi} \circ \freshpsh{\yoneda U}{\Xi} \to \id \\
		\constpsh{\yoneda U}{\Xi} &: \id \to \lollipsh{\yoneda U}{\Xi} \circ \freshpsh{\yoneda U}{\Xi}
		& \apppsh{\yoneda U}{\Xi} &: \freshpsh{\yoneda U}{\Xi} \circ \lollipsh{\yoneda U}{\Xi} \to \id \\
		\reindexpsh{\yoneda U}{\Xi} &: \id \to \transppsh{\yoneda U}{\Xi} \circ \lollipsh{\yoneda U}{\Xi}
		& \unmeridpsh{\yoneda U}{\Xi} &: \lollipsh{\yoneda U}{\Xi} \circ \transppsh{\yoneda U}{\Xi} \to \id \\
		\keyadj{\dropsym_u}{\constsym_u} &: \modshade{\id} \Rightarrow \modshade{\lollishrt{u} \circ \freshshrt{u}}
		& \keyadj{\copysym_u}{\appsym_u} &: \modshade{\freshshrt{u} \circ \lollishrt{u}} \Rightarrow \modshade{\id} \\
		\keyadj{\appsym_u}{\reindexsym_u} &: \modshade{\id} \Rightarrow \modshade{\transpshrt{u} \circ \lollishrt{u}}
		& \keyadj{\constsym_u}{\unmeridsym_u} &: \modshade{\lollishrt{u} \circ \transpshrt{u}} \Rightarrow \modshade{\id}
	\end{align*}
	where $\reindexsym$ stands for \emph{reindex} and $\unmeridsym$ is the negation of $\meridsym$ which stands for \emph{meridian}.
\end{theorem}
\begin{proof}
	Via left Kan extension, precomposition and right Kan extension \cite{stacks-adjoints}, the pair of adjoint functors $\sumslice{U}{\Xi} \dashv \freshslice{U}{\Xi}$ gives rise to a quadruple of adjoint functors $\sumpsh{\yoneda U}{\Xi} \dashv \freshpsh{\yoneda U}{\Xi} \dashv \lollipsh{\yoneda U}{\Xi} \dashv \transppsh{\yoneda U}{\Xi}$ between the presheaf categories. (For the middle two, we can choose whether we derive them from $\sumslice{U}{\Xi}$ or from $\freshslice{U}{\Xi}$; the resulting functors are naturally isomorphic. We will specify our choice when relevant.)
\end{proof}
\begin{notation}
	Again, due to the (purely sugarous) usage of shape variables, we may end up with variable renamings that are sugar for the identity, e.g.
	\begin{align*}
		&\iskey{\appsym_{(v/u : \IU)}}{\freshlong{v : \IU} \circ \lollilong{u : \IU}}{\renamelong{v : \IU, u := v}} \\
		&\iskey{\reindexsym_{(v/u : \IU)}}{\renamelong{v : \IU, u := v}}{\transplong{v : \IU} \circ \lollilong{u : \IU}}
	\end{align*}
	are exactly the same 2-cells as
	\begin{align*}
		&\iskey{\appsym_{(u : \IU)}}{\freshlong{u : \IU} \circ \lollilong{u : \IU}}{\id} \\
		&\iskey{\reindexsym_{(u : \IU)}}{\id}{\transplong{u : \IU} \circ \lollilong{u : \IU}}.
	\end{align*}
\end{notation}
Whereas \cref{thm:subst} clearly states the meaning of the functors introduced there, little can be said about the meaning of the functors introduced in \cref{thm:multip} without knowing more about the multiplier involved.
The following theorem clarifies the leftmost three functors:
\begin{theorem}[Quantification] \label{thm:quantification}
	If $\loch \multip U$ is
	\begin{enumerate}
		\item $\top$-slice fully faithful, then $\droppsh{\yoneda U}{\Xi}$, $\constpsh{\yoneda U}{\Xi}$ and $\unmeridpsh{\yoneda U}{\Xi}$ are natural isomorphisms.
		\item copointed, then we have
		\begin{enumerate}
			\item $\hidepsh{\yoneda U}{\Xi} : \pairpsh{\yoneda U}{\Xi} \to \sumpsh{\yoneda U}{\Xi}$ (if $\top$-slice right adjoint),
			\item $\spoilpsh{\yoneda U}{\Xi} : \freshpsh{\yoneda U}{\Xi} \to \wknpsh{\yoneda U}{\Xi}$, which can be internalized (if $\top$-slice right adjoint) as $\iskeyadj{\hidesym_u}{\spoilsym_u}{\sumshrt{u}}{\freshshrt{u}}{\pairshrt{u}}{\wknshrt{u}}$,
			\item $\cospoilpsh{\yoneda U}{\Xi} : \funcpsh{\yoneda U}{\Xi} \to \lollipsh{\yoneda U}{\Xi}$, which can be internalized as $\iskeyadj{\spoilsym_u}{\cospoilsym_u}{\wknshrt{u}}{\funcshrt{u}}{\freshshrt{u}}{\lollishrt{u}}$.
		\end{enumerate}
		\item cartesian, then we have:
		\[
			\sumpsh{\yoneda U}{\Xi} = \pairpsh{\yoneda U}{\Xi}, \qquad
			\freshpsh{\yoneda U}{\Xi} = \wknpsh{\yoneda U}{\Xi}, \qquad
			\lollipsh{\yoneda U}{\Xi} = \funcpsh{\yoneda U}{\Xi}.
		\]
		The equalities assume that $\sumbase U : \catW/U \to \catW$ is defined on the nose by $\sumbase U(W, \psi) = W$ and that $\freshpsh{\yoneda U}{\Xi}$ and $\lollipsh{\yoneda U}{\Xi}$ are constructed from $\sumslice{U}{\Xi}$ by precomposition and right Kan extension, respectively.
		Failing this, we only get natural isomorphisms.
	\end{enumerate}
\end{theorem}

\noindent Let us try to interpret this on a more intuitive level:
\begin{enumerate}
	\item For $\top$-slice fully faithful $\IU$, invertibility of $\keyshade{\constsym_u}$ means that any `function' $f$ of type $\Modify{\tf a}{\lollishrt{u}}{\Modify{\tf f}{\freshshrt{u}}{T}}$ is necessarily constant, i.e.\ elements of $\Modify{\tf f}{\freshshrt{u}}{T}$ cannot depend on the shape variable $u$ and are instead \emph{fresh} for $u$.
	With that in mind, invertibility of $\droppsh{\yoneda U}{\Xi}$ indicates that the $\Sigma$-like operation $\sumpsh{\yoneda U}{\Xi}$ hides its first component $u : \IU$.
	Indeed, knowing that the second component of $\sumpsh{\yoneda U}{\Xi} \freshpsh{\yoneda U}{\Xi} \Gamma$ is fresh for $u$, one might expect that $\sumpsh{\yoneda U}{\Xi} \freshpsh{\yoneda U}{\Xi} \Gamma$ behaves somewhat like a non-dependent product (which is exactly what happens in the cartesian setting).
	Instead, $\sumpsh{\yoneda U}{\Xi}$ cancels out $\freshpsh{\yoneda U}{\Xi}$, so apparently if the second component does not depend on $u$, then $u$ is lost altogether.
	Finally function application, which is basically the projection operation from \cref{thm:projmod},
	\begin{equation}
		\appsym_u : (\tymod{\tf f}{\freshshrt{u}}{\Modify{\tf a}{\lollishrt{u}}{T}}) \to T[\ttrans{\appsym_u}{\tf f \twhisk \tf a}{\ttriv}].
	\end{equation}
	requires that the applied function of type $\Modify{\tf a}{\lollishrt{u}}{T}$ be fresh for $u$, as one would expect in linear and affine systems and as we saw in \ruleref{ff:forall:elim} in \cref{fig:ff}.
	
	\item If $\IU$ is copointed (which is perfectly combinable with being $\top$-slice fully faithful), i.e.\ if weakening is allowed for $u : \IU$, then we get three additional operations.
	The 2-cell $\keyshade{\spoilsym_u}$ allows us to forget that something is fresh for $u$.
	This would not be possible without weakening because then the modality $\modshade{\wknshrt{u}}$ expressing potential non-freshness would not even be available.
	The 2-cell $\keyshade{\cospoilsym_u}$ allows us to unnecessarily restrict a function's usage to variables w.r.t. which the function is fresh.
	The semantic substitution $\hidepsh{\yoneda U}{\Xi} : \pairpsh{\yoneda U}{\Xi} \interp \Gamma \to \sumpsh{\yoneda U}{\Xi} \Gamma$ at mode $\Xi$ which is internally available as $\bikey{\hidesym_u}{\spoilsym_u} : (\Gamma, \bilock{\pairshrt{u}}{\wknshrt{u}}) \to (\Gamma, \bilock{\sumshrt{u}}{\freshshrt{u}})$ allows us to hide the first component $u : \IU$. The effect of applying this substitution is effectively a weakening over $u : \IU$, in the sense that the context $\sumpsh{\yoneda U}{\Xi} \Gamma$ can be regarded as not containing the variable $u$ except in the form of a hidden mention necessary to make the dependencies of $\Gamma$ work out.
	Note that for shapes that are not $\top$-slice fully faithful, the words `hidden' and `fresh' need to be taken with a grain of salt. Indeed, for cartesian shapes we should ignore them altogether:
	
	\item If $\IU$ is cartesian, then the leftmost three functors become identical to the ones in \cref{sec:wkn}. In particular, fresh weakening and weakening coincide and the word `fresh' becomes meaningless.
\end{enumerate}
In order to have some general terminology, we will speak of the \textbf{hiding} existential $\lmodshade{\sumshrt{u}}$, \textbf{fresh} weakening $\modshade{\freshshrt{u}}$ and \textbf{substructural} (e.g.\ linear/affine) functions $\modshade{\lollishrt{u}}$ even when the specifics of the multiplier are unclear and it is potentially cartesian.

\subsection{Cartesian multipliers} \label{sec:car-multip}
The cartesian case of the quantification \cref{thm:quantification} may look like all our efforts with multipliers are useless, but let's not forget that there is now a further right adjoint to these well-known functors:
\[
	\pairpsh{\yoneda U}{\Xi} \dashv \wknpsh{\yoneda U}{\Xi} \dashv \funcpsh{\yoneda U}{\Xi} \dashv \transppsh{\yoneda U}{\Xi}.
\]
What we have proven for cartesian shapes is that, for any representable presheaf $\yoneda U \in \Psh(\catW)$ such that cartesian products with $U$ exist in the base category $\catW$ (yielding a cartesian multiplier $\loch \times U : \catW \to \catW$), there is a right adjoint to the $\Pi$-type!
Licata et al.\ \cite{internal-universes} have used a right adjoint to the non-dependent function type functor $(\yoneda U \to \loch) = \funcpsh{\yoneda U}{\Xi} \circ \wknpsh{\yoneda U}{\Xi}$, called the \emph{amazing right adjoint} and necessarily given by $(\yoneda U \amaze \loch) = \funcpsh{\yoneda U}{\Xi} \circ \transppsh{\yoneda U}{\Xi}$, but the current result appears stronger.

Remarkably, it is not, and it is not novel either.
As conjectured by Lawvere, proven by Freyd and published by Yetter \cite{yetter}, for an arbitrary object $\IU$ in an arbitrary topos, the transpension functor (there unnamed and denoted $\nabla$) over $\IU$ exists if the amazing right adjoint exists. Indeed, in that case it can be constructed by the following pullback:
\begin{equation*}
	\begin{array}{c}
		\xymatrix{
			\pairlong{u : \IU}.\transpshrt{u} \codot T
				\ar[r]
				\ar[d]_{\fst}^(0.25){\lrcorner}
			& \IU \amaze (\Sigma(P : \Prop).(P \to T))
				\ar[d]^{\IU \amaze \fst}
			\\
			\IU \ar[r]_{(\lambda f . \idtp{}{f}{\idfun_\IU})^\top}
			& \IU \amaze \Prop
		}
	\end{array}
\end{equation*}
where $g^\top$ denotes the transpose of $g$ under $(\IU \to \loch) \dashv (\IU \amaze \loch)$.

\subsection{Further reading} \label{sec:techreport}
We refer to the technical report \cite{transpension-techreport} for more information on
\begin{itemize}
	\item composite multipliers $\loch \multip (U \multip U') := (\loch \multip U) \multip U'$,
	\item morphisms of multipliers $\loch \multip \upsilon : \loch \multip U \to \loch \multip U'$ (together with the previous point one could formalize the exchange rule),
	\item acting on slice objects as opposed to acting on elements (\cref{sec:act-elements}),
	\item properties of $\loch \multip \yoneda U : \Psh(\catW) \to \Psh(\catW)$, the left Kan extension along $\loch \multip U$, again viewed as a multiplier for $\yoneda U$,
	\item non-endo multipliers $\loch \multip U : \catW \to \cat V$,
	\item rules for commuting (co)quantifiers for multipliers, (co)quantifiers for substitution, and (when adding the transpension type to an already modal type system) prior modalities.
\end{itemize}
  
\section{The Fully Faithful Transpension System (\FFsys{}) Revisited} \label{sec:comparison}
In \cref{sec:comparison:embedding}, we give a pseudo-embedding of \FFsys{} (\cref{sec:ff}) into \Msys{} instantiated on a $\top$-slice fully faithful shape $\IU$.
In \cref{sec:comparison:transposition,sec:comparison:hdpm}, we revisit the results about internal transposition and higher-dimensional pattern-matching from \cref{sec:ff:transposition,sec:ff:hdpm}, as these also work for other shapes.
Poles (\cref{sec:ff:poles}) will be revisited in \cref{sec:structure:poles}.

\subsection{Pseudo-Embedding of \FFsys{} in \Msys{}} \label{sec:comparison:embedding}
We give a pseudo-embedding of \FFsys{} into \Msys{} instantiated on a $\top$-slice fully faithful shape $\IU$.
Pseudo, in the sense that \ruleref{ff:ctx-forall:nil} will be only an isomorphism and some commutation properties w.r.t.\ shape substitution will only hold up to isomorphism,
implying that a few other rules will need some adjustments before their translation is well-typed.
We do not pay too much attention to those matters: the purpose of \cref{sec:ff} was didactical and the purpose of the embedding is to show that it was also morally correct.

\subsubsection{Metatype of the embedding}
The judgement forms are translated as follows:
\begin{itemize}
	\item A context $\Gamma \ctx$ is translated as a pair consisting of a shape context $\accol \Gamma \shpctx$ listing the shape variables in $\Gamma$ and an internal context $\accol \Gamma \sep \angles \Gamma \ctx$.
	\item A substitution $\sigma : \Delta \to \Gamma$ is translated as a pair consisting of a shape substitution $\accol \sigma : \interp{\accol \Gamma} \to \interp{\accol \Delta}$ and an internal substitution $\accol \Gamma \sep \angles \sigma : \angles \Gamma \to (\angles \Delta, \bilock{\wknshrt{\accol{\sigma}}}{\funcshrt{\accol{\sigma}}})$.%
	\footnote{Note that the latter is equivalent to an internal substitution $\accol{\Delta} \sep \angles{\sigma}' : (\angles \Gamma, \bilock{\pairshrt{\accol{\sigma}}}{\wknshrt{\accol{\sigma}}}) \to \angles \Delta$, but the advantage of $\bilock{\wknshrt{\accol{\sigma}}}{\funcshrt{\accol{\sigma}}}$ is that it is strictly functorial.}
	We point out that if $\accol \sigma \neq \id$, then translating $t[\sigma] : T[\sigma]$ is far from trivial. Since \ruleref{ff:ctx-shp:wkn} is the only source of such substitutions, we will generally assume that $\accol{\sigma} = \id$ and not worry about the general case.
	Then, we have $\accol{\Gamma} = \accol{\Delta} \sep \angles \sigma : \angles \Gamma \to \angles \Delta$.
	\item A type $\Gamma \sez T \type$ is translated to $\accol \Gamma \sep \angles \Gamma \sez \angles T \type$.
	\item A term $\Gamma \sez t : T$ is translated to $\accol \Gamma \sep \angles \Gamma \sez \angles t : \angles T$.
\end{itemize}

\subsubsection{Structural rules of MLTT}
The shape interpretation $\accol \loch$ ignores non-shape variables, and the internal interpretation $\angles \loch$ respects the structural rules of MLTT:
\begin{align*}
	\accol{\cdot} &= {\cdot}
	& \angles{\cdot} &= {\cdot} \\
	\accol{\Gamma, x : A} &= \accol{\Gamma}
	& \angles{\Gamma, x : A} &= \angles{\Gamma}, x : \angles{A} \\
	\accol{\id, t/x} &= \id
	& \angles{t/x} &= \angles{t}/x \\
	\accol{x/\novar} &= \id
	& \angles{x/\novar} &= x/\novar \\
	&& \angles{x} &= x
\end{align*}

\subsubsection{Linear/affine shape variables} \label{sec:comparison:embedding:shp}
The shape interpretation $\accol \loch$ retains shape variables.
In the fully faithful system, if a variable occurred to the left of $u : \IU$, this meant that it was fresh for $u$.
In \Msys{}, the shape variable $u : \IU$ is of course added to the shape context, so we cannot use its position to signal which variables in $\Gamma$ are and are not fresh for $u$.
Instead, we keep track of this using the fresh weakening operation on contexts $\bilock{\freshshrt{u}}{\lollishrt{u}}$.
\begin{align*}
	\accol{\Gamma, u : \IU} &= \accol{\Gamma}, u : \IU
	& \angles{\Gamma, u : \IU} &= \angles{\Gamma}, \bilock{\freshshrt{u}}{\lollishrt{u}}
	& (\ruleref{ff:ctx-shp}) \\
	\accol{\sigma, u/u} &= \accol{\sigma}, u/u
	& \angles{\sigma, u/u} &= \angles{\sigma}, \bilock{\freshshrt{u}}{\lollishrt{u}}
	& (\ruleref{ff:ctx-shp:fmap}) \\
	\accol{\sigma, u/\novar} &= \accol{\sigma}, u/\novar
	& \angles{\sigma, u/\novar} &= \angles{\sigma}, \bikeytyped{\spoilsym_u}{\cospoilsym_u}{\wknshrt{u}}{\funcshrt{u}}{\freshshrt{u}}{\lollishrt{u}}
	& (\ruleref{ff:ctx-shp:wkn})
\end{align*}

\subsubsection{Linear/affine function type}
The $\forall$-type translates to a modal type:
\begin{align*}
	\inferencedeadl{\ruleref{ff:forall}}{
		\accol \Gamma, u : \IU \sep \angles \Gamma, \bilock{\freshshrt{u}}{\lollishrt{u}} \sez \angles A \type
	}{
		\accol \Gamma \sep \angles \Gamma \sez \Modifynovar{\lollishrt{u}}{\angles A} \type
	}{}
	\qquad
	\inferencedeadl{\ruleref{ff:forall:intro}}{
		\accol \Gamma, u : \IU \sep \angles \Gamma, \bilock{\freshshrt{u}}{\lollishrt{u}} \sez \angles a : \angles A
	}{
		\accol \Gamma \sep \angles \Gamma \sez \modifynovar{\lollishrt{u}}{\angles a} : \Modifynovar{\lollishrt{u}}{\angles A}
	}{}
\end{align*}
Application is translated using \cref{thm:projmod}.
Let $\Theta = \Gamma, u : \IU, \delta : \Delta$ with no shape variables in $\Delta$.
Then $\accol \Theta = \accol \Gamma, u : \IU$ and $\angles{\Theta} = \angles{\Gamma}, \bilock{\freshshrt{u}}{\lollishrt{u}}, \angles \Delta$.
\begin{align*}
	\inferencedeadl{\ruleref{ff:forall:elim}}{
	\inference{
	\inference{
		\accol{\Gamma} \sep \angles{\Gamma} \sez f : \Modifynovar{\lollishrt{u}}{\angles{A}}
	}{
		\accol{\Gamma} \sep \angles{\Gamma}, \bilock{\sumshrt{u} \circ \freshshrt{u}}{\lollishrt{u} \circ \freshshrt{u}} \sez f[\bikey{\dropsym_u}{\constsym_u}] : \Modifynovar{\lollishrt{u}}{\angles{A}[\bikey{\dropsym_u}{\constsym_u}, \bilock{\freshshrt{u}}{\lollishrt{u}}]}
	}{}
	}{
		\accol{\Gamma} \sep \angles{\Gamma}, \bilock{\freshshrt{u}}{\lollishrt{u}}, \angles{\Delta}, \bilock{\sumshrt{u}}{\freshshrt{u}} \sez f[\bikey{\dropsym_u}{\constsym_u}] : \Modifynovar{\lollishrt{u}}{\angles{A}[\bikey{\dropsym_u}{\constsym_u}, \bilock{\freshshrt{u}}{\lollishrt{u}}]}
	}{}
	}{
		\accol{\Gamma}, u : \IU \sep \angles{\Gamma}, \bilock{\freshshrt{u}}{\lollishrt{u}}, \angles{\Delta} \sez \appsym_u \modappnovar{\freshshrt{u}} f[\bikey{\dropsym_u}{\constsym_u}] : \angles{A}[\bikey{\dropsym_u}{\constsym_u}, \bilock{\freshshrt{u}}{\lollishrt{u}}][\bilock{\freshshrt{u}}{\lollishrt{u}}, \bikey{\copysym_u}{\appsym_u}] = \angles{A}
	}{}
\end{align*}
Observe that, lacking existential quantification for contexts in \cref{sec:ff}, the application rule \ruleref{ff:forall:elim} simply discarded the non-fresh part $\Delta$.
In the target language, variables under $\bilock{\sumshrt{u}}{\freshshrt{u}}$ can only be used if they are annotated with a modality $\modshade \mu$ from which there is a 2-cell $\iskey \alpha \mu {\freshshrt{u}}$, i.e.\ if they are fresh for $u$.
However, recall that the word `fresh' needs to be taken with a grain of salt when we are \emph{not} dealing with a $\top$-slice fully faithful shape. For example, if $\IU$ is cartesian, then $\bilock{\sumshrt{u}}{\freshshrt{u}} = \bilock{\pairshrt{u}}{\wknshrt{u}}$, so that the aggregation of shape and type context in the premise and conclusion of the $\appsym_u$-rule
\begin{align*}
	\inference{
		\XX, \sep \Gamma, \bilock{\pairshrt{u}}{\wknshrt{u}} \sez f : \Modifynovar{\funcshrt{u}}{A}
	}{
		\XX, u : \IU \sep \Gamma \sez \appsym_u \modappnovar{\wknshrt{u}} f : A[\bikey{\copysym_u}{\appsym_u}]
	}{}
\end{align*}
are isomorphic: $\interp \XX.(\Sigma (\yoneda U).\interp \Gamma) \cong (\interp \XX \times \yoneda U).\interp{\Gamma}$.

\subsubsection{Telescope quantification}
Let $\Theta = (\Gamma, u : \IU, \delta : \Delta)$ with no shape variables in $\Delta$.
Then $[\lollishrt{u}] \Theta = (\Gamma, \lollishrt{u}.(\delta : \Delta))$.
We translate this as follows:
\begin{align*}
	\accol{[\lollishrt{u}] \Theta} &= \accol{\Gamma, \lollishrt{u}.(\delta : \Delta)} = \accol{\Gamma}
	& \angles{[\lollishrt{u}] \Theta} &= \angles{\Theta}, \bilock{\lollishrt{u}}{\transpshrt{u}}
	& (\ruleref{ff:ctx-forall})
\end{align*}
Let $\rho = (\sigma, u/u, \tau/\delta') : \Theta = (\Gamma, u : \IU, \delta : \Delta) \to \Theta' = (\Gamma', u : \IU, \delta' : \Delta')$.
Then $[\lollilong{u/u}]\rho = (\sigma, \lambdabar u.\tau/\lambdabar u.\delta')$.
We translate this as follows:
\begin{align*}
	\accol{[\lollilong{u/u}]\rho} &= \accol{\sigma, \lambdabar u.\tau/\lambdabar u.\delta'} = \accol{\sigma}
	& \angles{[\lollilong{u/u}]\rho} &= \angles{\rho}, \bilock{\lollishrt{u}}{\transpshrt{u}}
	& (\ruleref{ff:ctx-forall:fmap})
\end{align*}
The rule \ruleref{ff:ctx-forall:nil} concerns $(\Gamma, \lollishrt{u}.()) = [\lollishrt{u}](\Gamma, u : \IU)$ which translates to $\accol \Gamma \sep \angles \Gamma, \bilock{\freshshrt{u}}{\lollishrt{u}}, \bilock{\lollishrt{u}}{\transpshrt{u}} \ctx$, which is isomorphic to $\accol \Gamma \sep \angles \Gamma \ctx$ by the 2-cell $(\id_\Gamma, \bikeytyped{\constsym_u}{\unmeridsym_u}{\lollishrt{u} \circ \freshshrt{u}}{\lollishrt{u} \circ \transpshrt{u}}{\id}{\id})$ because $\IU$ is $\top$-slice fully faithful (quantification \cref{thm:quantification}).
Naturality of $\bikey{\constsym_u}{\unmeridsym_u}$ models \ruleref{ff:ctx-forall:fmap:nil}.

\subsubsection{Telescope application}
Let $\Theta = (\Gamma, u : \IU, \delta : \Delta)$ with no shape variables in $\Delta$.
Then $\appsym_\Theta = (u/u, (\lambdabar u.\delta)\,u/\delta) : ([\lollishrt{u}]\Theta, u : \IU) = (\Gamma, \lollishrt{u}.(\delta : \Delta), u : \IU) \to \Theta$.
We translate this using the 2-cell
\begin{align*}
	&\accol{\Gamma}, u : \IU \sep \angles{\appsym_\Theta} = (\bikeytyped{\appsym_u}{\reindexsym_u}{\id}{\id}{\freshshrt{u} \circ \lollishrt{u}}{\transpshrt{u} \circ \lollishrt{u}}) : (\angles{\Theta}, \bilock{\lollishrt{u}}{\transpshrt{u}}, \bilock{\freshshrt{u}}{\lollishrt{u}}) \to \angles \Theta
	&(\ruleref{ff:ctx-app})
\end{align*}
naturality of which models \ruleref{ff:ctx-app:nat}.

If $\Delta$ is empty, then $\angles{\Theta} = (\angles{\Gamma}, \bilock{\freshshrt{u}}{\lollishrt{u}})$, so that
\begin{align*}
	\accol{\Gamma}, u : \IU &\sep \angles{\appsym_{(\Gamma, u : \IU)}} = (\bilock{\freshshrt{u}}{\lollishrt{u}}, \bikey{\appsym_u}{\reindexsym_u}) = (\bikey{\constsym_u\inv}{\unmeridsym_u\inv} , \bilock{\freshshrt{u}}{\lollishrt{u}}) \\
	&: (\angles{\Gamma}, \bilock{\freshshrt{u}}{\lollishrt{u}}, \bilock{\lollishrt{u}}{\transpshrt{u}}, \bilock{\freshshrt{u}}{\lollishrt{u}}) \to (\angles{\Gamma}, \bilock{\freshshrt{u}}{\lollishrt{u}})
\end{align*}
Now $\bikey{\constsym_u}{\unmeridsym_u}$ models the equation in \ruleref{ff:ctx-forall:nil} which well-typedness of \ruleref{ff:ctx-app:nil} relies on, so it can be regarded as the identity and the above
essentially models \ruleref{ff:ctx-app:nil}.

Similarly, the rule \ruleref{ff:ctx-forall:fmap:ctx-app} concerns $[\lollishrt{u/u}]\appsym_\Theta$, which translates to
\begin{align*}
	\accol{\Gamma} &\sep (\angles{\appsym_\Theta}, \bilock{\lollishrt{u}}{\transpshrt{u}}) = (\bikey{\appsym_u}{\reindexsym_u}, \bilock{\lollishrt{u}}{\transpshrt{u}}) = (\bilock{\lollishrt{u}}{\transpshrt{u}}, \bikey{\constsym_u\inv}{\unmeridsym_u\inv}) \\
	&: (\angles{\Theta}, \bilock{\lollishrt{u}}{\transpshrt{u}}, \bilock{\freshshrt{u}}{\lollishrt{u}}, \bilock{\lollishrt{u}}{\transpshrt{u}}) \to (\angles \Theta, \bilock{\lollishrt{u}}{\transpshrt{u}}),
\end{align*}
essentially modelling \ruleref{ff:ctx-forall:fmap:ctx-app}.

\subsubsection{Transpension type}
Let $\Theta = (\Gamma, u : \IU, \delta : \Delta)$ with no shape variables in $\Delta$.
The transpension type translates to a modal type:
\begin{align*}
	\inferencedeadl{\ruleref{ff:transp}}{
		\accol \Gamma \sep \angles \Theta, \bilock{\lollishrt{u}}{\transpshrt{u}} \sez \angles A \type
	}{
		\accol \Gamma, u : \IU \sep \angles \Theta \sez \Modifynovar{\transpshrt{u}}{\angles A} \type
	}{}
	\qquad
	\inferencedeadl{\ruleref{ff:transp:intro}}{
		\accol \Gamma \sep \angles \Theta, \bilock{\lollishrt{u}}{\transpshrt{u}} \sez \angles a : \angles A
	}{
		\accol \Gamma, u : \IU \sep \angles \Theta \sez \modifynovar{\transpshrt{u}}{a} : \Modifynovar{\transpshrt{u}}{\angles A}
	}{}
\end{align*}
The eliminator is translated using \cref{thm:projmod}.
\begin{align*}
	\inferencedeadl{\ruleref{ff:transp:elim}}{
		\accol{\Gamma}, u : \IU \sep \angles{\Gamma}, \bilock{\freshshrt{u}}{\lollishrt{u}} \sez t : \Modifynovar{\lollishrt{u}}{\angles{A}[\bikey{\constsym_u\inv}{\unmeridsym_u\inv}]}
	}{
		\accol{\Gamma} \sep \angles{\Gamma} \sez \unmeridsym_u \modappnovar{\lollishrt{u}} t : \angles{A}
	}{}
\end{align*}
Again $\bikey{\constsym_u\inv}{\unmeridsym_u\inv}$ just models \ruleref{ff:ctx-forall:nil} and can be ignored.
The $\beta$- and $\eta$-rules are the ones form \cref{thm:projmod},
and the naturality rules amount to
\begin{equation*}
	\paren{\modifynovar{\lollishrt{u}}{\angles{a}}}[\angles \rho]
	= \modifynovar{\lollishrt{u}}{\paren{\angles{a}\brac{\angles \rho, \bilock{\lollishrt{u}}{\transpshrt{u}}}}}, \qquad
	\Modifynovar{\lollishrt{u}}{\angles{A}}[\angles \rho]
	= \Modifynovar{\lollishrt{u}}{\angles{A}\brac{\angles \rho, \bilock{\lollishrt{u}}{\transpshrt{u}}}}.
\end{equation*}
This concludes the embedding for the selected typing rules in \cref{fig:ff}.

\subsection{Internal transposition} \label{sec:comparison:transposition}
In \cref{sec:ff:transposition}, we proved the isomorphism
\begin{equation}
	(\lollilong{u : \IU}.A) \to B
	\quad \cong \quad
	\lollilong{u : \IU}.(A \to \transpshrt{u} \codot B). \label{eq:comparison:can-aff-transposition}
\end{equation}
The \Msys{} equivalent,
\[
	(\tymod{\tf a}{\lollishrt{u}}{A}) \to B'
	\quad \cong \quad
	\Modify{\tf a}{\lollishrt{u}}{A \to \Modify{\tf t}{\transpshrt{u}}{B'[\bikey{\constsym_u\inv}{\unmeridsym_u\inv}]}}
\]
is not in general true (or even statable) but for $\top$-slice fully faithful multipliers it follows from the following instance of \cref{thm:transpose} (which holds in general)
\begin{equation}
	\Modify{\tf t}{\transpshrt{u}}{(\tymod{\tf a}{\lollishrt{u}}{A[\bikey{\appsym_u}{\reindexsym_u}]}) \to B}
	\quad \cong \quad
	\paren{A \to \Modify{\tf t}{\transpshrt{u}}{B}} \label{eq:comparison:transposition}
\end{equation}
by applying $\Modifynovar{\lollishrt{u}}{\loch}$ to both sides, redefining $B' = B[\bikey{\constsym_u}{\unmeridsym_u}]$ and using the quantification \cref{thm:quantification}.

\subsection{Higher-dimensional pattern matching} \label{sec:comparison:hdpm}
Deriving HDPM from internal transposition for general multipliers is a bit more involved than it was for \FFsys{} (\cref{sec:ff:hdpm}) because we have to use \cref{eq:comparison:transposition} instead of \cref{eq:comparison:can-aff-transposition}.
However, we can construct an isomorphism $i : \Modify{\tf a}{\lollishrt{u}}{A \uplus B} \cong \Modify{\tf a}{\lollishrt{u}}{A} \uplus \Modify{\tf a}{\lollishrt{u}}{B}$ directly by translating from \cref{sec:ff:hdpm}.
Again, the map to the left is trivial by pattern matching (which is still the original eliminator for modal types!).
The map to the right is given in either system by:
\begin{align*}
	&i : (\lollishrt{u}.A \uplus B) \to (\lollishrt{u}.A) \uplus (\lollishrt{u}.B) \qquad \qquad \qquad \qquad \framebox{\FFsys} \\
	&i\,\hat c = \unmeridsym\,\paren{  u.\case{\hat c\,u}{
		\inl\,a &\mapsto& \meridshrt{u} \codot (\inl\,(\lambdabar u.a)) \\
		\inr\,b &\mapsto& \meridshrt{u} \codot (\inr\,(\lambdabar u.b))
	}  } \\
	&i : \Modify{\tf a}{\lollishrt{u}}{A \uplus B} \to \Modify{\tf a}{\lollishrt{u}}{A} \uplus \Modify{\tf a}{\lollishrt{u}}{B} \qquad \qquad \qquad \,\, \framebox{\Msys{}} \\
	&i\,\hat c =
		\unmeridsym_u \modappnovar{\lollishrt{u}}{
			\case{(\appsym_u \modappnovar{\freshshrt{u}}{\hat c \varbikey{\dropsym_u}{\constsym_u}})}{
				\inl\,a \mapsto \modify{\tf t}{\transpshrt{u}}{(\inl\,(\modify{\tf a'}{\lollishrt{u}}{(a \varbikey{\appsym_u}{\reindexsym_u})}))} \\
				\inr\,b \mapsto \modify{\tf t}{\transpshrt{u}}{(\inr\,(\modify{\tf a'}{\lollishrt{u}}{(b \varbikey{\appsym_u}{\reindexsym_u})}))}
			}
		}.
\end{align*}

\section{Additional Typing Rules} \label{sec:add}
In this section, we add a few extensions to \Msys{} in order to reason about boundaries, and in order to recover all known presheaf operators in \cref{sec:recover}.

\subsection{Subobject classifier}
We add a universe of propositions (semantically the subobject classifier) $\Prop : \uni 0$, with implicit encoding and decoding operations \`a la Coquand.
This universe is closed under logical operators and weak DRAs \cite[\S 6.5]{nuyts-phd}.
This is necessary to talk about $\Psi$ and $\Phi$.
We identify all proofs of the same proposition.

\subsection{Boundary predicate} \label{sec:add:boundary}
We add the following shape context constructor:
\begin{equation*}
	\inferencel{shp-ctx-boundary}{
		\XX \shpctx \qquad
		\IU \shape
	}{
		\XX, u : \partial \IU \shpctx
	}{}
\end{equation*}
modelling $\interp{\XX, u : \partial \IU} = \interp{\XX} \multip \partial U$ (\cref{def:dir-boundary}).
Write $(u \in \partial \IU)$ for the presheaf morphism that includes $\interp{\XX, u : \partial \IU}$ in $\interp{\XX, u : \IU}$.
We add a predicate of the same name $\XX, u : \IU \sep \cdot \sez u \in \partial \IU : \Prop$ corresponding in the model to this subobject $\interp{\XX, u : \partial \IU} \subseteq \interp{\XX, u : \IU}$.
Note that, since the direct boundary was \emph{not} defined by pullback, the boundary predicate is not preserved by shape substitution $\sigma : \interp{\XX_1} \to \interp{\XX_2}$, i.e.\ $\Modify{\tf o}{\wknlong{\sigma, u := u}}{(u \in \partial \IU)_{\XX_2}}$ is not in general isomorphic to $(u \in \partial \IU)_{\XX_1}$.

If we had modal type formers for \emph{left} adjoints, then we could define the boundary predicate as $(\top, \bilock{\pairlong{u \in \partial \IU}}{\wknlong{u \in \partial \IU}})$.
However, \MTT{} does not support such type formers and we do not know how to do this%
\footnote{It is worth noting that $\pairpsh{}{\sigma}$ is a parametric right adjoint so the work by Gratzer et al.\ \cite{gratzer-pra} could be relevant.}
so we simply axiomatize the predicate by decreeing for every type $\XX, u : \IU \sep \Gamma \sez A \type$ an isomorphism
\begin{align}
	(u \in \partial \IU) \to A &\quad \cong \quad \Modify{\tf p}{\funclong{u \in \partial \IU}}{\Modify{\tf o}{\wknlong{u \in \partial \IU}}{A[\bikey{\droplong{(u \in \partial\IU)}}{\constlong{(u \in \partial\IU)}}]}}. \label{eq:boundary}
\end{align}
In practice, for concrete systems, we will want axioms based on our findings in \cref{sec:examples}, e.g.\ in a binary cubical system we would decree ${\cdot} \sep (i \in \partial \IX) \leftrightarrow (\idtp{\IX}{i}{0}) \vee (\idtp{\IX}{i}{1})$ where the latter two predicates could be axiomatized similarly to \eqref{eq:boundary}.

\subsection{Strictness axiom} \label{sec:strictness}
The strictness axiom \cite{orton-pitts-axioms} allows to extend a partial type $T$ to a total type if $T$ is isomorphic to a total type $A$, effectively strictifying the isomorphism:
\begin{equation*}
	\inference{
		\XX \sep \Gamma \sez \vfi : \Prop \qquad
		\XX \sep \Gamma \sez A : \uni \ell \qquad
		\XX \sep \Gamma, \_ : \vfi \sez T : \uni \ell \qquad
		\XX \sep \Gamma, \_ : \vfi \sez i : A \cong T \\
	}{
		\XX \sep \Gamma \sez \StrictNice{A}{\tysysclause{\vfi}{T}{i}} : \uni \ell \qquad
		\XX \sep \Gamma \sez \strictNice{\sysclause{\vfi}{i}} : A \cong \StrictNice{A}{\tysysclause{\vfi}{T}{i}} \\
		\infwhere \Gamma, \_ : \vfi \sez \StrictNice{A}{\tysysclause{\vfi}{T}{i}} = T : \uni \ell \qquad
		\Gamma, \_ : \vfi \sez \strictNice{\sysclause{\vfi}{i}} = i : A \cong T
	}{}
\end{equation*}

\section{Investigating the Transpension Type}\label{sec:structure}
In \cref{sec:ff:poles}, we have briefly investigated the structure of a fully faithful transpension type in \FFsys{}.
In this section, we investigate the structure of the general transpension type $\Modify{\tf t}{\transpshrt{u}}{A}$ in \Msys{}.

\subsection{Poles} \label{sec:structure:poles}
Our first observation is that on the boundary, the transpension type is trivial.
Let $\ismod{\top}{\XX_1}{\XX_2}$ be the modality, between any two modes, which maps any presheaf to the terminal presheaf.
We clearly have $\modshade{\top \circ \mu} = \modshade{\top}$ for any $\modshade \mu$, but also $\modshade{\mu \circ \top} \cong \modshade{\top}$ because all internal modalities are right adjoints and therefore preserve the terminal object.
\begin{theorem}[Pole] \label{thm:pole}
	We have $\modshade{\wknlong{u \in \partial \IU} \circ \transpshrt{u}} \cong \modshade{\top}$.
	We can thus postulate a term $\XX, u : \IU \sep \Gamma, \_ : u \in \partial \IU \sez \pole : \Modify{\tf t}{\transpshrt{u}}{T}$ for any $\XX \sep \Gamma, \lock{\tf t}{\transpshrt{u}} \sez T \type$, with an $\eta$-rule $\XX, u : \IU \sep \Gamma, \_ : u \in \partial \IU \sez t = \pole : \Modify{\tf t}{\transpshrt{u}}{T}$.
\end{theorem}
\begin{proof}[Sketch of proof]
	The left adjoints $\lmodshade{\lollilong{u : \IU} \circ \pairlong{u \in \partial \IU}}$ and $\lmodshade{\bot}$ of the concerned modalities are isomorphic because $\lollilong{u : \IU}.(u \in \partial \IU)$ is false.
	We give a full proof in the technical report \cite{transpension-techreport}.
\end{proof}
\Cref{def:boundary} of the boundary relied on the notion of dimensional splitness.
The following result shows that it was a good one: the transpension is \emph{only} trivial on the boundary:
\begin{theorem}[Boundary] \label{thm:boundary}
	In the model, we have \cite{transpension-techreport}
	\[
		\XX, u : \IU \sep \Gamma \sez (u \in \partial \IU) \cong \Modify{\tf t}{\transplong{u : \IU}}{\Empty}.
	\]
\end{theorem}

\subsection{Meridians}
As all our modalities are proper DRAs \cite{dra}, the modal introduction rule is invertible in the model.
This immediately shows that sections%
\footnote{By a section of a dependent type, we mean a dependent function with the same domain as the type.}
of the transpension type
\begin{align*}
	&\XX \sep \Gamma \sez f : \Modify{\tf a}{\lollilong{u : \IU}}{\Modify{\tf t}{\transplong{u : \IU}}{T}}
\end{align*}
(which we call meridians) are in 1-1 correspondence with terms
\begin{align*}
	&\XX \sep \Gamma, \bilock{\lollishrt{u} \circ \freshshrt{u}}{\lollishrt{u} \circ \transpshrt{u}} \sez t : T.
\end{align*}
If it were not for the locking of the context, this characterization in terms of poles and meridians would make the transpension type look quite similar to a dependent version of the suspension type in HoTT \cite{hottbook}, whence our choice of name.
If $\IU$ is $\top$-slice (hence presheafwise) fully faithful, then the applied locks are actually isomorphic to the identity lock (\cref{thm:quantification}).
In any case, regardless of the properties of $\IU$, \cref{thm:projmod} tells us that the $\name{let}$-rule for $\modshade{\transplong{u : \IU}}$ has the same power as
\begin{align*}
	\unmeridsym_u &: (
	\tymod{\tf a}{\lollilong{u : \IU}}{\Modify{\tf t}{\transplong{u : \IU}}{T}}
) \to T[\bikey{\constsym_u}{\unmeridsym_u}]
\end{align*}
which extracts meridians.
If $\IU$ is $\top$-slice fully faithful, then the 2-cell $\keyshade{\unmeridsym_u}$ is invertible (\cref{thm:quantification}) and we can also straightforwardly \emph{create} meridians from elements of $T[\bikey{\constsym_u}{\unmeridsym_u}]$.

\subsection{Pattern matching} \label{sec:structure:transp-elim}
The eliminator $\unmeridsym_u$ is only capable of eliminating \emph{sections} of the transpension type.
If the quotient \cref{thm:quotient} applies to $\IU$, we can eliminate locally by pattern matching:%
\begin{figure}
	\begin{equation*}
		\inferencel{transp:elim}{
			\XX, u : \IU \sep \Gamma \ctx \\
			\XX \sep \Gamma, \bilock{\lollishrt{u}}{\transpshrt{u}} \sez A \type \\
			\XX, u : \IU  \sep \Gamma, r : \Modify{\tf t}{\transpshrt{u}}{A} \sez C \type \\
			\XX, u : \IU \sep \Gamma, \_ : u \in \partial \IU \sez c_\pole : C[\pole/r] \\
			\XX, u : \IU \sep \Gamma, \bilock{\lollishrt{u}}{\transpshrt{u}}, x : A, \bilock{\freshshrt{u}}{\lollishrt{u}} \sez c_\merid : C
			\brac{
				\idsub_\Gamma, \bikey{\appsym_u}{\reindexsym_u},
				\modify{\tf t'}{\transpshrt{u}}{(x \varbikey{\constsym_u\inv}{\unmeridsym_u\inv})}/r
			} \\
			\XX, u : \IU \sep \Gamma, \bilock{\lollishrt{u}}{\transpshrt{u}}, x : A, \bilock{\freshshrt{u}}{\lollishrt{u}}, \_ : u \in \partial \IU \sez c_\merid = c_\pole[\idsub_\Gamma, \bikey{\appsym_u}{\reindexsym_u}] : C[\idsub_\Gamma, \bikey{\appsym_u}{\reindexsym_u}, \pole/r] \\
			\XX, u : \IU \sep \Gamma \sez t : \Modify{\tf t}{\transpshrt{u}}{A}
		}{
			\XX, u : \IU \sep \Gamma \sez c := \case{t}{
				\pole \mapsto c_\pole
				\sep
				\merid\,x \mapsto c_\merid
			} : C[t/r] \\
			\infwhere c[\pole/t] = c_\pole \\
			\infnowhere \XX, u : \IU \sep \Delta, \bilock{\freshshrt{u}}{\lollishrt{u}} \sez c =
			c_\merid[\idsub_\Delta, \bikey{\constsym_u}{\unmeridsym_u}, \unmeridsym_u \modappnovar{\lollishrt{u}} t/x, \bilock{\freshshrt{u}}{\lollishrt{u}}] : C
		}{}
	\end{equation*}
	\caption[Transpension: Pattern matching]{Transpension elimination by pattern matching (sound if $\IU$ is $\top$-slice fully faithful and shard-free). Recall \cref{eq:variance-compose}.}
	\label{fig:transp:elim}
\end{figure}
\begin{theorem} \label{thm:transp-elim}
	If $\IU$ is $\top$-slice (hence presheafwise) fully faithful and shard-free, then the rule \ref{rule:transp:elim} in \cref{fig:transp:elim} is sound \cite{transpension-techreport}.
\end{theorem}
The elimination rule is best understood by looking at the left names.
We get a context $\Gamma$ depending on $u : \IU$, a type $A$ depending on sections of $\Gamma$ (as represented by $\Gamma, \bilock{\lollishrt{u}}{\transpshrt{u}}$), a type $C$ depending on $u$ and $r : \Modify{\tf a}{\transpshrt{u}}{A}$, and an argument $t$ of type $\Modify{\tf a}{\transpshrt{u}}{A}$.
To obtain a value of type $C$, we need to give an action $c_\pole$ on the boundary, where $t$ is necessarily $\pole$ (pole \cref{thm:pole}), and a compatible action on sections of the transpension type, i.e.\ meridians, which live over sections of $\Gamma$ but are themselves essentially elements of $A$ (quantification \cref{thm:quantification}), producing sections of $C$ (but the quantifier $\modshade{\lollishrt{u}}$ has been brought to the left as $\bilock{\freshshrt{u}}{\lollishrt{u}}$).
Thanks to shard-freedom, we know that everything that is not a section\footnote{or a `dimensional section' in case of base categories that are not objectwise pointable}, is on the boundary, so this suffices.

The computation rule for meridians fires when all of $\Gamma$ is fresh for $u$. In this situation, the judgement for $t$ is $\XX, u : \IU \sep \Delta, \bilock{\freshshrt{u}}{\lollishrt{u}} \sez t : \Modify{\tf a}{\transpshrt{u}}{A}$ which by transposition boils down to $\XX \sep \Delta \sez t' : \Modify{\tf f}{\lollishrt{a}}{\Modify{\tf a}{\transpshrt{u}}{A}}$, i.e.\ it fires when $t$ can actually be seen as a full section of the transpension type, so that we can apply the action on sections given by $c_\meridsym$.
\begin{remark} \label{rem:close-subst}
This computation rule for \ruleref{transp:elim} is in a non-general context and needs to be forcibly closed under substitution.
We could not find a better way to phrase this computation rule at the time of writing, but while preparing the camera-ready version of this paper, we believe the rule can be mended using the fact that $\bilock{\freshshrt u}{\lollishrt u}$ has a further left adjoint $\bilock{\sumshrt u}{\freshshrt u}$ so that an arbitrary context $\Theta$ can be universally approximated in the image of $\bilock{\freshshrt u}{\lollishrt u}$ as $\paren{\Theta, \bilock{\sumshrt u}{\freshshrt u}, \bilock{\freshshrt u}{\lollishrt u}}$, not unlike how the modal formation and introduction rules of MTT itself were conceived (\cref{fig:mtt:wdra}).
\end{remark}

\section{Recovering Known Operators} \label{sec:recover}
In this section, we explain how to recover the amazing right adjoint $\amaze$ \cite{internal-universes}, BCM's $\Phi$ and $\Psi$ combinators \cite{moulin-param3,moulin}, $\Glue$ \cite{cubical,paramdtt}, $\Weld$ \cite{paramdtt} and $\mill$ \cite{psh-charting-design-space} and (without formal claims) locally fresh names \cite{freshmltt} from the transpension type, the strictness axiom \cite{orton-pitts-axioms} and certain pushouts.
\Cref{fig:recover} gives an overview of the dependencies.

\begin{figure}
	\[
	\hspace*{-1em}\resizebox{\textwidth}{!}{
		\xymatrix{
			&
			{\txt{\textbf{Transpension} \\
			$\lollishrt u \dashv {\transpshrt u}$ \\
			\cite{yetter}, this paper
			}}
				{\ar[ld]}
				{\ar@{-->}[ldd]}
				{\ar[dd]}
				{\ar[rdd]}
				{\ar[rrdd]}
			&&
			{\txt{%
			\textbf{\S\ref{sec:strictness} Strictness} \\
			$\Strict$ \\
			\cite{orton-pitts-axioms}
			}}
				{\ar[d]}
				{\ar[rd]}
				{\ar@/_{1em}/[ldd]}
			&
			\txt{%
			\textbf{Pushout} \\
			\cite[\S 6.3.3]{nuyts-phd}
			}
				{\ar[d]}
			\\
			\txt{
			\textbf{\S\ref{sec:recover:amaze} Amazing r.\ adj.} \\
			$(\IU \to \loch) \dashv (\IU \amaze \loch)$ \\
			\cite{internal-universes} \\
			for HoTT
			}
			&
			&
			&
			\txt{
			\textbf{\S\ref{sec:recover:psh} ${\Glue}$} \\
			\cite{cubical,paramdtt,weaver-licata-dua} \\
			for HoTT/DirTT/param.
			}
			&
			\txt{
			\textbf{\S\ref{sec:recover:psh} ${\Weld}$} \\
			\cite{paramdtt} \\
			for param.
			}
				{\ar[ld]}
			\\
			\txt{
			\textbf{\S\ref{sec:recover:names} Nominal TT} \\
			$\nameshrt{i}, \Angles{i}, {\nu[i]}$ \\
			\cite{freshmltt}
			}
			&
			\txt{
			\textbf{\S\ref{sec:recover:phi} ${\Phi}/{\name{extent}}$} \\
			\cite{moulin-param3} \\
			Relates functions
			}
			&
			\txt{
			\textbf{\S\ref{sec:recover:psi} ${\Psi}/{\Gel}$} \\
			\cite{moulin-param3} \\
			Relates types
			}
				\ar@/^{3em}/[ld]
			&
			\txt{
			\textbf{\S\ref{sec:recover:psh} ${\mill}$} \\
			\cite{psh-charting-design-space} \\
			Swaps $\Weld$ and $\forall\,i$
			}
			\\
			&
			\txt{
			\textbf{\S\ref{sec:recover:transpensive} Transpensivity} \\
			Poor man's $\Phi$/$\name{extent}$
			}
		}
	}
	\]
	\caption{Recovering known operators: Dependency graph}
	\label{fig:recover}
\end{figure}

\subsection{The amazing right adjoint\texorpdfstring{ $\surd$}{}} \label{sec:recover:amaze}
Licata et al.\ \cite{internal-universes} use presheaves over a cartesian base category of cubes and introduce $\surd$ as the right adjoint to the non-dependent exponential $\IX \to \loch$.
We generalize to \emph{semicartesian} base categories (indeed to copointed multipliers) and look for a right adjoint to $\IU \multimap \loch$, which decomposes as substructural quantification after cartesian weakening $\modshade{\lollilong{u : \IU} \circ \wknlong{u : \IU}}$.
Then the right adjoint is obviously $\modshade{\amaze_\IU} := \modshade{\funclong{u : \IU} \circ \transplong{u : \IU}}$.
The type constructor has type $\Modifynovar{\amaze_\IU}{\loch} : (\tymod{\tf r}{\amaze_\IU}{\uni \ell}) \to \uni \ell$ and the transposition rule is as in \cref{thm:transpose}.
This is an improvement in two ways: First, we have introduction, elimination and computation rules, so that we do not need to postulate functoriality of $\amaze_\IU$ and invertibility of transposition.
Secondly, we have no need for a global sections modality $\modshade \flat$.
Instead, we use the modality $\modshade{\amaze_\IU}$ to escape Licata et al.'s no-go theorems.

Our overly general mode theory does contain a global sections modality $\ismod{\flat}{\cdot}{\cdot}$ acting in the empty shape context, and we can use this to recover Licata et al.'s axioms for the amazing right adjoint.
Let us write $\iskeyadj{\spoildropsym_u}{\spoilconstsym_u}{\id}{\id}{\pairshrt{u} \circ \freshshrt{u}}{\lollishrt{u} \circ \wknshrt{u}}$ and $\iskey{\spoilunmeridsym_u}{\funcshrt{u} \circ \transpshrt{u}}{\id}$ for the 2-cells built by partial transposition from either $\keyadj{\hidesym_u}{\spoilsym_u}$ or $\keyshade{\cospoilsym_u}$ (\cref{thm:quantification}), which are each other's transposite. The following are isomorphisms:
\begin{align*}
	\keyshade \kappa := \keyshade{\spoilconstsym_u \whisk \id_{\modshade \flat}} &: \modshade{\flat} \cong \modshade{\lollishrt{u} \circ \wknshrt{u} \circ \flat}, \\
	\keyshade \zeta := \keyshade{\id_{\modshade \flat} \whisk \spoilunmeridsym_u} &: \modshade{\flat \circ \funcshrt{u} \circ \transpshrt{u}} \cong \modshade{\flat}.
\end{align*}
For $\keyshade \kappa$, this is intuitively clear from the fact that we are considering $\IU$-cells in a discrete presheaf produced by $\flat$. For $\keyshade \zeta$, this is similarly clear after taking the left adjoints:
\begin{equation*}
	\lkeyshade{\spoilconstsym_u \whisk \id_{\lmodshade{\shp}}} : \lmodshade{\shp} \cong \lmodshade{\lollishrt{u} \circ \wknshrt{u} \circ \shp}
\end{equation*}
where $\lmodshade{\shp}$, the left name of $\modshade{\flat}$, is semantically the connected components functor which also produces discrete presheaves.
Write $\iskeyadj \eta \eps \shp \flat \id \id$ for the co-unit of the comonad $\modshade \flat$.
We can define
\begin{align*}
	&\IU \amaze \loch : (\tymod{\tf b}{\flat}{\uni{}}) \to \uni{}
	&&\IU \multimap \loch : \uni{} \to \uni{} \\
	&\IU \amazevar{\tf b} A = \Modify{\tf p \twhisk \tf t}{\funcshrt{u} \circ \transpshrt{u}}{A[\key{\zeta\inv}][\bikey{\eta}{\eps}, \bilock{\lollishrt u \circ \wknshrt u}{\funcshrt u \circ \transpshrt u}]}
	&&\IU \multimap A = \Modify{\tf a \twhisk \tf o}{\lollishrt{u} \circ \wknshrt{u}}{A[\bikey{\spoildropsym_u}{\spoilconstsym_u}]}.
\end{align*}
Let two global types ${\cdot} \sep \Gamma, \lock{\tf b}{\flat} \sez A, B \type$ be given.
Applying the non-dependent version of \cref{thm:transpose} to the adjunction $\modshade{\lollishrt{u} \circ \wknshrt{u}} \dashv \modshade{\amaze_\IU} = \modshade{\funcshrt{u} \circ \transpshrt{u}}$ with unit $\keyshade{(1_{\modshade{\funcshrt{u}}} \whisk \reindexsym_u \whisk 1_{\modshade{\wknshrt{u}}})}$ $\iskey{{} \circ \constsym_u}{\id}{\funcshrt{u} \circ \transpshrt{u} \circ \lollishrt{u} \circ \wknshrt{u}}$ yields:
\begin{align*}
	&\paren{A[\bikey{\eta}{\eps}] \to \IU \amazevar{\tf b} B} \\
	&\cong \Modify{\tf p \twhisk \tf t}{\amaze_\IU}{\Modify{\tf a \twhisk \tf o}{\lollishrt{u} \circ \wknshrt{u}}{A[\bikey{\eta}{\eps}][\bikey{\dropsym_u}{\constsym_u}][\bilock{\wknshrt{u}}{\funcshrt{u}}, \bikey{\appsym_u}{\reindexsym_u}, \bilock{\pairshrt{u}}{\wknshrt{u}}]} \to B[\key{\zeta\inv}][\bikey{\eta}{\eps}, \bilock{\lollishrt u \circ \wknshrt u}{\funcshrt u \circ \transpshrt u}]} \\
	&= \Modify{\tf p \twhisk \tf t}{\amaze_\IU}{\Modify{\tf a \twhisk \tf o}{\lollishrt{u} \circ \wknshrt{u}}{A[\bikey{\spoildropsym_u}{\spoilconstsym_u}]}[\key{\zeta\inv}][\bikey{\eta}{\eps}, \bilock{\lollishrt u \circ \wknshrt u}{\funcshrt u \circ \transpshrt u}] \to B[\key{\zeta\inv}][\bikey{\eta}{\eps}, \bilock{\lollishrt u \circ \wknshrt u}{\funcshrt u \circ \transpshrt u}]} \\
	&= \Modify{\tf p \twhisk \tf t}{\amaze_\IU}{\paren{\Modify{\tf a \twhisk \tf o}{\lollishrt{u} \circ \wknshrt{u}}{A[\bikey{\spoildropsym_u}{\spoilconstsym_u}]} \to B}[\key{\zeta\inv}][\bikey{\eta}{\eps}, \bilock{\lollishrt u \circ \wknshrt u}{\funcshrt u \circ \transpshrt u}]} \\
	&= \IU \amazevar{\tf b} \paren{(\IU \multimap A) \to B}.
\end{align*}
Equality of the substitutions applied to $A$ is proven by transposing $\modshade{\lollishrt{u} \circ \wknshrt{u}}$ to the left as $\modshade{\funcshrt{u} \circ \transpshrt{u}}$. Then the unit $\iskey{\constsym_u \circ (\id_{\modshade{\funcshrt{u}}} \whisk \reindexsym_u \whisk \id_{\modshade{\wknshrt{u}}})}{\id}{\funcshrt{u} \circ \transpshrt{u} \circ \lollishrt{u} \circ \wknshrt{u}}$ becomes $\keyshade{\id_{\modshade{\funcshrt{u} \circ \transpshrt{u}}}}$, leaving just $\iskey \eps \flat \id$, whereas $\iskey{\spoilconstsym_u}{\id}{\lollishrt{u} \circ \wknshrt{u}}$ becomes $\iskey{\spoilunmeridsym_u}{\funcshrt{u} \circ \transpshrt{u}}{\id}$ and cancels against $\keyshade{\zeta\inv}$, again leaving just $\keyshade \eps$.

Applying the $\modshade \flat$ modality to both sides of the isomorphism and using $\keyshade \zeta$ to get rid of the amazing right adjoint on the right, yields transposition functions as given by Licata et al.\ \cite{internal-universes}.

We refer back to \cref{sec:car-multip} for the opposite construction: a transpension type for a cartesian multiplier can be constructed from an amazing right adjoint \cite{yetter}.

\subsection{The \texorpdfstring{$\Phi$}{Phi}-combinator} \label{sec:recover:phi}
\begin{figure}%
	\textbf{Binary and destrictified reformulation of the original $\Phi$-rule \cite{moulin-param3,moulin,cavallo-harper-paramhott-journal}:}
	\begin{equation*}
		\inference{
			\Delta, i : \IX \sez B \type \\
			\Delta, i : \IX, y : B \sez C \type \\
			\Delta, y : B[\epsilon/i] \sez c_\epsilon : C[\epsilon/i]
			& (\epsilon \in \accol{0, 1}) \\
			\Delta, h : \lollilong{i : \IX}.B, i : \IX \sez c_\lollisym : C[h\,i/y] \\
			\Delta, h : \lollilong{i : \IX}.B \sez c_\lollisym[\epsilon/i] = c_\epsilon[h\,\epsilon/y] : C[h\,\epsilon/y]
			& (\epsilon \in \accol{0, 1})
		}{
			\Delta \sez \Phi\,c_0\,c_1\,c_\lollisym : \lollishrt{i} . (y : B) \to C \\
			\infwhere \Phi\,c_0\,c_1\,c_\lollisym\,\epsilon\,b = c_\epsilon[b/y]
			& (\epsilon \in \accol{0, 1}) \\
			\infnowhere \Delta, i : \IX \sez \Phi\,c_0\,c_1\,c_\lollishrt{i}\,b = c_\lollisym[\lambda i.b/h]
		}{}
	\end{equation*}
	
	\figskip
	
	\textbf{$\Phi$-rule in \Msys{}} (recall \cref{eq:variance-compose})\textbf{:}
	\begin{equation*}
		\inferencel{phi}{
			\XX, u : \IU \sep \Gamma \sez C \type \\
     		\XX, u : \IU \sep \Gamma, \_ : u \in \partial \IU \sez c_\partial : C \\
			\XX, u : \IU \sep \Gamma,
				\bilock{\lollishrt{u}}{\transpshrt{u}},
				\bilock{\freshshrt{u}}{\lollishrt{u}}
				\sez c_\lollisym : C
				[\bikeytyped{\appsym_u}{\reindexsym_u}{\id}{\id}{\freshshrt{u} \circ \lollishrt{u}}{\transpshrt{u} \circ \lollishrt{u}}] \\
			\XX, u : \IU \sep \Gamma,
				\bilock{\lollishrt{u}}{\transpshrt{u}},
				\bilock{\freshshrt{u}}{\lollishrt{u}},
				\_ : u \in \partial \IU
				\sez %
				c_\lollisym = c_\partial[\bikey{\appsym_u}{\reindexsym_u}] : C
				[\bikey{\appsym_u}{\reindexsym_u}]
		}{
			\XX, u : \IU \sep \Gamma \sez 
			\Phi_u\,c_\partial\,c_\lollisym
			: C \\
			\infwhere \XX, u : \IU \sep \Gamma, \_ : u \in \partial \IU \sez \Phi_u\,c_\partial\,c_\lollisym = c_\partial : C \\
			\infnowhere \XX, u : \IU \sep \Delta, 
				\bilock{\freshshrt{u}}{\lollishrt{u}}
				\sez\Phi_u\,c_\partial\,c_\lollisym = c_\lollisym[\bikeytyped{\constsym_u}{\unmeridsym_u}{\lollishrt{u} \circ \freshshrt{u}}{\lollishrt{u} \circ \transpshrt{u}}{\id}{\id}, \bilock{\freshshrt{u}}{\lollishrt{u}}] : C
		}{}
	\end{equation*}
	\caption{The $\Phi$-rule (sound if $\IU$ is $\top$-slice fully faithful and shard-free).}
	\label{fig:phi}
\end{figure}
In \cref{fig:phi}, we \emph{state} BCM's $\Phi$-rule \cite{moulin-param3,moulin}, also known as $\name{extent}$ \cite{cavallo-harper-paramhott-journal}; both a slight reformulation adapted to \FFsys{} and the rule \ref{rule:phi} adapted to \Msys{}.

In the binary version of the BCM system, or in \FFsys{} with an interval shape as in cubical type theory, the $\Phi$-combinator allows us to define functions of type $\lollishrt{i}. (y : B\,i) \to C\,i\,y$ from an action $c_\epsilon$ at every endpoint $\epsilon$ and a compatible action $c_\lollisym$ on sections $\lollishrt{i}. B\,i$.
When the resulting function $\Phi\,c_0\,c_1\,c_\lollisym$ is applied to an endpoint $\epsilon : \IX$ it just reduces to the corresponding action $\lambda y.c_\epsilon$.
When it is applied to an interval variable $i : \IX$ and expression $b$ that depends only on $i$ and variables fresh for $i$, then the variable $i$ can be captured in $b$ yielding a section $\lambda i.b : \lollishrt{i}. B\,i$ to which we can apply the action on sections $c_\lollisym$.

A few remarks are necessary in order to move to \Msys{}.
First of all, note that the fully applied conclusion of the $\Phi$-rule is $\Delta, i : \IX, y : B \sez \Phi\,c_0\,c_1\,c_\lollishrt{i}\,y : C$.
Of course the endpoints together constitute the boundary of the interval, so if we want to generalize away from cubical type theory, we can expect something like $\Delta, u : \IU, y : B \sez \Phi\,c_\partial\,c_\lollishrt{u}\,y : C$.
We will assume that the shape $\IU$ is $\top$-slice fully faithful and shard-free.
In order to translate the context $(\Delta, u : \IU, y : B)$ to \Msys{}, recall that in \FFsys{} (as well as in the BCM system) the variables to the left of $u$ are fresh for $u$, whereas those to the right need not be.
In \Msys{} we put the shape variables in the shape context, so the context translates to $(\Delta, \bilock{\freshshrt{u}}{\lollishrt{u}}, y : B)$.
In \Msys{} we will treat this entire thing as a single abstract context $\Gamma$, so we do not grant the domain of the $\Phi$-function any special status; in a way all of $\Gamma$ takes the role of $B$.

Then, completely analogously to BCM, we need a type $C$ in context $\Gamma$, an object $c_\partial : C$ whenever $u$ is on the boundary, and a compatible action $c_\lollisym$ that acts on sections of $\Gamma$ and produces sections of $C$, but the quantifier $\modshade{\lollishrt{u}}$ has been brought to the left as $\bilock{\freshshrt{u}}{\lollishrt{u}}$.
Again the resulting term $\Phi_u\,c_\partial\,c_\lollisym$ reduces to $c_\partial$ when on the boundary, and to $c_\lollisym$ when all variables in use are fresh for $u$.
Again, the computation rule for sections is in a non-general context and needs to be forcibly closed under substitution, but see \cref{rem:close-subst}.

Note that the \FFsys{}/BCM substitutions $h\,i/b$ and $h\,\epsilon/b$ (i.e.\ $h\,i/b$ where $i$ is on the boundary) involving applications, all turn into usages of $\bikey{\appsym_u}{\reindexsym_u}$ whereas the variable capturing substitution $\lambda i.b/h$ has turned into $(\bikey{\constsym_u}{\unmeridsym_u}, \bilock{\freshshrt{u}}{\lollishrt{u}}) : (\Delta, \bilock{\freshshrt{u}}{\lollishrt{u}}) \to (\Delta, \bilock{\freshshrt{u}}{\lollishrt{u}}, \bilock{\lollishrt{u}}{\transpshrt{u}}, \bilock{\freshshrt{u}}{\lollishrt{u}})$ which for $\top$-slice fully faithful multipliers is inverse to $(\bilock{\freshshrt{u}}{\lollishrt{u}}, \bikey{\appsym_u}{\reindexsym_u})$ and therefore denotes an inverse of application: variable capture.
Furthermore, regarding the context of $c_\forall$, we remark that for $\top$-slice fully faithful multipliers there is an isomorphism of contexts
\begin{align}
	(\Delta, \bilock{\freshshrt{u}}{\lollishrt{u}}, y : B, \bilock{\lollishrt{u}}{\transpshrt{u}})
	&\cong_{\text{\ref{thm:transpose}}}
	\paren{\Delta, \bilock{\freshshrt{u}}{\lollishrt{u}}, \bilock{\lollishrt{u}}{\transpshrt{u}}, \ctxmod{\tf a}{\lollishrt{u}}{y}{B\brac{\bilock{\freshshrt{u}}{\lollishrt{u}}, \bikey{\appsym_u}{\reindexsym_u}}}} \nonumber \\
	&\cong_{\text{\ref{thm:quantification}}}
	(\Delta, \ctxmod{\tf a}{\lollishrt{u}}{y}{B}, \bilock{\freshshrt{u}}{\lollishrt{u}}) \label{eq:recover:telescope-quantification}
\end{align}
so, recalling from \cref{sec:comparison:embedding:shp} that $i : \IX$ translates to $\bilock{\freshshrt{i}}{\lollishrt{i}}$ in the internal context, there really is a close correspondence to what happens in the BCM system.

We remark that if $\IU$ is $\top$-slice fully faithful but not necessarily shard-free, then the $\Phi$-rule remains valid for creating terms \emph{of a transpension type} $C = \Modify{\tf t}{\transpshrt{u}}{D}$.
Indeed, using pole \cref{thm:pole} we can then define:
\begin{align*}
	\Phi_u\,\pole\,c_\lollisym
	&:= \modify{\tf t}{\transpshrt{u}}{(\unmeridsym_u \modappnovar{\lollishrt{u}} c_\lollisym)} : \Modify{\tf t}{\transpshrt{u}}{D}.
\end{align*}
The non-trivial computation rule follows from the $\eta$-rule for projections (\cref{thm:projmod}) and quantification \cref{thm:quantification}.

\begin{theorem} \label{thm:phi}
	If $\IU$ is $\top$-slice fully faithful and shard-free, then the $\Phi$-rule (\cref{fig:phi}) is sound and indeed derivable from \ref{rule:transp:elim} for all $C$.
\end{theorem}
\begin{proof}
	Use the $\name{case}$-eliminator for $\Modify{\tf t}{\transpshrt{u}}{\name{Unit}}$ (\cref{thm:transp-elim}):
	\begin{align*}
		\Phi_u\,c_\partial\,c_\lollisym
		&:= \case{(\modify{\tf t}{\transpshrt{u}}{\unit})}{
			\pole \mapsto c_\partial \sep
			\merid\,\_ \mapsto c_\lollisym
		}. \tag*{\qedhere}
	\end{align*}
\end{proof}

\subsection{The \texorpdfstring{$\Psi$}{Psi}-type} \label{sec:recover:psi}
\begin{figure}%
	\small
	\textbf{Binary and destrictified reformulation of the original $\Psi$-type \cite{moulin-param3,moulin,cavallo-harper-paramhott-journal}:}
	\begin{equation*}
		\inference{
			\Delta \sez A_\epsilon \type \quad (\epsilon \in \accol{0, 1}) \\
			\Delta, x_0 : A_0, x_1 : A_1 \sez R \type
		}{
			\Delta, i : \IX, \theta : \Theta \sez \Psi_i\,A_0\,A_1\,(x_0.x_1.R) \type \\
			\infwhere \Psi_\epsilon\,A_0\,A_1\,R = A_\epsilon \quad (\epsilon \in \accol{0, 1})
		}{}
	\end{equation*}
	\begin{equation*}
		\inference{
			\Delta \sez a_\epsilon : A_\epsilon \quad (\epsilon \in \accol{0, 1}) \\
			\Delta \sez r : R[a_0/x_0, a_1/x_1]
		}{
			\Delta, i : \IX, \theta : \Theta \sez \inPsisym_i\,a_0\,a_1\,r : \Psi_i\,A_0\,A_1\,(x_0.x_1.R) \\
			\infwhere \inPsisym_\epsilon\,a_0\,a_1\,r = a_\epsilon \quad (\epsilon \in \accol{0, 1}) \\
			\infnowhere \Delta, i : \IX \sez q = \inPsisym_i\,q[0/i]\,q[1/i]\,(\outPsi(j.q[j/i]))
		}{}
		\quad
		\inference{
			\Delta, i : \IX \sez q : \Psi_i\,A_0\,A_1\,(x_0.x_1.R)
		}{
			\Delta \sez \outPsi(i.q) : R[q[0/i]/x_0, q[1/i]/x_1] \\
			\infwhere \outPsi(i.\inPsisym_i\,a_0\,a_1\,r) = r
		}{}
	\end{equation*}
	
	\figskip

	\textbf{$\Psi$-type in \FFsys{}:}
	\begin{equation*}
		\inference{
			\Delta, \theta : \Theta[\epsilon/i] \sez A_\epsilon \type \quad (\epsilon \in \accol{0, 1}) \\
			\Delta, \lollishrt{i}.(\theta : \Theta), x_0 : A_0[(\lambdabar i.\theta)\,0/\theta], x_1 : A_1[(\lambdabar i.\theta)\,1/\theta] \sez R \type
		}{
			\Delta, i : \IX, \theta : \Theta \sez \Psi_i\,A_0\,A_1\,(x_0.x_1.R) \type \\
			\infwhere \Psi_\epsilon\,A_0\,A_1\,R = A_\epsilon \quad (\epsilon \in \accol{0, 1})
		}{}
	\end{equation*}
	\begin{equation*}
		\inference{
			\Delta, \theta : \Theta[\epsilon/i] \sez a_\epsilon : A_\epsilon \quad (\epsilon \in \accol{0, 1}) \\
			\Delta, \lollishrt{i}.(\theta : \Theta) \sez \\
			\qquad r : R[a_0[(\lambdabar i.\theta)\,0/\theta]/x_0, a_1[(\lambdabar i.\theta)\,1/\theta]/x_1]
		}{
			\Delta, i : \IX, \theta : \Theta \sez \inPsisym_i\,a_0\,a_1\,r : \Psi_i\,A_0\,A_1\,(x_0.x_1.R) \\
			\infwhere \inPsisym_\epsilon\,a_0\,a_1\,r = a_\epsilon \quad (\epsilon \in \accol{0, 1}) \\
			\infnowhere q = \inPsisym_i\,q[0/i]\,q[1/i]\,(\outPsi(j.q[j/i, (\lambdabar i.\theta)\,j / \theta]))
		}{}
		\quad
		\inference{
			\Delta, i : \IX \sez q : \Psi_i\,A_0\,A_1\,(x_0.x_1.R)
		}{
			\Delta \sez \outPsi(i.q) : R[q[0/i]/x_0, q[1/i]/x_1] \\
			\infwhere \outPsi(i.\inPsisym_i\,a_0\,a_1\,r) = r
		}{}
	\end{equation*}
	
	\figskip
	
	\textbf{$\Psi$-type in \Msys{}:}
	\begin{equation*}
		\inferencel{psi}{
			\XX, u : \IU \sep \Gamma, \_ : u \in \partial \IU \sez A \type \\
			\XX \sep \Gamma, \hat x : (\_ : u \in \partial \IU) \to A, \bilock{\lollishrt{u}}{\transpshrt{u}} \sez R \type
		}{
			\XX, u : \IU \sep \Gamma \sez \Psitype u {\hat x} A {\tf t} R \type \\
			\infwhere \_ : u \in \partial \IU \sez \Psitype u {\hat x} A {\tf t} R = A
		}{}
	\end{equation*}
	\begin{equation*}
		\inferencel{psi:intro}{
			\XX, u : \IU \sep \Gamma, \_ : u \in \partial \IU \sez a : A \\
			\XX \sep \Gamma, \bilock{\lollishrt{u}}{\transpshrt{u}} \sez r : R[\lambda \_ .a/\hat x, \bilock{\lollishrt{u}}{\transpshrt{u}}]
		}{
			\XX, u : \IU \sep \Gamma \sez \inPsi u a {\tf t} r : \Psitype u {\hat x} A {\tf t} R \\
			\infwhere \_ : u \in \partial \IU \sez \inPsi u a {\tf t} r = a \\
			\infnowhere q = \inPsi{u}{q}{\tf t}{\paren{\outPsi \modappnovar{\lollishrt{u}}{q[\bikey{\appsym_u}{\reindexsym_u}]}}}
		}{}
		\qquad
		\inferencel{psi:elim}{
			\XX, u : \IU \sep \Delta, \bilock{\freshshrt{u}}{\lollishrt{u}} \sez q : \Psitype u {\hat x} A {\tf t} R
		}{
			\XX \sep \Delta \sez \outPsi \modappnovar{\lollishrt{u}} q : R
				[\lambda \_ . q/\hat x, \bilock{\lollishrt{u}}{\transpshrt{u}}]
				[\bikey{\constsym_u}{\unmeridsym_u}] \\
			\infwhere \outPsi \modappnovar{\lollishrt{u}}{\inPsi u a {\tf t} r} = r
				[\bikey{\constsym_u}{\unmeridsym_u}]
		}{}
	\end{equation*}
	\caption[Known operators: The $\Psi$-type]{Typing rules for the $\Psi$-type.}
	\label{fig:psi}
\end{figure}
BCM's $\Psi$-combinator (\cref{fig:psi}, also known as $\name{Gel}$ \cite{cavallo-harper-paramhott-journal}) constructs a line $\lollilong{i : \IX}.\uni{}$ in the universe with endpoints $A_\epsilon$ from a relation $R : A_0 \to A_1 \to \uni{}$.
A section of the $\Psi$-type with endpoints $a_\epsilon$ is a proof of $R\,a_0\,a_1$.
The constructor $\inPsisym$ creates a section $\lollilong{i : \IX}.\Psi_i\,A_0\,A_1\,(x_0.x_1.R)$ from the expected inputs.
The disappearance of $\Theta$ in the premises of $\Psi$ and $\inPsisym$ is entirely analogous to the shape application rule \ruleref{ff:forall:elim} in \cref{fig:ff}.
The eliminator $\outPsi$ extracts from a section of the $\Psi$-type the proof that its endpoints satisfy the relation $R$.

In fact, using the typing rules of \FFsys{} and a strictness axiom as in \cref{sec:strictness} we can already implement a stronger $\Psi$-type, also given in \cref{fig:psi}, where $\Theta$ does not disappear but gets universally quantified.
This is done by strictifying the right hand sides below:\footnote{For the identity type, we use pattern-matching abstractions to abbreviate the usage of the J-rule. We are in an extensional type system anyway.}
\begin{align*}
	\alpha : \Psi_i\,A_0\,A_1\,(x_0.x_1.R) &:\cong (\hat x_0 : (\refl : \idtp{\IX}{i}{0}) \to A_0) \times (\hat x_1 : (\refl : \idtp{\IX}{i}{1}) \to A_1) \times {} \\
	&\qquad \qquad \transpshrt{i} \codot (R[(\lambdabar i.\hat x_0)\,0\,\refl/x_0, (\lambdabar i.\hat x_1)\,1\,\refl/x_1]) \\
	\inPsisym_i\,a_0\,a_1\,r &:= \alpha\inv(\lambda \refl.a_0, \lambda \refl.a_1, \meridshrt{i}\,r) \\
	\outPsi(i.q) &:= \unmeridsym(i.\pi_3(\alpha(q)))
\end{align*}
The fact that this construction is isomorphic to $A_\epsilon$ at endpoint $\epsilon$ follows from our findings about poles in \cref{sec:ff:poles}.

When we move to \Msys{}, once again we translate the context $(\Delta, u : \IU, \Theta)$ to $(\Delta, \bilock{\freshshrt{u}}{\lollishrt{u}}, \Theta)$, which we treat as a single abstract context $\Gamma$.
By a reasoning identical to that in \cref{eq:recover:telescope-quantification}, applying $\bilock{\lollishrt{u}}{\transpshrt{u}}$ only affects the non-fresh part $\Theta$ if the shape $\IU$ is $\top$-slice fully faithful.
This leads to the \Msys{} rules listed in \cref{fig:psi}.
Note again how substitutions $(\lambdabar i.\theta)\,j/\theta$ in \FFsys{} give rise to usages of $\bikey{\appsym_u}{\reindexsym_u}$ in \Msys{}.
The usages of $\bikey{\constsym_u}{\unmeridsym_u}$ are entirely absent in \FFsys{}, but for $\top$-slice fully faithful multipliers, this is an isomorphism anyway.

The eliminator $\outPsi$ only eliminates sections.
For $\top$-slice fully faithful and shard-free multipliers, the $\Phi$-rule provides a pattern-matching eliminator which lets us treat the boundary and section cases separately.

\begin{theorem}
	For \emph{any} multiplier, the $\Psi$-type in \cref{fig:psi} is implementable from the transpension type and the strictness axiom.
\end{theorem}
\begin{proof}
	We strictify the right hand sides below:
	\begin{align*}
		\alpha : \Psitype{u}{\hat x}{A}{\tf t}{R} &:\cong (\hat x : (\_ : u \in \partial U) \to A) \times \Modify{\tf t}{\transpshrt{u}}{R}, \\
		\inPsi u a {\tf t} r &:= \alpha\inv(\lambda \_ . a, \modify{\tf t}{\transpshrt{u}}{r}), \\
		\outPsi \modappnovar{\lollishrt{u}} q &:= \unmeridsym_u \modappnovar{\lollishrt{u}} \pi_2(\alpha(q)).
	\end{align*}
	The fact that this is isomorphic to $A$ on the boundary follows from pole \cref{thm:pole}
\end{proof}
Obviously then the transpension type $\Modify{\tf t}{\transpshrt{u}}{T}$ is also implementable from the $\Psi$-type as $\Psitype{u}{\_}{\Unit}{\tf t}{T}$.

\subsection{Transpensivity} \label{sec:recover:transpensive}
The $\Phi$-rule is extremely powerful but not available in all systems.
However, when the codomain $C$ is a $\Psi$-type, then the $\name{in}\Psi$-rule is actually quite similar to the $\Phi$-rule if we note that sections of the $\Psi$-type are essentially elements of $R$.
As such, we take an interest in types that are very $\Psi$-like.
We have a monad (idempotent if $\IU$ is $\top$-slice fully faithful)
\begin{align*}
	\bar \Psi_u A &:= \Psitype{u}{{\hat x}}{(\_.A)}{\tf t}{\Modify{\tf a}{\lollishrt{u}}{\extension{A}{\sysclause{u \in \partial \IU}{{\hat x}\,\_}}[\bikey{\appsym_u}{\reindexsym_u}]}},
\end{align*}
where $\extension A {\sysclause{\vfi}{a}}$ is the type of elements of $A$ that are equal to $a$ when $\vfi$ holds:
\begin{equation*}
	\extension A {\sysclause{\vfi}{a}} := (x : A) \times ((\_ : \vfi) \to (\idtp A x a)).
\end{equation*}
\begin{definition}
	A type is \textbf{transpensive over $u$} if it is a monad-algebra for $\bar \Psi_u$.
\end{definition}
For $\top$-slice fully faithful and shard-free multipliers, $\Phi$ entails that all types are transpensive.
For other systems, the universe of $u$-transpensive types will still be closed at least semantically under many interesting type formers, allowing to eliminate \emph{to} these types in a $\Phi$-like way.

\subsection{\texorpdfstring{\textsf{Glue}, \textsf{Weld}, \textsf{mill}}{Glue, Weld, mill}} \label{sec:recover:psh}
$\GlueNice{A}{\tysysclause{\vfi}{T}{f}}$ and $\WeldNice{A}{\tysysclause{\vfi}{T}{g}}$ are similar to $\Strict$ but extend unidirectional functions.
Orton and Pitts \cite{orton-pitts-axioms} already show that $\Glue$ \cite{cubical,paramdtt} can be implemented by strictifying a pullback along $A \to (\vfi \to A)$ \cite{psh-charting-design-space} which is definable internally using a $\Sigma$-type.
Dually, $\Weld$ \cite{paramdtt} can be implemented if there is a type former for pushouts along $\vfi \times A \to A$ where $\vfi : \Prop$ \cite[\S 6.3.3]{nuyts-phd}, which is sound in all presheaf categories.

Finally, $\mill$ \cite{psh-charting-design-space} states that $\modshade{\lollilong{u : \IU}}$ preserves $\Weld$ and is provable by higher-dimensional pattern matching (where $\circledast$ is the applicative operation of the modal type):
\begin{align*}
	&\mill : \Modify{\tf a}{\lollishrt{u}}{\WeldNice{A}{\tysysclause{\vfi}{T}{g}}} \\
	&\qquad \to \WeldNice{\Modify{\tf a}{\lollishrt{u}}{A}}{\tysysclause{\Modify{\tf a}{\lollishrt{u}}{\vfi}}{\Modify{\tf a}{\lollishrt{u}}{T}}{(\modify{\tf a}{\lollishrt{u}}{g}) \circledast \loch}} \\
	&\mill\,\hat w = \unmeridsym_u \modappnovar{\lollishrt{u}}{{}} \\
	&\qquad {
			\case{(\appsym_u \modappnovar{\freshshrt{u}}{(\hat w \varbikey{\dropsym_u}{\constsym_u})})}{
				\weld\,a \mapsto \modify{\tf t}{\transpshrt{u}}{(\weld\,(\modify{\tf a'}{\lollishrt{u}}{(a \varbikey{\appsym_u}{\reindexsym_u})}))} \\
				\sysclause{\vfi}{t} \mapsto \modify{\tf t}{\transpshrt{u}}{(\modify{\tf a'}{\lollishrt{u}}{(t \varbikey{\appsym_u}{\reindexsym_u})})}
			}
		}.
\end{align*}

In the first clause, we get an element $a : A$ and can proceed as in \cref{sec:comparison:hdpm}.
In the second clause, we are asserted that $\vfi$ holds (call the witness $p$) so that the left hand $\Weld$-type equals $T$, and we are given $t : T$.
Then inside the meridian constructor $\modifynovar{\transpshrt{u}}{}$ we know that $\Modify{\tf a}{\lollishrt{u}}{\vfi}$ holds as this is proven by $\modify{\tf a'}{\lollishrt{u}}{(p \varbikey{\appsym_u}{\reindexsym_u})}$; hence the $\Weld$-type in the codomain simplifies to $\Modify{\tf a}{\lollishrt{u}}{T}$.
When $\vfi$ holds and $t = \weld\,a = g\,a$, then the $\weld$-constructor of the right hand $\Weld$-type reduces to $(\modify{\tf a}{\lollishrt{u}}{g}) \circledast \loch$ which effectively applies $g$ under the $\modifynovar{\lollishrt{u}}{}$-constructor, so that both clauses match as required.

\subsection{Locally fresh names} \label{sec:recover:names}
Nominal type theory is modelled in the Schanuel topos \cite{nominal-transp} which is a subcategory of nullary affine cubical sets $\Psh(\ary{\cubecat_\affine}{0})$ (\cref{ex:multip:affine-cubes}).
As fibrancy is not considered in this paper, we will work directly in $\Psh(\ary{\cubecat_\affine}{0})$.
Names can be modelled using the multiplier $\loch * (i : \IX)$.
Interestingly, the fresh weakening functor $\freshbase{(i : \IX)}$ is then \emph{inverse} to its left adjoint $\sumbase{(i : \IX)}$.
By consequence, we get $\modshade{\sumshrt{i}} \cong \modshade{\lollishrt{i}} =: \modshade{\nameshrt{i}}$ (the fresh name quantifier) with inverse $\modshade{\freshshrt{i}} \cong \modshade{\transpshrt{i}}$.
For consistency, we will only use $\modshade{\nameshrt{i}}$ and $\modshade{\freshshrt{i}}$, and these will be each other's left names.
The relevant 2-cells are $\iskey{\appsym_i}{\freshshrt{i} \circ \nameshrt{i}}{\id}$ and its inverse $\iskey{\copysym_i}{\id}{\freshshrt{i} \circ \nameshrt{i}}$ (these are each other's left names), as well as $\iskey{\constsym_i}{\id}{\nameshrt{i} \circ \freshshrt{i}}$ and its inverse $\iskey{\dropsym_i}{\nameshrt{i} \circ \freshshrt{i}}{\id}$ (these are each other's left names).

The nominal dependent type system FreshMLTT \cite{freshmltt} used in Pitts's examples of interest \cite{nominal-transp} is substantially different from ours:
\begin{itemize}
	\item It features a name swapping operation that is semantically \emph{not} merely a substitution.
	\item Freshness for a name $i$ is not a modality or a type, but a judgement that can be derived for an expression $t$ if and only if $t$ is invariant under swapping $i$ with a newly introduced name $j$. As a consequence, freshness propagates through type and term constructors.
	\item Many equalities are strict where we can only guarantee an isomorphism.
\end{itemize}
For these reasons, we do not try to formally state that we can support locally fresh names in the sense of FreshMLTT.
Nevertheless, in \cref{fig:freshmltt} we give at least a heuristic $\interp \loch$ for translating programs in a subsystem of FreshMLTT to programs in the current system.
This subsystem does not feature name swapping, but it does feature the \emph{non-binding abstractions} originally defined in terms of it, as well as locally fresh names.
\begin{figure}%
	\footnotesize
	\begin{tabular}{l || c | c}
		name & FreshMLTT & \Msys{}
		\\ \hline \hline
		name quantification
		& $\namesym[i : N].T$
		& $\Modify{\tf a}{\namelong{i : \IX}}{\interp T}$
		\\
		name abstraction
		& $\alpha[i : N].t$
		& $\modify{\tf a}{\namelong{i : \IX}}{\interp t}$
		\\
		name application
		& $t @ i$
		& $\appshrt{i} \modappnovar{\freshshrt{i}} \interp t$
		\\ \hline
		non-binding quant.
		& $\Angles{i : N}.T$
		& $\Modify{\tf f}{\freshshrt{i}}{\Modify{\tf a}{\nameshrt{i}}{\interp T[\bikey{\appshrt{i}}{\copyshrt{i}}]}}$
		\\
		non-binding abs.
		& $\angles{i : N}.t$
		& $\begin{array}{l}
			\copyshrt{i}\,\interp t := \\
			\modify{\tf f}{\freshshrt{i}}{(\modify{\tf a}{\nameshrt{i}}{(\interp t[\bikey{\appshrt{i}}{\copyshrt{i}}])})}
		\end{array}$
		\\ \hline
		locally fresh name
		& $\nu[i : N].t$
		& $\dropshrt{i} \modappnovar{\nameshrt{i}} (\modifynovar{\freshshrt{i}}{\interp t})$
	\end{tabular}
	\caption{A heuristic for translating FreshMLTT \cite{freshmltt} to the current system.}
	\label{fig:freshmltt}
\end{figure}

Ordinary name quantification is simply translated to the modality $\modshade{\nameshrt{i}}$, and as usual application corresponds to the modal projection function (\cref{thm:projmod}).
The non-binding abstraction in FreshMLTT abstracts over a name that is already in scope, \emph{without} shadowing, i.e.\ it is a variable capturing operation.
This is translated essentially to the 2-cell $\iskey{\copyshrt{i}}{\id}{\freshshrt{i} \circ \nameshrt{i}}$, which is inverse to $\keyshade{\appshrt{i}}$ as we have seen earlier for a variable capturing operation (\cref{sec:recover:phi}).
Finally, a locally fresh name abstraction $\nu[i : N].t$ brings a name $i$ into scope in its body $t$, but requires that $t$ be fresh for $i$; in our system we would say that $t$ is a subterm of modality $\modshade{\nameshrt{i} \circ \freshshrt{i}}$.
The type of $\nu[i : N].t$ is the same as the type of $t$, which we can justify with the isomorphism $\keyshade{\dropshrt{i}} : \modshade{\nameshrt{i} \circ \freshshrt{i}} \cong \modshade{\id}$.
This isomorphism is essentially the content of the modal projection function of $\modshade{\freshshrt{i}}$, which we use to translate locally fresh name abstractions.

Note that for general multipliers, $\modshade{\freshshrt{i}}$ does not have an internal left adjoint and hence not a modal projection function either.
For the nullary affine interval, however, $\modshade{\freshshrt{i}} \cong \modshade{\transpshrt{i}}$, so the projection function is essentially $\keyshade{\unmeridsym_i}$!

\begin{example}
Consider Pitts's implementation of higher dimensional pattern matching \cite{nominal-transp}:
\begin{align*}
	&j : (\namesym[i : N].A \uplus B) \to (\namesym[i : N] . A) \uplus (\namesym[i : N] . B) \\
	&j\,\hat c = \nu[i : N] . \case{\hat c @ i}{
		\inl\,a &\mapsto& \inl\,(\angles{i : N}.a) \\
		\inr\,b &\mapsto& \inr\,(\angles{i : N}.b)
	}
\end{align*}
A brainless translation using \cref{fig:freshmltt} yields a type mismatch, because the non-binding abstractions will put the freshness constructor inside the coproduct constructors, whereas the translation of the locally fresh name abstraction mentions it again on the outside.
This is related to our earlier remark that in FreshMLTT, freshness silently propagates through type and term constructors, so here we have to manually intervene.
This is also necessary to insert invertible 2-cell substitutions in some places,
and to check whether we need a modal argument (i.e.\ the modality is handled at the judgemental level) or we need to explicitly use the modal constructor as is always done in \cref{fig:freshmltt}.
Doing so, we find
\begin{align*}
	&j : \Modify{\tf a}{\nameshrt{i}}{A \uplus B} \to \Modify{\tf a}{\nameshrt{i}}{A} \uplus \Modify{\tf a}{\nameshrt{i}}{B} \\
	&j\,\hat c =
		\dropshrt{i} \modappnovar{\nameshrt{i}}{
			\case{(\appshrt{i} \modappnovar{\freshshrt{i}}{(\hat c \varbikey{\dropsym_i}{\constsym_i})})}{
				\inl\,a \mapsto \modify{\tf f}{\freshshrt{i}}{(\inl\,(\modify{\tf a'}{\nameshrt{i}}{(a \varbikey{\appsym_i}{\copysym_i})}))} \\
				\inr\,b \mapsto \modify{\tf f}{\freshshrt{i}}{(\inr\,(\modify{\tf a'}{\nameshrt{i}}{(b \varbikey{\appsym_i}{\copysym_i})}))}
			}
		},
\end{align*}
which is \emph{exactly} what we found in \cref{sec:comparison:hdpm} adapted to our convention that we only want to mention $\modshade{\freshshrt{i}}$ and $\modshade{\lollishrt{i}}$, $\keyshade{\appshrt{i}}$, $\keyshade{\copyshrt{i}}$, $\keyshade{\constshrt{i}}$ and $\keyshade{\dropshrt{i}}$.
\end{example}
\begin{example}
	Consider Pitts et al.'s \emph{implementation} \cite[ex.\ 2.2]{freshmltt} of what is essentially BCM's $\Phi$-rule \cite{moulin,moulin-param3} since the boundary is empty:
	\begin{align*}
		&g : ((\namesym[i : N].A) \to \namesym[i : N].B) \to \namesym[i : N].(A \to B) \\
		&g\,f = \alpha[i : N].(\lambda x.f\,(\angles{i : N}.x)\,@\,i.
	\end{align*}
	We can translate this to the current system using \cref{fig:freshmltt}:
	\begin{align*}
		&g : (\Modify{\tf a}{\nameshrt{i}}{A} \to \Modify{\tf a}{\nameshrt{i}}{B}) \to \Modify{\tf a}{\nameshrt{i}}{A \to B} \\
		&g\,f = \modify{\tf a}{\nameshrt{i}}{\paren{ \lambda x.\appshrt{i} \modappnovar{\freshshrt{i}} (f[\bikey{\dropsym_i}{\constsym_i}]\,(\modify{\tf a'}{\nameshrt{i}}{(x\varbikey{\appsym_i}{\copysym_i}))}) }}.
	\end{align*}
\end{example}
The effect of conflating $\bilock{\sumshrt{i}}{\freshshrt{i}}$ with $\bilock{\lollishrt{i}}{\transpshrt{i}}$ is that affine function application no longer renders the non-fresh part of the context inaccessible using $\bilock{\sumshrt{i}}{\freshshrt{i}}$ but instead universally quantifies it using $\bilock{\lollishrt{i}}{\transpshrt{i}}$, so that we can capture variables as in FreshMLTT.
Remarkably, we do \emph{not} need the $\Phi$-rule (\cref{fig:phi}) for this nor pattern-matching for the transpension type (\cref{fig:transp:elim}), although these rules hold as the affine cubical interval is $\top$-slice fully faithful and shard-free (\cref{ex:multip:affine-cubes}).

\section{Conclusion} \label{sec:discussion}
To summarize, the transpension type can be defined in a broad class of presheaf models and generalizes previous internalization operators.
For now, we only present an extensional type system without an algorithmic typing judgement.
The major hurdles towards producing an intensional version with decidable type-checking, are the following:
\begin{itemize}
	\item We need to decide equality of 2-cells. Solutions may exist in the literature on higher-dimensional rewriting. Alternatively, we need to extend MTT with a language to reason about 2-cell equality \cite{abstract-mtt-lc}.
	\item The substitution modality should ideally reduce like ordinary substitution. \Cref{rem:compute-wknsym} explores what is needed for this to work.
	\item We need a syntax-directed way to close the section computation rules of $\Phi$ (\cref{fig:phi}) and transpension elimination (\cref{sec:structure}) under substitution, but see \cref{rem:close-subst}.
	\item We need to be able to decide whether the boundary predicate, or any similar predicate about shape variables such as $\idtp \IX i 0$ in cubical type theory, is true.
    This problem has been dealt with in special cases, e.g. in implementations of cubical type theory \cite{agda-cubical}.
\end{itemize}
Applications include all applications (discussed in \cref{sec:intro}) of the presheaf internalization operators recovered from the transpension type in \cref{sec:recover}.
Moreover, our modal approach to shape variables via multipliers allows the inclusion of Pinyo and Kraus's twisted prism functor \cite{pinyo-twisted} as a semantics of an interval variable, which we believe is an important advancement towards higher-dimensional directed type theory.

\section*{Acknowledgements}
We thank
	Jean-Philippe Bernardy,
	Lars Birkedal,
	Daniel Gratzer,
	Alex Kavvos,
	Magnus Baunsgaard Kristensen,
	Daniel Licata,
	Rasmus Ejlers M\o{}gelberg
	and
	Andrea Vezzosi
for relevant discussions,
and the anonymous reviewers for their feedback which has been a great guidance in improving the clarity of this paper.

\appendix
\section{Changelog} \label{sec:changelog}
The first preprint of this paper appeared in 2020 and is subsumed in \cite[ch.\ 7]{nuyts-phd}.
Since then, there have been significant changes, primarily terminological ones.
To help out readers coming back to this paper after having consulted earlier versions (or associated presentations), we list the most important changes here.

\subsection{Terminology}
\subsubsection{\Cref{def:multip}}
\begin{itemize}
	\item \textbf{Copointed} multipliers were formerly called \textbf{semicartesian},
	\item Multipliers that are \textbf{comonads} were formerly called \textbf{3/4-cartesian},
	\item \textbf{$\top$-slice faithful} multipliers were formerly called \textbf{cancellative},
	\item \textbf{$\top$-slice full} multipliers were formerly called \textbf{affine},
	\item \textbf{$\top$-slice shard-free} multipliers were formerly called \textbf{connection-free}, and \textbf{shards} were formerly called \textbf{connections},
	\item \textbf{$\top$-slice right adjoint} multipliers were formerly called \textbf{quantifiable}.
\end{itemize}

\subsubsection{\Cref{def:pointable}}
\begin{itemize}
	\item \textbf{Unpointable} objects were formerly called \textbf{spooky},
	\item \textbf{Not objectwise pointable} categories were formerly called \textbf{spooky}.
\end{itemize}

\subsubsection{\Cref{def:act-elements}}
\textbf{Presheafwise} faithful/full/shard-free/right adjoint multipliers were formerly called \textbf{providently} cancellative/affine/connection-free/quantifiable.

A previous version of the paper featured a different definition of shards/connections and presheafwise shard-freedom / provident connection-freedom, that was based on the \emph{indirect} boundary and \emph{indirectly} dimensionally split morphisms (\cref{sec:log:dir-dim-split}).
These notions are obsolete and are retained in the technical report \cite{transpension-techreport} under the names \textbf{indirect shard} and \textbf{indirect shard-freedom} (alongside the direct notions from \cref{def:act-elements}) solely for consistency with \cite[ch.\ 7]{nuyts-phd}.

\subsubsection{\Cref{def:dir-dim-split}} \label{sec:log:dir-dim-split}
A previous version of this paper had a different definition of boundary and of dimensionally split cells which was defined using the pullback of $\Xi \multip \yoneda U \to \yoneda U \supseteq \partial U$.
These notions are obsolete and are retained in the technical report \cite{transpension-techreport} under the name \textbf{indirect boundary} and \textbf{indirectly dimensionally split} cells (alongside the direct notions) solely for consistency with \cite[ch.\ 7]{nuyts-phd}.

As the notions coincide for $\Xi = \top$, \cref{thm:boundary} holds w.r.t.\ the indirect boundary if $\Xi = \top$.
By construction, the indirect boundary is respected by substitution, while the transpension type is generally not if the multiplier is not $\top$-slice fully faithful (\cref{sec:ff:rules:discussion}).
Hence, the indirect notions generally diverge from the direct ones and the indirect version of \cref{thm:boundary} breaks down if $\Xi \neq \top$ and the multiplier is not $\top$-slice fully faithful.

\subsubsection{\Cref{thm:quotient}}
The \textbf{quotient theorem} was formerly called \textbf{kernel theorem}.

\subsection{Ticks}
A previous version of this paper assigned names to locks, called \emph{ticks}, after the tick of the (c)lock in \cite{clocks-ticking}.
The usage of this notation is orthogonal to the introduction of \Msys{}; a version of the original \MTT{} paper \cite{mtt-paper} in tick notation is included in \cite[\S 5.3]{nuyts-phd}.
An improved notation is proposed in \cite{abstract-mtt-lc}.

\subsection{Lockless notation}
A previous version of this paper supplemented the \MTT{} notation with a so-called \emph{lockless} notation in which $(\Gamma, \bilock{\kappa}{\mu})$ was denoted as $\lmodshade{\kappa}(\Gamma)$.
This notation was abolished in favour of \emph{left adjoint reminders} (\cref{sec:mtt:left}).

\bibliographystyle{alphaurl}
\bibliography{lmcs-refs.bib}

\end{document}